\documentclass[twoside,leqno]{article}

\usepackage[letterpaper]{geometry}

\usepackage{siamproceedings}

\usepackage[T1]{fontenc}
\usepackage{amsfonts,amsmath}
\usepackage{graphicx}
\usepackage{epstopdf}
\usepackage{enumitem}
\usepackage{algorithmic}
\usepackage{booktabs}
\usepackage[skins,listings,breakable]{tcolorbox}
\usepackage{listings}

\ifpdf
  \DeclareGraphicsExtensions{.eps,.pdf,.png,.jpg}
\else
  \DeclareGraphicsExtensions{.eps}
\fi

\newsiamremark{remark}{Remark}
\newsiamremark{hypothesis}{Hypothesis}
\crefname{hypothesis}{Hypothesis}{Hypotheses}
\newsiamthm{claim}{Claim}
\newsiamthm{observation}{Observation}
\usepackage{amsopn}

\usepackage{float}

\usetikzlibrary{positioning}
\usetikzlibrary{shapes,matrix}
\usepackage{tikz}
\usepackage{float}
\usetikzlibrary{arrows.meta}
\usetikzlibrary{decorations.pathreplacing,calligraphy}
\usetikzlibrary{decorations.pathmorphing}

\usepackage{mathtools}

\newcommand*{\cE}{\mathcal{E}}
\newcommand*{\cL}{\mathcal{L}}

\newcommand{\Odd}{\mathsf{Odd}}

\usepackage{todonotes}

\DeclareMathOperator{\E}{\mathbf{E}}

\DeclareMathOperator{\R}{\mathbb{R}}
\DeclareMathOperator{\N}{\mathbb{N}}
\DeclareMathOperator{\Z}{\mathbb{Z}}

\newcommand{\cond}{\ensuremath{ ~\Big|~ }} 

\newcommand{\Pro}[1]{\mathbf{Pr} \left[\,#1\,\right]}

\newcommand{\Proco}[2]{\mathbf{Pr} \left[\,#1~\middle\vert~#2\,\right]}
\newcommand{\Ex}[1]{\mathbf{E} \left[\,#1\,\right]}
\newcommand{\Exco}[2]{\mathbf{E} \left[\,#1~\middle\vert~#2\,\right]}
\newcommand{\Exf}[2]{\mathbf{E}^{(#1)}\left[\,#2\,\right]}
\newcommand{\Prof}[2]{\mathbf{Pr}^{(#1)}\left[\,#2\,\right]}

\newcommand{\M}{\mathbf{M}}
\renewcommand{\P}{\mathbf{P}}

\newcommand{\B}{\mathcal{B}}

\newcommand{\xbar}{\ensuremath{\overline{x}}}

\DeclareMathOperator{\discr}{\mathrm{disc}}

\DeclareMathOperator{\1}{\mathbf{1}}

\usepackage{bbm}

\newcommand{\taucont}[1]{{\tau}_{\mathrm{cont}}\left({#1}\right)}
\newcommand{\tauspectral}[1]{{\widetilde{\tau}}_{\mathrm{cont}}\left({#1}\right)}

\newcommand{\taulocal}{\tau_{\mathrm{local}}}

\newcommand{\tauglobal}{\tau_{\mathrm{global}}}

\numberwithin{equation}{section}

\newlength{\leftstackrelawd}
\newlength{\leftstackrelbwd}
\def\leftstackrel#1#2{\settowidth{\leftstackrelawd}
{${{}^{#1}}$}\settowidth{\leftstackrelbwd}{$#2$}
\addtolength{\leftstackrelawd}{-\leftstackrelbwd}
\leavevmode\ifthenelse{\lengthtest{\leftstackrelawd>0pt}}
{\kern-.5\leftstackrelawd}{}\mathrel{\mathop{#2}\limits^{#1}}}

\def\?#1{}
\makeatletter
\def\whp{w.h.p\@ifnextchar-{.}{\@ifnextchar.{.\?}{\@ifnextchar,{.}{\@ifnextchar){.}{\@ifnextchar:{.:\?}{.\ }}}}}}
\makeatother

\let\epsilon\varepsilon

\makeatletter
\newcommand{\DeclareMathActive}[2]{
  \expandafter\edef\csname keep@#1@code\endcsname{\mathchar\the\mathcode`#1 }
  \begingroup\lccode`~=`#1\relax
  \lowercase{\endgroup\def~}{#2}
  \AtBeginDocument{\mathcode`#1="8000}
}

\newcommand{\std}[1]{\csname keep@#1@code\endcsname}
\patchcmd{\newmcodes@}{\mathcode`\-\relax}{\std@minuscode\relax}{}{\ddt}
\AtBeginDocument{\edef\std@minuscode{\the\mathcode`-}}
\makeatother

\allowdisplaybreaks

\makeatletter
\def\Hy@Warning{\@gobble}
\def\Hy@WarningNoLine{\@gobble}
\makeatother

\usepackage[skins,listings,breakable]{tcolorbox}
\usepackage{listings}

\usepackage{etoolbox} 

\makeatletter
\newif\if@gather@prefix 
\preto\place@tag@gather{
  \if@gather@prefix\iftagsleft@ 
    \kern-\gdisplaywidth@ 
    \rlap{\gather@prefix} 
    \kern\gdisplaywidth@ 
  \fi\fi 
} 
\appto\place@tag@gather{ 
  \if@gather@prefix\iftagsleft@\else 
    \kern-\displaywidth 
    \rlap{\gather@prefix} 
    \kern\displaywidth 
  \fi\fi 
  \global\@gather@prefixfalse 
} 
\preto\place@tag{ 
  \if@gather@prefix\iftagsleft@ 
    \kern-\gdisplaywidth@ 
    \rlap{\gather@prefix} 
    \kern\displaywidth@ 
  \fi\fi 
} 
\appto\place@tag{ 
  \if@gather@prefix\iftagsleft@\else 
    \kern-\displaywidth 
    \rlap{\gather@prefix} 
    \kern\displaywidth 
  \fi\fi 
  \global\@gather@prefixfalse 
} 
\newcommand*{\beforetext}[1]{ 
  \ifmeasuring@\else
  \gdef\gather@prefix{#1} 
  \global\@gather@prefixtrue 
  \fi
} 
\makeatother

\newcommand{\token}[3]
{
 \draw[fill=yellow!30,draw=black,rounded corners=2pt] (#1-0.3,#2-0.23) rectangle node[pos=0.5]{\footnotesize{#3}} (#1+0.3,#2+0.23);
}

\newcommand{\ttoken}[3]
{
 \draw[fill=yellow!30,thick,draw=red,rounded corners=2pt] (#1-0.3,#2-0.23) rectangle node[pos=0.5]{\footnotesize{#3}} (#1+0.3,#2+0.23);
}

\newcommand{\newL}{L}

\begin{document}

\title{\Large (Almost) Perfect Discrete Iterative Load Balancing}

\author{
    Petra Berenbrink\thanks{University of Hamburg, Germany (\email{petra.berenbrink@uni-hamburg.de}). 
    }
    \and Robert Els\"asser\thanks{University of Salzburg, Austria (\email{elsa@cs.sbg.ac.at}).}
    \and Tom Friedetzky\thanks{Durham University, U.K. (\email{tom.friedetzky@durham.ac.uk}).}
    \and Hamed Hosseinpour\thanks{University of Hamburg, Germany (\email{hamed.hosseinpour@uni-hamburg.de}).}
    \and Dominik Kaaser\thanks{Technical University of Hamburg, Germany (\email{dominik.kaaser@tuhh.de}).}
    \and Peter Kling\thanks{Darmstadt University of Applied Sciences, Germany (\email{peter.kling@h-da.de}).}
    \and Thomas Sauerwald\thanks{University of Cambridge, U.K. (\email{tms41@cam.ac.uk}).}
}

\date{}

\maketitle

\fancyfoot[C]{\thepage}

\begin{abstract}
We consider discrete, iterative load balancing via matchings on arbitrary
graphs. Initially each node holds a certain number of tokens, defining the load
of the node, and the objective is to redistribute the tokens such that
eventually each node has approximately the same number of tokens. We present
results for a general class of simple local balancing schemes where the
tokens are balanced via matchings. In each round the process averages the
tokens of any two matched nodes. If the sum of their tokens is odd, the node to
receive the one excess token is selected at random. Our class covers three
popular models: in the matching model a new matching is generated randomly in
each round, in the balancing circuit model a fixed sequence of matchings is
applied periodically, and in the asynchronous model the load is balanced over a
randomly chosen edge.

We measure the quality of a load vector by its discrepancy, defined as the
difference between the maximum and minimum load across all nodes. As our main
result we show that with high probability our discrete balancing scheme reaches
a discrepancy of $3$ in a number of rounds which asymptotically matches the spectral bound for continuous load balancing with fractional load.

This result improves and tightens a long line of previous works, by not only
achieving a small constant discrepancy (instead of a non-explicit, large
constant) but also holding for arbitrary instead of regular graphs. The result
also demonstrates that in the general model we consider, discrete load
balancing is no harder than continuous load balancing.
\end{abstract}

\section{Introduction}

In load balancing the goal is to distribute a set of jobs to multiple computing resources. Many large-scale applications use some form of decentralized load balancing, with examples including balancing search engine queries, distributing internet traffic, updating distributed databases, or performing numerical simulations.

In a classical model for parallel and distributed computing, $n$ processors are connected by an arbitrary, undirected graph. Initially each processor (\emph{node}) has a given number of jobs (\emph{tokens}). The objective is to redistribute the tokens by moving them across the edges such that
eventually each node has approximately the same number of tokens, thereby minimizing the \emph{discrepancy}\footnote{Note that this term is the standard definition in the field of load balancing, in contrast to its use in discrepancy theory.}, defined as the difference between the maximum and minimum load across all nodes. Particularly for large-scale networks, it is desirable to employ iterative protocols that are simple and local. This means each node only needs to know its current and neighboring node's loads and, based on this information, determines the number of tokens to be transferred. While
load balancing on graphs is a well-studied and natural area, there are profound connections to other problems including opinion dynamics \cite{DBLP:conf/podc/BerenbrinkCGMMR23}, gossiping \cite{DBLP:journals/tit/BoydGPS06}, information aggregation \cite{DBLP:conf/focs/KempeDG03}, and clustering \cite{DBLP:conf/soda/BecchettiCNPT17}.

Two important classifications of iterative  load balancing schemes are between  \emph{diffusion models} and \emph{matching models} \cite{DBLP:journals/pc/DiekmannFM99}.
In the former model every node concurrently balances its load with \emph{all} neighbors in each round.
In the latter model every node balances its load with \emph{at most one} neighbor in each round.
Another, orthogonal, distinction is between \emph{continuous} and \emph{discrete} load.
In continuous load balancing, load tokens are assumed to be arbitrarily divisible.
Continuous iterative load balancing is essentially equivalent to a Markov chain, and the spectral gap captures the convergence speed and time to reach a constant discrepancy (e.g., \cite{DBLP:conf/focs/RabaniSW98}).

While the continuous setting is well understood, the question remains whether it is a good approximation of the discrete setting, where load is composed of unit-size tokens that are not divisible \cite{DBLP:conf/focs/RabaniSW98}. The deviation between the processes is caused by the accumulation of rounding errors across different nodes and rounds, which makes the discrete process non-linear and hard to analyze.
This stark contrast between the continuous setting and the discrete setting is highlighted in Rabani, Sinclair and Wanka~\cite{DBLP:conf/focs/RabaniSW98}, where they call the discrete setting \emph{``true process''} and the continuous setting \emph{``idealized process''}. In the same work from 1998, the authors also point out that ``the question of a precise quantitative relationship between Markov chains and load-balancing algorithms has been posed by several authors'', notably in \cite{DBLP:journals/jcss/GhoshM96,lovasz1995mixing,DBLP:journals/mst/MuthukrishnanGS98,DBLP:journals/jcss/GhoshM96,DBLP:conf/spaa/SubramanianS94}, and ``seems to be of interest in its own right.'' One concretization of this question, which has been the objective of many previous works in this area \cite{DBLP:conf/focs/RabaniSW98,DBLP:conf/focs/SauerwaldS12,DBLP:conf/stoc/FriedrichS09,DBLP:journals/jcss/BerenbrinkCFFS15}, is as follows:

\paragraph{Open Question:} For a given undirected, connected graph $G=(V,E)$ and an arbitrary initial load vector in $\N_0^n$ with initial discrepancy $K$, let $\tau_{\text{cont}}(K)$ denote the number of rounds needed in the continuous setting to reach discrepancy at most $1$. Find a tight bound on the discrepancy after $O(\tau_{\text{cont}}(K))$ rounds in the discrete model. 

\medskip
In this work we nearly close the gap between the discrete and the continuous setting by bounding the discrepancy by $3$, which holds for any connected graph and a large class of models.
Additionally, our analytical framework is simpler and more flexible.

\paragraph{Our Results.} We focus on the matching model, as it requires less communication per round than the diffusion model; yet it tends to perform better in theory and practice \cite{DBLP:conf/focs/RabaniSW98,DBLP:conf/focs/SauerwaldS12}. We present results for a general class of matching models including the prominent \emph{random matching model} and the \emph{balancing circuit model}. 
In the former model a new matching is generated randomly in each round, while in the latter model a fixed sequence of matchings is applied periodically. We also consider an \emph{``asynchronous''} process in which at each iteration one edge $e \in E$ is chosen uniformly at random.
The very simple load balancing process works as follows.
For each edge in the matching, the two incident nodes average their load as much as possible; in case their total load is odd, the so-called excess token is assigned randomly to either node. This resembles randomized rounding in optimization, which leads to the question of how well the discrete setting can approximate the continuous one with this randomized strategy. 
Note that the algorithm  considered in this paper is  natural and simple, that is, straightforward to implement. Moreover it is    strictly distributed in the sense that only local information based on a node and its neighbours are used. Due to the nature of the communication patterns (matchings) there is no multi-port communication, another nod toward real-world constraints. 

As is standard in the area of discrete load balancing, we use a very strict notion of smoothness of a load vector called \emph{discrepancy}. Let $X^{(t)}$ be the load vector at the end of round $t$, and recall that the discrepancy, denoted by $\discr(X^{(t)})$, is defined as the difference between the maximum and minimum load across all nodes.
In the following theorem we focus on the three specific models described above;
our general result can be found in \cref{thm:main-result}. 
\begin{theorem*}[simplified version of \cref{thm:main-result}]
Let $G$ be any undirected, connected graph with $n$ nodes and consider any initial load vector $x^{(0)}$ with $\discr(x^{(0)})\leq K$ in the discrete setting.
Then, in the random matching, balancing circuit and asynchronous model there exists a time $\tau=O(\tauspectral{K})$, where $\tauspectral{K}$ is the standard spectral bound in the continuous setting (see \cref{eq:balancing_spectral} and \cref{eq:matching_spectral}), such that for an arbitrarily small constant $c>0$,
\[
\Pro{ \discr\left(X^{(\tau)}\right) \leq 3} \ge 1 -\exp \left(-\log^{1-c}(n) \right).
\]
\end{theorem*}

As is commonly done we bound the convergence time in terms of the natural spectral bound $\tauspectral{K}$ on $\taucont{K}$, instead of $\taucont{K}$ directly.
Furthermore, it was shown in \cite[Theorem~2.10]{DBLP:conf/focs/SauerwaldS12}, that this upper bound is tight for the random matching model if $K \geq n^{1+\Omega(1)}$. 
With this slight slack in the number of rounds, our work essentially closes the lid on the long standing problem of consolidating the continuous and the discrete setting. Concretely, we obtain the following quantitative improvements. First, the previously best known result \cite{DBLP:conf/focs/SauerwaldS12} used the same balancing times but only achieved a discrepancy which is a non-explicit (and large) constant.
Second, the results for constant discrepancy in \cite{DBLP:conf/focs/SauerwaldS12} only hold for regular graphs, whereas our theorem applies to regular and non-regular graphs alike.
Third, our result holds with higher probability compared to \cite{DBLP:conf/focs/SauerwaldS12}.

On a very high level, the way we reduce the discrepancy from $K$ to $3$ is as follows. The most involved building block is to show a constant bound on the maximum load for instances with a \emph{linear} number of tokens. Then we show that for instances with an \emph{arbitrary} number of tokens, after $\tau$ rounds the number of tokens above the average is $O(n)$. Hence we can apply the result for a linear number of tokens, which results into a discrepancy of $4$ first, and
using another $\tau$ rounds we finally reduce the discrepancy down to $3$. See \cref{sec:techniques} for more details and
for a more detailed comparison with previous work, see~\cref{tab:related}.
We note that, with slight abuse of notation, in the balancing circuit model we consider $\Delta$ to be length of the sequence of matchings, whereas in the random matching model it is the maximum degree of the given graph.

\newcommand{\rounddet}{det}
\newcommand{\roundrand}{rand}

\begin{table}\centering
\caption{Overview of related results for the balancing circuit model.
The stated results from \cite{DBLP:conf/focs/SauerwaldS12,DBLP:conf/stoc/FriedrichS09}, as well as ours, also hold for the random matching model. $\tauspectral{K}$ is the spectral bound for the continuous setting, and it is defined in \cref{eq:matching_spectral,eq:balancing_spectral}, respectively.
All runtime bounds for randomized rounding hold with probability at least $1-\operatorname{o}(1)$. 
}\label{tab:related}
\medskip
\small
\begin{tabular}{cllcl}
\toprule
\textbf{Reference}&\textbf{Rounds}&\textbf{Discrepancy}&\textbf{Rounding}&\textbf{Graphs}\\
\midrule
\cite{DBLP:conf/focs/RabaniSW98} &
$O\left( \tauspectral{K} \right)$ &
$O\left(\tauspectral{n}\right)$ &  \rounddet & all \\
\cite{DBLP:conf/stoc/FriedrichS09} &
$O\left(\tauspectral{K} \right)$ &
$O\left(\sqrt{\tauspectral{n}}\right)$  & \roundrand & all \\
\cite{DBLP:conf/stoc/FriedrichS09} &
$O\left(\tauspectral{K}\cdot(\log \log(n))^3\right)$ &
$O(1)$  & \roundrand & expander \\
\cite{MS10} 
& $2 \log_2(n)$ &
$16$  & \roundrand & hypercube \\
\cite{MS09} &
 $O(\log (n))$ & $2$  & \roundrand & hypercube \\
\cite{DBLP:conf/focs/SauerwaldS12} &
$O\left(\tauspectral{K}\right)$ &
$O(\log^\epsilon n)$ & \roundrand & all \\
\cite{DBLP:conf/focs/SauerwaldS12} &
$O\left(\tauspectral{K} \cdot \log\log(n)\right)$ &
$O(\log\log(n))$ & \roundrand & all \\
\cite{DBLP:conf/focs/SauerwaldS12} &
$O\left( \tauspectral{K}\right)$ &
$O(1)$ & \roundrand & constant degree, $\Delta=O(1)$ \\
Thm.~\ref{thm:main-result} &
$O\left(\tauspectral{K}\right)$ &
3 & \roundrand & all \\
\bottomrule
\end{tabular}
\end{table}

\subsection{Related Work}

\cite{DBLP:journals/jpdc/Cybenko89,DBLP:journals/concurrency/Boillat90, DBLP:books/daglib/0067089,DBLP:journals/tit/BoydGPS06}
pioneer the use of Markov chains for analyzing diffusion-based load balancing schemes in the continuous model. Note that \cite{DBLP:journals/tit/BoydGPS06} considers the asynchronous process where in each round a single edge is chosen uniformly at random, whereas the other papers consider the synchronous variant.
For these processes it is known that the discrepancy can be reduced from $K$ to $\ell$ within  $O(\tauspectral{(K/\ell)}$ rounds in the diffusion model, the balancing circuit model and the random matching model (see, e.g., \cite{DBLP:conf/focs/RabaniSW98}). 
All these upper bounds are essentially tight, which follows from the connection between the spectral gap of the graph and mixing times of Markov chains \cite{Venkat}.

\medskip

One of the earliest rigorous analyses of the discrete setting is due to \cite{DBLP:journals/mst/MuthukrishnanGS98}.
Their algorithm always computes, for each edge, the flow of load that would occur in the continuous model, and rounds that down to obtain the numbers of tokens to be sent.
For the diffusion model they show that after $O(\tauspectral{K})$ rounds the discrepancy is at most $O(\Delta n/(1-\lambda))$, where $\lambda$ is the second largest eigenvalue of the diffusion matrix.
\cite{DBLP:journals/jcss/GhoshM96} shows similar results for the matching model.
\cite{DBLP:conf/focs/RabaniSW98} presents a more refined analysis based on Markov chains. 
In the same spirit as the above two papers, they assume that an excess token is always kept at the current node. Their technique works for a large class of processes including  matchings and diffusion models.
Among other results, they provide bounds for general graphs showing that the discrepancy is at most $O(\Delta\tauspectral{K})$ after $O(\tauspectral{K})$ rounds in the balancing circuit model as well as in the diffusion model.
In \cite{DBLP:conf/ipps/BerenbrinkFKK19} the authors analyzed the asynchronous process  for the complete graph. They prove that after $O(n \cdot \log (Kn))$ rounds a discrepancy of two is reached. However, to the best of our knowledge, there are no results for arbitrary graphs in the discrete, asynchronous setting.

\cite{DBLP:journals/siamcomp/FriedrichGS12} introduces a quasi-random version which rounds up or down deterministically such that the accumulated rounding errors
on each edge are minimized. For torus graphs they show a constant discrepancy, while for hypercubes the discrepancy is $O(\log^{3/2}(n))$.
\cite{DBLP:journals/dc/AkbariBS16} analyzes a similar framework for load balancing achieving discrepancy $O(d)$ for $d$-regular graphs, but their algorithm requires additional memory to keep track of the decisions in the past.

\medskip

Randomized rounding was first analyzed in \cite{DBLP:conf/stoc/FriedrichS09} with a focus on the matching model.
Their results demonstrate that for many networks randomized rounding yields roughly a square-root improvement of the discrepancy; see \cref{tab:related} for more details about their results for arbitrary and regular graphs.
For expanders a separate and more tailored analysis leads to a constant discrepancy in $O(\tauspectral{K} \cdot(\log \log(n))^3)$ rounds in the random matching model.

Most closely related, and in many ways our direct competitor, is the paper \cite{DBLP:conf/focs/SauerwaldS12}. For \emph{arbitrary} graphs, a discrepancy of $O(\log^{\epsilon} n)$ for an arbitrarily small constant $\epsilon > 0$ in $O(\tauspectral{K})$ rounds is shown  for the random matching model and
the balancing circuit model.
For regular graphs the authors show constant discrepancy for the random matching model and for the balancing circuit model; for the latter their result needs $\Delta=O(1)$. The number of rounds is again $O(\tauspectral{K})$.
This is a non-trivial result;  the intricate analysis spanning some 40 pages results in  large (and non-explicit) constant factors in the discrepancy bound as well as the running time.

\medskip
There are many results in related models.
\cite{MS10,MS09} consider smoothing networks and 
 show that dimension-exchange on hypercubes achieves a discrepancy of $16$ in $2 \log_2 (n)$ rounds and a discrepancy of $2$ in $O(\log(n))$ rounds.\footnote{Note that the dimension-exchange communication scheme can be regarded as one particular instance of the balancing circuit model. When we write $\log$ we mean $\log_e$ unless specified differently.}
Due to the special structure of the network and the matchings, there is no dependence on $K$. 

Another interesting albeit less related class of load balancing processes are inspired by the Rotor-Router model \cite{DBLP:journals/rsa/CooperDFS10,DBLP:journals/cpc/CooperS06}, where each node evenly distributes tokens to its neighbors (usually in a deterministic way). 
\cite{DBLP:journals/talg/BerenbrinkKKMU19} considers rotor router inspired diffusion-type algorithms. Their results only hold for $d$-regular graphs and their best algorithm achieves  a discrepancy of $d$.

A different direction is dynamic load balancing, where in addition to the balancing, at each round one new token is added to a random node \cite{DBLP:conf/icalp/BerenbrinkHHKR23,DBLP:journals/algorithmica/AlistarhNS22}.
\cite{DBLP:journals/algorithmica/AlistarhNS22} considers such a process on cycles, while in \cite{DBLP:conf/icalp/BerenbrinkHHKR23} the authors recently presented results for arbitrary graphs and wider class of processes. However, due to the addition of new tokens, the proven discrepancy bounds all diverge in $n$. The same holds for so-called ``selfish'' load balancing models \cite{DBLP:journals/talg/BerenbrinkHS14} where tokens act selfishly when deciding whether to migrate to a neighboring node. 
There are also several works on models with dynamic load generation \emph{and} consumption
\cite{DBLP:journals/siamcomp/AnagnostopoulosKU05,DBLP:conf/icalp/BerenbrinkFM05,DBLP:journals/siamcomp/BerenbrinkFG03}, resulting in discrepancy bounds which are either not independent of $n$, or showing only stability.

\paragraph{Road Map.} The remainder of this paper is organized as follows. In \Cref{sec:model} we describe our process and models formally, list some preliminary results and state our main result more formally. In \Cref{sec:techniques} we give a detailed overview of our proof techniques. We summarize our main results and point to some open problems in \Cref{sec:conclusions}. The remaining sections contain the formal analysis. In \cref{sec:auxiliary-results}, we provide some negatice association and concentration results for load vectors. These results enable us to reduce the task of balancing any load vector to balancing one with at most $O(n)$ tokens. \Cref{sec:spreading} contains the backbone of our analysis, in which we analyze load vectors that have at most $O(n)$ tokens. In \Cref{sec:arbitrary} we use the results from \Cref{sec:spreading} to show results for arbitrary load vectors. Equipped with the results in \cref{sec:arbitrary,sec:spreading}, it is relatively straightforward to establish our main results, in particular~\Cref{thm:main-result}, which is done in \Cref{sec:proofsmainresults}. Finally, in the appendix we present an assortment of elementary and well-known tools applied in our analysis.

\section{Model and Main Results}\label{sec:model}

We are given an arbitrary connected and undirected graph $G = (V, E)$ with $n$ nodes.
Initially a set $\mathcal{T}$ of unit-sized \emph{tokens} are distributed arbitrarily among the nodes.
The initial load vector is denoted $x^{(0)}$, and the load vector at (the end of) round $t$ is denoted by $X^{(t)}$. These vectors are row-vectors and the $i$-th entry represents the (integral) load of node~$i$, i.e., the number of tokens on node $i$.
Note that due to the inherent randomization in the process, $X_i^{(t)}$ for $t > 0$ is a random variable.
We will use uppercase letters for random variables and matrices, but lowercase letters for fixed outcomes. We define $\overline{x}:= \sum_{i\in V} x_i^{(0)}/n$ as the average load.
For an $n$-dimensional vector $X$ we define the discrepancy of $X$ as $\discr(X) := \max_{i\in [n]} X_i - \min_{j\in [n]} X_j$. We usually assume that the tokens on the nodes are ordered; the height of the token is its number in that order.  

\subsection{Process Definition}
In this section we first define the standard discrete balancing process which does not specify how tokens are exchanged across the matching edges (similar to \cite{DBLP:conf/focs/RabaniSW98, DBLP:conf/focs/SauerwaldS12}). After that we define a so-called height-sensitive variant of the process, which also specifies the movements of individual tokens. However, both processes generate, at any point of time, exactly the same load distribution.  

For both processes we are given a \emph{sequence of matchings} $\left(\M^{(s)}\right)_{s=1}^{\infty}:= ( \M^{(1)},\M^{(2)},\ldots )$.
The standard load balancing process updates the (discrete) load vector iteratively as follows.
\begin{tcblisting}{breakable,listing only,
  listing options={%
     mathescape,%
     tabsize=4,%
     numbers=left,%
     numberstyle=\tiny,%
     numberblanklines=false,%
     basicstyle=\rmfamily,%
     columns=fullflexible,%
     emph={%
         [1]if, each, else, then, and, or, for, do, while, return, exit, not, output, initialize %
         },
     emphstyle={%
         [1]\bfseries%
         },%
},
size=fbox,boxrule=0pt,frame hidden,arc=0pt,colback=black!10}
for each round $t=1,2,\ldots $ do
    for each edge $\{u,v\} \in \M^{(t)} $ do
        $\displaystyle \left(X^{(t)}_u,X^{(t)}_v\right) \gets \begin{cases} \left(\left\lceil\frac{X^{(t-1)}_u+X^{(t-1)}_v}{2}\right\rceil, \left\lfloor\frac{X^{(t-1)}_u+X^{(t-1)}_v}{2}\right\rfloor\right) & \text{with probability } 1/2, \\[1ex]     \left(\left\lfloor\frac{X^{(t-1)}_u+X^{(t-1)}_v}{2}\right\rfloor, \left\lceil\frac{X^{(t-1)}_u+X^{(t-1)}_v}{2}\right\rceil\right) & \text{with probability } 1/2 .    \end{cases}$
\end{tcblisting}

We assume the tokens in $\mathcal{T}$ are numbered from $1$ to $|\mathcal{T}|$.
In each round $t$, each token $i$ has a \emph{location} $W_i^{(t)} \in V$ 
and a \emph{height} $H_i^{(t)} \in \{1, \dots, X^{(t)}_{W_i^{(t)}}\}$. Initially tokens are ordered arbitrarily on each node, and the initial height of a token is its position in that order. 
We now define the height-sensitive process which is a realization (and refinement) of the load balancing process above (an illustration can be found in~\cref{fig:heightDescription}). The process uses the notion of siblings: two tokens with the same height on two matched nodes are called siblings.

\medskip

\noindent%
\colorbox{black!10}{\minipage{\textwidth-2\fboxsep}%
\textbf{Height-Sensitive Process} (round $t$, matching edge $\{u,v\} \in E$ with $X_u^{(t-1)} \geq X_v^{(t-1)}$)
\smallskip

1. \emph{Moving step:} Move the top $\lceil (X_u^{(t-1)}-X_v^{(t-1)})/2\rceil$ tokens from node $u$ to node $v$, preserving their relative order and adjusting their height accordingly.
\smallskip

2. \emph{Shuffling step:}  Swap each token on node $v$ with its sibling on $u$ with probability $1/2$, independently from all other tokens. In case where the topmost token at $v$ has no sibling, this token is also moved to $u$ with probability $1/2$.
\endminipage}

\medskip

\begin{figure}[ht]
\begin{center}
\begin{tikzpicture}[scale=1.0]
\tikzstyle{knoten}=[rectangle,rounded corners=3pt,draw=black,fill=white,scale=1.0]
\tikzstyle{bknoten}=[rectangle,rounded corners=3pt,draw=black,fill=blue!60,scale=1.0]

\begin{scope}[yshift=-0.05cm]
\draw[-stealth] (-3.5,0) to node[pos=1.0,above]{height} (-3.5,3.5);
\draw (-3.6,0.05) to node[pos=0,left=-1pt]{$1$} (-3.4,0.05);
\draw (-3.6,0.55) to node[pos=0,left=-1pt]{$2$} (-3.4,0.55);
\draw (-3.6,1.05) to node[pos=0,left=-1pt]{$3$} (-3.4,1.05);
\draw (-3.6,1.55) to node[pos=0,left=-1pt]{$4$} (-3.4,1.55);
\draw (-3.6,2.05) to node[pos=0,left=-1pt]{$5$} (-3.4,2.05);
\draw (-3.6,2.55) to node[pos=0,left=-1pt]{$6$} (-3.4,2.55);
\draw (-3.6,3.05) to node[pos=0,left=-1pt]{$7$} (-3.4,3.05);

\draw[gray,dotted,thin] (-3.5,0.05) to (5.6,0.05);
\draw[gray,dotted,thin] (-3.5,0.55) to (5.6,0.55);
\draw[gray,dotted,thin] (-3.5,1.05) to (5.6,1.05);
\draw[gray,dotted,thin] (-3.5,1.55) to (5.6,1.55);
\draw[gray,dotted,thin] (-3.5,2.05) to (5.6,2.05);
\draw[gray,dotted,thin] (-3.5,2.55) to (5.6,2.55);
\draw[gray,dotted,thin] (-3.5,3.05) to (5.6,3.05);
\draw[gray,dotted,thin] (-3.5,3.55) to (5.6,3.55);

\end{scope}

\node[knoten] (0) at (-2.6,3) {$7$};
\node[knoten] (0) at (-2.6,2.5) {$6$};
\node[knoten] (0) at (-2.6,2) {$5$};
\node[knoten] (0) at (-2.6,1.5) {$4$};
\node[knoten] (0) at (-2.6,1) {$3$};
\node[knoten] (0) at (-2.6,0.5) {$2$};
\node[knoten] (0) at (-2.6,0) {$1$};
\draw[fill=black,opacity=1] (-2.6,-0.4) circle (0.15cm);
\node[] () at (-2.6,-0.7) {$u$};

\draw[-stealth] (-2.3,2.55) to [bend right=-20] node[pos=0.5,above]{} (-1.9,2.55);
\draw[-stealth] (-2.3,3.05) to [bend right=-20] node[pos=0.5,above]{} (-1.9,3.05);
\draw[-stealth] (-2.3,2.05) to [bend right=-20] node[pos=0.5,above]{} (-1.9,2.05);

\node[knoten] (0) at (-1.6,0.5) {$9$};
\node[knoten] (0) at (-1.6,0) {$8$};
\draw[fill=black,opacity=1] (-1.6,-0.4) circle (0.15cm);
\node[] () at (-1.6,-0.7) {$v$};

\node[] () at (-2.1,-1.5) {$(a)$};

\begin{scope}[xshift=0.6cm]

\node[knoten] (0) at (0,1.5) {$4$};
\node[knoten] (0) at (0,1) {$3$};
\node[knoten] (0) at (0,0.5) {$2$};
\node[knoten] (0) at (0,0) {$1$};
\draw[fill=black,opacity=1] (0,-0.4) circle (0.15cm);
\node[] () at (0,-0.7) {$u$};

\node[knoten] (0) at (1.0,2) {$7$};
\node[knoten] (0) at (1.0,1.5) {$6$};
\node[knoten] (0) at (1.0,1.0) {$5$};
\node[knoten] (0) at (1.0,0.5) {$9$};
\node[knoten] (0) at (1.0,0) {$8$};
\draw[fill=black,opacity=1] (1.0,-0.4) circle (0.15cm);
\node[] () at (1.0,-0.7) {$v$};

\draw[dashed,stealth-stealth] (0.3,0.05) to [bend right=-20] node[pos=0.5,above]{$?$} (0.7,0.05);
\draw[dashed,stealth-stealth] (0.3,0.55) to [bend right=-20] node[pos=0.5,above]{$?$} (0.7,0.55);
\draw[dashed,stealth-stealth] (0.3,1.05) to [bend right=-20] node[pos=0.5,above]{$?$} (0.7,1.05);
\draw[dashed,stealth-stealth] (0.3,1.55) to [bend right=-20] node[pos=0.5,above]{$?$} (0.7,1.55);
\draw[dashed,stealth-stealth] (0.3,2.05) to [bend right=-20] node[pos=0.5,above]{$?$} (0.7,2.05);

\draw[-stealth,thick] (-2.25,-1.5) to [bend right=0] node[pos=0.5,above] {Moving Step} (0.25,-1.5);

\node[] () at (0.625,-1.5) {$(b)$};

\end{scope}

\begin{scope}[xshift=4cm]

\node[knoten] (0) at (0,1.5) {$6$};
\node[knoten] (0) at (0,1.0) {$3$};
\node[knoten] (0) at (0,0.5) {$2$};
\node[knoten] (0) at (0,0) {$8$};
\draw[fill=black,opacity=1] (0,-0.4) circle (0.15cm);
\node[] () at (0,-0.7) {$u$};

\node[knoten] (0) at (1,2) {$7$};
\node[knoten] (0) at (1,1.5) {$4$};
\node[knoten] (0) at (1,1) {$5$};
\node[knoten] (0) at (1,0.5) {$9$};
\node[knoten] (0) at (1,0) {$1$};
\draw[fill=black,opacity=1] (1,-0.4) circle (0.15cm);
\node[] () at (1,-0.7) {$v$};

\draw[-stealth,thick] (-2.35,-1.5) to [bend right=0] node[pos=0.5,above] {Shuffling Step} (0.1,-1.5);

\node[] () at (0.5,-1.5) {$(c)$};

\end{scope}

\end{tikzpicture}
~\vspace{-1.5em}
\end{center}
\caption{Illustration of the height-sensitive process and the effect on the tokens, which are labelled from $1$ to $9$. $(a)$: The configuration of tokens before the averaging; $(b)$: the configuration after the moving step and $(c)$: the configuration after the shuffling step.
}\label{fig:heightDescription}
\end{figure}
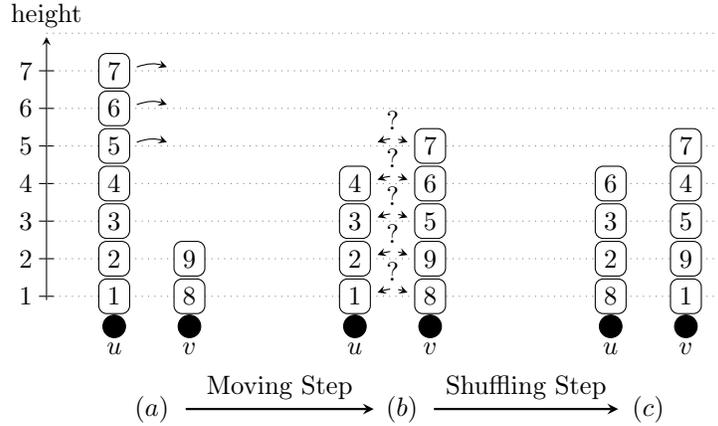

\smallskip

In the following we will refer to  the topmost token as \emph{excess token} if it does not have a sibling.
It is easy to verify that the height of a token can never increase (see \cref{lem:OneTokenRW}, $(i)$ in the appendix).
Furthermore, \emph{any individual} token performs a random walk with the sequence of matching matrices as transition matrices (\cref{lem:OneTokenRW}, $(ii)$).
Crucially, we will prove later that the movements of \emph{different} tokens satisfy a negative association property (see~\cref{lem:height}).

\subsection{Properties of Matchings}
To state the next definitions and results more formally, we adopt the notation from \cite{DBLP:conf/focs/RabaniSW98}. Assume that nodes are labelled from $1$ to $n$. We define the balancing matrix $\M^{(t)} \in [0,1]^{n\times n}$ that represents a matching of $G$ in round $t$ as
$\M_{u,v}^{(t)}:=1/2$ if $u \neq v$ are matched in round $t$, $\M_{u,v}^{(t)}:=0$ if $u \neq v$ are not matched in round $t$, and $\M_{u,u}^{(t)}:=1$ if $u$ is not matched in round $t$. 
 Following \cite{DBLP:conf/focs/SauerwaldS12,DBLP:conf/focs/RabaniSW98}, we use the notation $[u:v]\in \M^{(t)}$ to indicate that nodes $u$ and $v$ with $u<v$ are matched in round $t$.
 Further, we define the product of all matching matrices from time $t_1$ to $t_2$ as
$\M^{[t_1,t_2]}:= \prod_{s=t_1}^{t_2}\M^{(s)}$ (if $t_1 > t_2$ then $\M^{[t_1,t_2]}$ is the identity matrix). For simplicity, we use the same symbol $\M^{(t)}$ for both the matching of round $t$ and the corresponding balancing matrix, sometimes referring to both just as ``matching''. Also, we will always use upper case $\M$ for matchings. Later we also use $\M^{[t_1+1,t_2]}_{u,D} : = \sum_{v\in D} \M^{[t_1+1,t_2]}_{u,v}$. Finally, $\M^{[t_1,t_2]}_{u,.}$ denotes the row of the matrix $\M^{[t_1,t_2]}$ corresponding to node $u$. 
By $\vec{1}$ we denote the row vector of length $n$ in which each entry is $1$.

\begin{definition}[{\cite[Definition 2.1]{DBLP:conf/focs/SauerwaldS12}}]\label{def:smoothing}
A fixed sequence of matchings $\left(\M^{(s)}\right)_{s=1}^{t}$, $t \geq 1$ is called $(K,\epsilon)$-smoothing if for any $x^{(0)}\in\R^{n}$ with $\discr(x^{(0)})\le K$
we have $\discr(x^{(0)}\cdot \M^{[1,t]})\le \epsilon$.
\end{definition}

The expression $x^{(0)}\cdot \M^{[1,t]}$ equals the load vector of the continuous load balancing process with initial load $x^{(0)}$ applying the matchings $\left(\M^{(s)}\right)_{s=1}^{t}$ (see, e.g., \cite{DBLP:conf/focs/RabaniSW98}).
Hence the definition above states that the matching sequence is sufficient to balance the load up to $\epsilon$ in the continuous model. 

In general the matching sequence can be generated deterministically or randomly, and our main result will cover both cases. 
The next definition states the properties of the matchings which we require for our main result.

\begin{definition}\label{def:taus}
A sequence of matchings $\left(\M^{(s)}\right)_{s=1}^{\infty}$ is called $(\tauglobal,\taulocal)$-\emph{good} if
\begin{align}
\inf_{t \in \N_{0}} \Pro{ \bigcap_{u \in V} \left\| \M_{u,\cdot}^{[t+1,t+\tauglobal]} - \vec{ \frac{1}{n}} \right\|_2^2 \leq \frac{1}{n^7}} &\ge 1-\frac{1}{n^3}
\label{eq:taus:eq-1}, \\
\inf_{t \in \N_{0}} \min_{u \in V} \Pro{ \left\| \M_{u,\cdot}^{[t+1,t+\taulocal]} \right\|_2^2 \leq \frac{1}{\log^{10}(n)}} &\ge 1- \frac{1}{\log^{11}(n)}.\label{eq:taus:eq-2}
\end{align}
\end{definition}
Since each token is performing a random walk according to the matching sequence (\cref{lem:OneTokenRW}), we can interpret these events in terms of a distribution of a token performing a time-inhomogeneous random walk and its $\ell_2$-distance to the stationary (i.e, uniform) distribution. The event in \cref{eq:taus:eq-1} means that the distribution of any token will be very close to uniform after $\tauglobal$ rounds.  
This basically corresponds to a complete, i.e., ``global'' mixing of a random walk, since it holds for \emph{any} start node $u$.
With regards to the event in \cref{eq:taus:eq-2}), it only requires a very coarse mixing of the distribution after $\taulocal$ rounds, and this condition holds only ``locally'', i.e., from a specific node $u$. 
Using standard spectral techniques, for the random matching model we have $\tauglobal = O( \log(n)/(1-\lambda))$ and $\taulocal=O(\log \log(n)/(1-\lambda))$ (assuming
 $p_{\min}=\Omega(1/\Delta)$).

\subsection{Main Results}\label{sec:mainresintro}

Our main result is the following theorem; the proof can be found in \cref{sec:proofsmainresults}. \smallskip

\noindent\colorbox{black!10}{\minipage{\textwidth-2\fboxsep}
\begin{theorem}
\label{thm:main-result}%
Let $G$ be any undirected, connected graph on $n$ nodes and consider any initial load vector $x^{(0)} \in \N_0^n$ with $\discr(x^{(0)})\le K$.
Assume our process balances the tokens via a $(\tauglobal,\taulocal)$-good sequence of matchings $\left(\M^{(s)}\right)_{s=1}^{\infty}$.
Then there exists a time $\tau$ with
\[
\tau = O\left(\frac{\log (Kn)}{\log(n)}\cdot \tauglobal+\frac{\log(n)}{\log\log(n)}\cdot \taulocal\right)
\]
such that
\[
\Pro{ \discr\left(X^{(\tau)}\right) \leq 4} \geq 1 -\exp \left(- (1/200) \cdot \frac{\log(n)}{\log \log (n)} \right),
\] 
and for any constant $c > 0$,
\[
\Pro{ \discr\left(X^{(\tau)}\right) \leq 3}\geq 
1 - \exp\left(- \log^{1-c}(n) \right).
\] 
\end{theorem}\endminipage}
~

\medskip

\noindent Note that \cref{thm:main-result} does not make any assumption regarding the provenance of the matching sequences; all we use is the abstract property of $(\tauglobal,\taulocal)$-goodness. 
We now provide details for three explicit ways to create such sequences, namely the \emph{balancing circuit model}, the \emph{random matching model} and the \emph{asynchronous model} (a.k.a.~single edge model). To state the result for these models we need some definitions.
For any $n$ by $n$ real symmetric matrix $\M$, let $\lambda_1(\M)\ge \lambda_2(\M)\ge\cdots \ge \lambda_n(\M)$ be the $n$ eigenvalues of $\M$. For simplicity, let $\lambda(\M) := \max\{|\lambda_2(\M)|,|\lambda_n(\M)|\}$. For a non-symmetric matrix~$\M$, we define the  symmetric matrix as $\widetilde{\M}=\M\cdot \M^T$ and let $\lambda(\M):=\max\{|\lambda_2(\widetilde{\M})|,|\lambda_n(\widetilde{\M})|\}$.

\paragraph{Application to Balancing Circuits.}
In the balancing circuit model all or a subset of the edges of $G$ are covered using a periodic sequence of $\Delta$ fixed matchings $\M^{(1)},\M^{(2)},\ldots,\M^{(\Delta)}$. 
The sequence $\left(\M^{(s)}\right)_{s=1}^{\infty}$ is chosen deterministically and periodically such that $\M^{(s)}=\M^{((s-1)\bmod\Delta+1)}$. Such matchings can be found, e.g., via edge-coloring, and there exist many efficient distributed algorithms that compute such a coloring \cite{DBLP:journals/jcss/GhoshM96}. In the balancing circuit model $\Delta$ is (an upper bound on) the maximum degree of the graph induced by the union of all matchings.
In this model, \Cref{thm:main-result} provides, see \Cref{cor:BC},
\begin{align}
\tau&= O(\tauspectral{K}), \qquad \mbox{ where } \qquad \tauspectral{K} := \Theta \left( \frac{\Delta \cdot \log(Kn)}{1-\lambda\left(\vphantom{X^1}\smash{\M^{[1,\Delta]}}\right)} \right). \label{eq:balancing_spectral}
\end{align}
\paragraph{Application to Random Matchings.}
In line with previous work \cite{DBLP:journals/jcss/GhoshM96,DBLP:journals/tit/BoydGPS06,DBLP:conf/focs/SauerwaldS12}, we consider the following class of randomized algorithms that generate a sequence of (random) matchings $\left(\M^{(s)}\right)_{s=1}^{\infty}$, and simply refer to this as the ``random matching model''. First, we require that the matchings are mutually independent across all rounds. Second, for all edges $\{u,v\}$ and all rounds $t \geq 1$, the probability that $\{u,v\}$ is included in $\M^{(t)}$ is at least $p_{\min}:=c/\Delta$ for some time-independent value $c>0$ (however, $c$ is allowed to be a function of $n$). Recall that we use $\Delta$ as the number of matchings for the balancing circuit model and as maximum degree of $G$ in the random matching model.
Note that the decisions whether or not to include two edges into a matching within the same round are clearly \emph{not} independent.
Some concrete distributed algorithms that satisfy both conditions are described in \cite{DBLP:journals/jcss/GhoshM96,DBLP:journals/tit/BoydGPS06}. 
To state the results for the random matching model we 
define the diffusion matrix $\P$  as $\P_{u,v}:=1/(2\Delta)$ if $(u,v)\in E$, $\P_{u,v}:=1-\operatorname{deg}(u)/(2\Delta)$ if $u=v$, and $\P_{u,v}:=0$ otherwise.
In this case, \Cref{thm:main-result} provides, see \Cref{cor:RM},
\begin{align}
 \tau &= O(\tauspectral{K}), \qquad \mbox{ where } \qquad \tauspectral{K} :=
 \Theta \left(\frac{\log(Kn)}{p_{\min}\cdot \Delta \cdot\left(1-\lambda\left(\vphantom{X^1}\P\right)\right)}\right)  \label{eq:matching_spectral}
 \end{align}
 Note that in case $p_{\min}=\Omega(1/\Delta)$ we have 
$\tauspectral{K} =\Theta \left( \log(Kn)/(1-\lambda(\P)) \right)$.

\paragraph{Application to the Asynchronous (Single Edge) Model.} 
 In this model, at each round we pick a single edge $e$ uniformly at random. This is a special case of the random matching model with $p_{\min} = 1/|E|$; and thus, the same of definition of $\tauspectral{K}$ applies here. Therefore, \Cref{thm:main-result} provides $\tau =O\left(\frac{|E|}{\Delta}\cdot\log(Kn)/\left(1-\lambda\left(\vphantom{X^1}\P\right)\right)\right)
 =O\left( n \cdot \frac{d}{\Delta} \cdot \log(Kn)/\left(1-\lambda\left(\vphantom{X^1}\P\right)\right)\right)
 $, where $d=2|E|/n$ is the average degree of $G$ (see \Cref{cor:RM}).

\subsection{Comparison of our Results and Techniques with Previous Work}

An important novelty of our approach is that our analysis framework seamlessly covers the balancing circuit model, the random matching model and the asynchronous (single edge) model. This is in contrast to previous works \cite{DBLP:conf/focs/RabaniSW98,DBLP:journals/mst/MuthukrishnanGS98,DBLP:conf/focs/SauerwaldS12,DBLP:conf/stoc/FriedrichS09}, which either focus only on one specific matching model or provide tailored analyses for each of them. To unify the three models covered here, our analysis is based on a coarse (and local) and a fine (and global) mixing/balancing property of these models.

In comparison with \cite{DBLP:conf/focs/SauerwaldS12}, the constant factors in both our running time and discrepancy are explicit and small. In addition to our analysis being tighter and simpler, we also obtain a larger success probability. We note that the tool box we use differs substantially from that of \cite{DBLP:conf/focs/SauerwaldS12}, not least because the pièce de résistance of our analysis, getting from a large constant bound for the discrepancy down to a very small constant is not being done in \cite{DBLP:conf/focs/SauerwaldS12} at all, and is one of the hardest problems we solve.
Finally, the results of \cite{DBLP:conf/focs/SauerwaldS12} require regular graphs in the random matching, and a constant $\Delta$ for the balancing circuit model. Our results are covering both of these models (and additionally the asynchronous model), and hold without any of these restrictions.

While our proof method yields more general and tighter results, we also believe that our proof is more intuitive and direct. The first and more straightforward step is to carefully bound the number of tokens in each possible subset $S \subseteq V$ and then apply a union bound. This is based on a new Hoeffding-type concentration bound, which generalizes and tightens previous bounds \cite{DBLP:conf/stoc/FriedrichS09, DBLP:conf/focs/SauerwaldS12} in that it works for \emph{arbitrary} convex combinations of the load vector. This leaves us with only $O(n)$ tokens to balance.

The second part of the analysis makes uses of our new height-sensitive process, which constraints the movement of tokens in such a way that their heights are non-increasing. Despite this restriction, we can prove that their movements are negatively associated. This property, together with the Hoeffding-type concentration, is then used in an involved analysis to show that eventually a discrepancy $4$ is reached. An extra iteration of this method finally yields a discrepancy bound of $3$.

\section{An Outline of the Analysis}\label{sec:techniques}

In this section we give a more detailed outline of our analysis. First, in \cref{sec:technical_contributions}, we present a collection of the most important technical results. These technical tools are crucial for our analysis and their proofs are considerably more challenging than deriving the discrepancy bounds using these tools. Additionally, we believe that some of these tools are of independent interest. At the end of this section, in \cref{sec:together}, we give a brief summary on how these technical results are combined to obtain our main result, which is a discrepancy bound of $3$.

\subsection{Our Technical Contributions}\label{sec:technical_contributions}

In order to prove small discrepancy bounds, we need to keep track of the number of tokens at a specific height. Thanks to the height-sensitive process defined earlier, the height of a token is non-increasing over time, and the sequence of locations of a token $i \in \mathcal{T}$, $(W_i^{(t)})_{t \geq 0}$, form a random walk (\cref{lem:OneTokenRW}). Crucially, we establish that these random walks are negatively associated: 
\begin{lemma}[simplified version of \cref{lem:height}]
Fix a subset of tokens $\mathcal{B} \subseteq \mathcal{T}$ at round $0$. Let $t > 0$ be any round, and fix the matchings between round $0$ and $t$. Then for any set $D \subseteq V$, the events $\{ W_i^{(t)} \in D \}, i \in \mathcal{B}$ are negatively associated.
\end{lemma}
In comparison to previous work (Lemma~4.2~from \cite{DBLP:conf/focs/SauerwaldS12}), our lemma yields the same statement but here tokens move following the definition of the height-sensitive process, whereas in \cite{DBLP:conf/focs/SauerwaldS12}, the two nodes exchange all tokens freely, which means that the height of a token could increase. Even though in this sense our process might be slightly harder to describe and analyze, our proof is simpler and more elementary than \cite{DBLP:conf/focs/SauerwaldS12}, e.g., we do not need the somewhat unwieldy negative regression condition from \cite{DBLP:journals/rsa/DubhashiR98}. 
We continue with a Hoeffding-like concentration bound.

\begin{lemma}[simplified version of \cref{ChernoffBound}]
Consider any load vector $x^{(0)}$ with $\discr(x^{(0)}) \leq K$ and any round $t\ge 1$ 
such that the sequence of matchings from round $1$ to $t$ is $(K,1/(2n))$-smoothing. Then for any stochastic vector $(a_w)_{w \in V}$, it holds for any $\delta > 0$,
\[
  \Pro{ 
        \left| \sum_{w \in V} a_w \cdot X_w^{(t)} - \overline{x} \right| \geq \delta} \leq 2 \cdot \exp\left( -\frac{ (\delta-1/(2n))^2}
        { 4 \|a\|_2^2} \right).
\]
\end{lemma}
In order to appreciate this result, we first discuss a more general (but somewhat harder to apply) version, which is derived in the proof of \cref{ChernoffBound}. This states that for \emph{any} matching sequence, and any stochastic vector $(a_w)_{w \in V}$, $\mu:= \Ex{\sum_{w \in V} a_w \cdot X_w^{(t)}}$,
and any $\delta > 0$ it holds
\[
  \Pro{ 
        \left| \sum_{w \in V} a_w \cdot X_w^{(t)} - \mu \right| \geq \delta} \leq 2 \cdot \exp\left( -\frac{ \delta^2}
        { 4 \|a\|_2^2} \right).
\]
This tail bound essentially matches the one from Hoeffding's inequality for $\sum_{w \in V} Y_w$ if the $Y_w$ are all independent and $Y_w \in [-a_w,a_w]$. However, in the load balancing process the range of the $X_w$'s are unbounded and also the $X_w$'s are far from being independent; for instance, two nodes matched in round $t$ must have a load difference of at most $1$.
 
Our concentration bound is established by carefully aggregating all rounding errors contributing to $\sum_{w \in V} a_w \cdot X_w^{(t)} - \mu$ by means of a quadratic potential function. On a high level, our proof resembles that in \cite{DBLP:conf/focs/SauerwaldS12}, however, one key difference is that we employ a more general potential function which involves the coefficients $a:=(a_w)_{w \in V}$.

Compared to prior work, our result generalizes~\cite[Lemma~3.5]{DBLP:conf/focs/SauerwaldS12}
to \emph{arbitrary} stochastic vectors $a$. The corresponding result in \cite{DBLP:conf/focs/SauerwaldS12} only works for the specific vector $a$ with $a_w:=\M_{w,u}^{[t+1,t_1]}$, for a fixed node $u\in V$ and round $t_1 \geq t$; in particular, if we choose $t_1=t$ then $a$ is a unit-vector.  There are also earlier and weaker versions of this inequality, e.g., in \cite[Theorem~4.6]{DBLP:conf/stoc/FriedrichS09}, which only match our form if the underlying graph and matchings have constant expansion. We therefore believe that our result will be the final word in the quest to find a tight and general concentration inequality of this type.

The first application of our concentration inequality (\cref{ChernoffBound}) is to bound the number of tokens above the average load:
\begin{lemma}[simplified version of \cref{lem:toNToken}]
Consider an initial load vector with $\discr(x^{(0)})\le K$ and round $t \geq 0$ such that the sequence $\left(\M^{(s)}\right)_{s=1}^{t}$ which is
$(K, 1/(2n))$-smoothing. Then,
\begin{align*}
\Pro{ \sum_{w \in V} \max \left\{ X_w^{(t)} - \overline{x}, 0 \right\} \leq 16 \cdot n }&\geq 1-2\cdot n^{-2}.
\end{align*}
\end{lemma}
This lemma exploits that in \cref{ChernoffBound} we can choose the stochastic vector $(a_w)_{w \in V}$ freely. For each possible subset $S \subseteq V$ and integer $i \geq 1$, we bound the existence of a ``bad'' set $S$ of size $\Theta(n / 2^{i})$ in which all nodes have load at least $\xbar+4 \cdot i$. This is done by choosing an appropriate vector $(a_w)$ and threshold $\delta$ in \cref{ChernoffBound}. The proof is then concluded by a simple union bound over all possible ``bad'' subsets $S \subseteq V$.

The next proposition is the most involved step in our analysis. It shows that we can reduce the discrepancy to a constant if the initial load vector has at most $O(n)$ tokens.
\begin{proposition}[simplified and informal version of \cref{lem:NtokensTo25Height}]
Let $t :=O( \tauglobal + \taulocal \cdot \log (n)/\log\log(n))$, and let $1 \leq L =O(1)$. Consider a $(\tauglobal,\taulocal)$-good sequence of matchings $\left(\M^{(s)}\right)_{s=1}^{\infty}$ and an arbitrary load vector $x^{(0)}$ with at most $( L-\Omega(1/\log (n))) \cdot n$ tokens.
Then,
\[
\Pro{ \max_{w \in V}X_w^{(t)} \leq \newL+1 }  
\geq 1-\exp\left( \Omega\left( \frac{\log(n)}{\log \log(n)} \right) \right).
\]
\end{proposition}

\begin{figure}[ht]
\begin{tikzpicture}[xscale=0.45,yscale=0.4,
win/.style={fill=green,opacity=0.4},
los/.style={fill=red,opacity=0.4},
r/.style={draw=black},
knoten/.style={rectangle,yscale=2.6,rounded corners=3pt,scale=2,draw=black, fill=white},
sty/.style={pos=0.0,below,scale=0.75}]

\draw[-stealth] (0,-0.25) to node[pos=1.0,above] {$V$} (0,8);
\draw[-stealth] (-0.25,0) to node[pos=1.0,right] {$t$} (30,0);

\draw (0.3,-0.1) to node[sty] {$0$} (0.3,0.1);
\draw (0.8,-0.1) to node[sty] {$1$} (0.8,0.1);

\draw (-0.1,0.5) to node[pos=0.0,left] {$1$} (0.1,0.5);

\draw (-0.1,7.5) to node[pos=0.0,left] {$L \cdot n$} (0.1,7.5);
\draw[dashed, gray] (0,7.5) to (28,7.5);
\draw (-0.1,6.1) to node[pos=0.0,left] {$n/\log(n)$} (0.1,6.1);
\draw[dashed, gray] (0,6.1) to (28,6.1);

\draw (7.1,-0.1) to node[sty] {$\tauglobal$} (7.1,0.1);
\draw[dashed, gray] (7.1,-0.1) to (7.1,8);

\draw (8.6,-0.1) to node[sty] {} (8.6,0.1);
\draw (10.1,-0.1) to node[sty] {} (10.1,0.1);
\draw (14.1,-0.1) to node[sty] {$t_1:=\tauglobal + \ell \cdot \taulocal$} (14.1,0.01);
\draw[dashed, gray] (14.1,-0.1) to (14.1,8);
\draw (21.1,-0.1) to node[sty] {$t_1+\tauglobal$} (21.1,0.1);
\draw[dashed, gray] (21.1,-0.1) to (21.1,8);

\draw (22.6,-0.1) to node[sty] {} (22.6,0.1);
\draw (24.1,-0.1) to node[sty] {} (24.1,0.1);
\draw (28.1,-0.1) to node[sty] {$t_1+ \ell \cdot \taulocal$} (28.1,0.1);

\draw [decorate,decoration={brace,amplitude=5pt,mirror,raise=4ex}] (8.6,-0.95) -- (7.1,-0.95) node[midway,yshift=2.75em,scale=0.75]{$\taulocal$};
\draw [decorate,decoration={brace,amplitude=5pt,mirror,raise=4ex}] (22.6,-0.95) -- (21.1,-0.95) node[midway,yshift=2.75em,scale=0.75]{$\taulocal$};

\draw [decorate,decoration={brace,amplitude=5pt,mirror,raise=4ex}] (0.4,-0.4) -- (13.9,-0.4) node[midway,yshift=-3em]{Phase $1$};
\draw [decorate,decoration={brace,amplitude=5pt,mirror,raise=4ex}] (14.1,-0.4) -- (27.9,-0.4) node[midway,yshift=-3em]{Phase $2$};

\draw[smooth, domain = 0.3:7.1, color=blue, thick] plot (\x,7.5);
\draw[smooth, domain = 7.1:14.1, color=blue, thick] plot (\x,{6.1+1.4*exp(-(\x-7.1))});
\node[blue,thick] at (7,8.2) {$Y^{(t)}$};

\draw[smooth, domain = 14.1:21.1, color=red, thick] plot (\x,4.5);
\draw[smooth, domain = 21.1:28.1, color=red, thick] plot (\x,{4.5 * exp(-(1.1*(\x-21.1)))});
 \node[red,thick] at (21,8.2) {$\widehat{Y}^{(t)}$};

\end{tikzpicture}

\caption{Illustration of Phases $1$ and $2$ in the proof of \cref{lem:NtokensTo25Height}. Phase 1 decreases $Y^{(t)}$, the number of tokens with height at least $L+1$, from $n \cdot L$ to $\frac{n}{\log(n)}$. Then Phase 2 decreases $\widehat{Y}^{(t)}$, the number of tokens with height at least $L+2$, to $0$. Both phases first use $\tauglobal$ rounds for a ``global mixing'', and then $\ell:=\frac{\log(n)}{\log \log(n)}$ short epochs of length $\taulocal$.
}\label{fig:pro_illustration}

\end{figure}
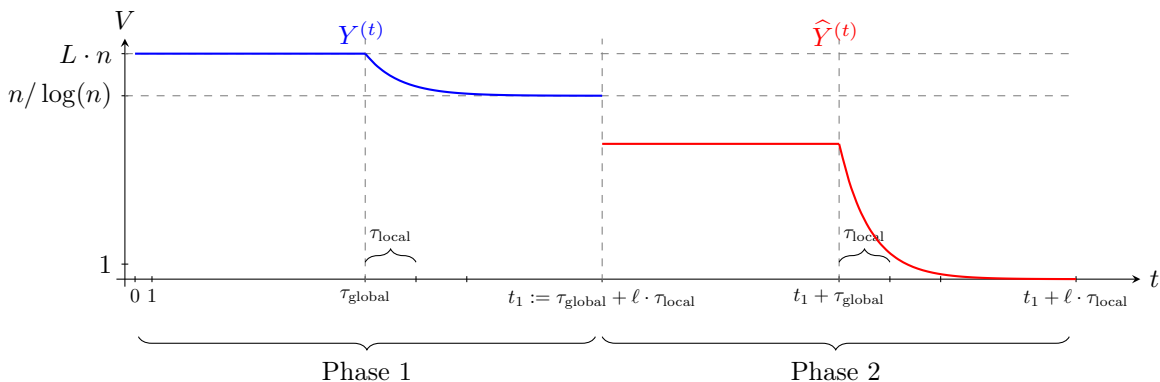

The proof of this proposition constitutes the most challenging part of our analysis.
Recall that we only have to balance $(\newL-\epsilon) \cdot n$ tokens for a constant $\epsilon > 0$ and  integer $\newL \geq 1$. Our goal is to gradually reduce the number of tokens with height $\geq \newL+1$, until no token with height $\geq \newL+2$ remains. The proposition uses two phases, which are illustrated in \cref{fig:pro_illustration}. 
\begin{itemize}\itemsep0pt
\item \textbf{Phase 1:} Reduce the number of tokens with height $\geq \newL+1$ from $(\newL-\epsilon) n$ to $n/\log (n)$ (see \cref{lem:rmlemmaphaseone}).

\item \textbf{Phase 2:} Reduce the number of tokens with height $\geq \newL+2$ further to $0$ (see \cref{lem:rmlemmaphasetwo}). 
\end{itemize}

Phases $1$ and $2$ run successively and are analyzed with the same framework, only with alternate parameters. The key difference is that after Phase $1$ we only need to cope with a sublinear number, that is, $n/\log(n)$, tokens at height at least $L+1$, and we wish to bound the number of tokens that remain at height $L+2$. This allows us to make faster progress in Phase $2$. Specifically, we prove exponential decay every $\taulocal$ rounds in Phase 1 and even super-exponential decay (factor $1/\log n$) in Phase 2. 
The analyses of both phases hinges on the key lemma stated below, which establishes a multiplicative drop on the number of tokens at height at least $L+1$ within $\taulocal$ rounds:
\begin{lemma}[simplified and informal version of \cref{lem:generallemmaone:RM} (Key Lemma)]\label{lem:keylemma}
Assume that a load vector $x^{(t)}$ has at most $(\newL-\epsilon) n$ tokens, where $t > 0$, $0<\epsilon < 1$ and $1 \leq \newL \leq \log^7(n)$ is an integer. Let $Y^{(t)}$ be the number of tokens with height at least $L+1$ in round $t$. 
Then,
\[ \Ex{ Y^{(t+\taulocal)} ~\Big|~ x^{(t)} } \le
\left( 1 - \frac{\epsilon}{\newL } + \frac{2}{\newL \cdot \log^4(n)} \right)
\cdot Y^{(t)}.
\]
\end{lemma}

In the following we will sketch the central ideas needed to establish the key lemma above.
We start with a token $i \in \mathcal{T}$ at height at least $\newL+1$, located at a node $u \in V$ in round $t$. Our goal is to lower bound the probability that after $\taulocal$ additional rounds, the token is still at this height. The location of token $i$ at time $t+ \taulocal$ is determined by a random walk with law $\M_{u,\cdot}^{[t+1,t+\taulocal]}$ (\cref{lem:OneTokenRW}); we call this (random) node $v \in V$. As heights of tokens are non-decreasing in $[t,t+\taulocal]$, the only way for token $i$ to remain at height at least $\newL+1$ is for there to be at least $L$ many other tokens $j$ which are also on node $v \in V$ at time $t+\taulocal$. 
Using the negative association lemma (\cref{lem:height}), the expected number of tokens which collide with token $i$ at round $t+\taulocal$ can be upper bounded by (see \cref{lem:collision}),
\begin{align*}
 \sum_{w \in V}  \Biggl( \underbrace{ \sum_{v \in V}  
 \M_{u,v}^{[t+1,t+\taulocal]} \cdot \M_{w,v}^{[t+1,t+\taulocal]}}_{=:a_w}  \Biggr) \cdot X_w^{(t)}.
\end{align*}
We proceed with upper bounding the right hand side. This is a convex combination of a load vector, which makes it amenable to our new concentration inequality for convex combinations of loads. Using that the matching sequences are of length $\taulocal$, we obtain with reasonably large probability over the matching sequence that $\|a\|_2^2$ is small. Once we have established this, we apply \cref{ChernoffBound} to the above sum; here, for fixed $a$ the randomness is over the matching sequences and shuffling steps in $[1,t]$, where $t \geq \tauglobal$. Taking aside some technicalities, we can then conclude that there is an expected drop in the number of tokens at height $L+1$ within $\taulocal$ rounds. 

One significant technical challenge is to iterate this argument over consecutive epochs of length $\taulocal$, as each epoch also depends on the random decisions in previous epochs (both matchings and shuffling decisions). We will overcome these dependencies by carefully defining events (corresponding to the local and global behaviour of the process; similar to $\taulocal$ and $\tauglobal)$, and then integrate these events into a submartingale which shows that the number of tokens at height $L+1$ drops.

\subsection{Putting the Pieces Together}\label{sec:together}
With the powerful \cref{lem:NtokensTo25Height} at hand, we are able to reduce the discrepancy from any $K$ to $3$ in $O(\log(Kn)/\log(n) \cdot \tauglobal + \taulocal \cdot \log (n)/\log\log(n))$
rounds as follows:
\begin{enumerate}\itemsep0pt
    \item \textbf{Discrepancy $K$ $\rightarrow$ Linear Number of Tokens}: We first reduce the number of tokens above the average load to $16 n$ using $O(\tauspectral{K})$
    many rounds (see~\cref{lem:toNToken}).
    \item \textbf{Linear Number of Tokens $\rightarrow$ Discrepancy $38$}: Applying \cref{lem:NtokensTo25Height} we then reduce the discrepancy to $38$ within $O(\tauglobal + \taulocal \cdot \log (n)/\log\log(n))$ many rounds (see \cref{lem:discrepancy-38}).
    \item \textbf{Discrepancy $38$ $\rightarrow$ Discrepancy $4$}: By an additional application of \cref{lem:NtokensTo25Height}, after $O(\tauglobal + \taulocal \cdot \log (n)/\log\log(n))$ further rounds the discrepancy drops to $4$ (see \cref{lem:disc54to4}).
    \item \textbf{Discrepancy $4$ $\rightarrow$ Discrepancy $3$}: Now \cref{lem:NtokensTo25Height} cannot be applied directly to reduce the discrepancy any further, but we can still reduce the discrepancy to $3$ by a more tailored analysis of the movement of the tokens at the two highest levels in the next $O(\tauglobal + \taulocal \cdot \log (n)/\log\log(n))$ rounds, and obtain the main result of our paper. 
\end{enumerate}

\section{Summary and Open Problems}\label{sec:conclusions}

In this paper we show that, for the matching model, discrete load balancing is as efficient and effective as continuous load balancing.
We showed that in the discrete setting with integer loads, a discrepancy of $3$ is reached in a time that matches the standard spectral bound on the time needed by the continuous setting. In particular, this means that for expanders and a polynomial initial discrepancy, our load balancing schemes achieve a discrepancy of $3$ using only $O(\log (n))=O(\operatorname{diam}(G))$ rounds, which is optimal not only for a distributed but also a centralized setting.

As an improvement over previous works, our general result holds for a wider class of graphs (e.g., including non-regular graphs) and models (e.g., including the asynchronous model). Also the constants in our runtime as well as in the achieved discrepancy are explicit and small, compared to large and non-explicit constants in \cite{DBLP:conf/focs/SauerwaldS12}.

At the heart of our analysis lie new association and concentration 
results, which we believe to be of independent interest. As the most involved step, showing that $O(n)$ tokens can be balanced with maximum load $O(1)$, which can be leveraged to prove a discrepancy of $3$ for an arbitrary number of tokens. It should be noted that reaching a discrepancy of $1$ needs $\Omega(n)$ rounds for any sequence of matchings (see~\cite{MS10}). Hence one may wonder whether our discrepancy bound could be improved from $3$ to $2$. We believe that this might be possible, but it will likely need stronger assumptions on the matchings than just being $(\tauglobal,\taulocal)$-good.

\section{Association  and Concentration Results}
\label{sec:auxiliary-results}

In this section we show two main components our analysis; first, the negative association result about token movements and secondly, the Hoeffding-like concentration inequality.

\subsection{Height-Sensitive Negative Association}

\begin{lemma}[Height-Sensitive Negative Association]
\label{lem:height}
Consider any pair of rounds $0 \leq t_1 < t_2$, and let $ (\M^{(s)})_{s=t_1}^{t_2}$ be an arbitrary but fixed sequence of matchings. Further, let $w^{(t_1)}=(w_i)^{(t_1)}_{i \in \mathcal{T}}$ be a fixed location vector in round $t_1 \geq 0$. Then for any subset of tokens $\mathcal{B} \subseteq \mathcal{T}$ and any subset of nodes  $D \subseteq V$, it holds that
\[
\Proco{ \bigcap_{ i \in \mathcal{B} } \{ W_i^{(t_2)} \in D \}}{ w^{(t_1)}} \leq \prod_{ i \in \mathcal{B}} \Proco{ W_{i}^{(t_2)} \in D}{w^{(t_1)}} 
= \prod_{ i \in \mathcal{B}} \M_{w_i^{(t_1)},D}^{[t_1+1,t_2]}.
\]

\end{lemma}

\begin{proof}
Throughout the proof we always assume that $w^{(t_1)}$ is fixed and we omit the conditioning.
Note that by \cref{lem:OneTokenRW}, we have $ \prod_{ i \in \mathcal{B}} \Pro{ W_i^{(t_2)} \in D } = \prod_{i \in \mathcal{B}} \M_{w_i^{(t_1)},D}^{[t_1+1,t_2]} $  and therefore it only remains to prove the inequality
\[ \Pro{ \bigcap_{ i \in \mathcal{B} } \{ W_i^{(t_2)} \in D \} }\leq \prod_{i \in \mathcal{B}} \M_{w_i^{(t_1)},D}^{[t_1+1,t_2]}.\]
For every round $t\in[t_1,t_2]$ and a token $i \in \B$ we define a random variable
\begin{equation}
Z_i^{(t)}:= \sum_{u\in V}\1_{W_i^{(t)}=u}\cdot \M_{u,D}^{[t+1,t_2]} \qquad \text{and} \qquad Z^{(t)}:=\prod_{i \in \mathcal{B}} \ Z_i^{(t)}. \label{eq:Z-Mar}    
\end{equation}

Hence, conditional on the token location $W_i^{(t)}$, $Z_i^{(t)}$ is the probability that token $i$ starting at time $t$ is on a node $w\in D$ at time $t_2$. 
In the following analysis we use the simplified notation \[\M_{W_i^{(t)},D}^{[t+1,t_2]} := \sum_{u\in V}\1_{W_i^{(t)}=u} \cdot \M_{u,D}^{[t+1,t_2]}.\]
Since $\M^{[t_2+1,t_2]} = \mathbf{I}$, 
\[Z^{(t_2)} = \prod_{i\in \B}\M_{W_i^{(t_2)},D}^{[t_2+1,t_2]}.\]
Hence, $Z_i^{(t_2)}$ is a random variable that can take on the values zero or one.
Therefore,
\[\E\left[Z^{(t_2)}\right]=\Pro{Z^{(t_2)}=1} =
\Pro{ \bigcap_{i \in \mathcal{B}} \left\{ Z_i^{(t_2)}=1 \right\} } =
\Pro{\bigcap_{i\in \B}\left\{W_i^{(t_2)}\in D\right\}}.\]

The key idea is to prove that the sequence $Z^{(t)}$, $t\in[t_1,t_2]$ forms a supermartingale with respect to 
the filtration $\mathfrak{F}^{(t)}$, which reveals all random decisions of the process between rounds $t_1$ and $t$ (so in particular, it reveals the location vectors $W^{(t)}=\left(W_1^{(t)},\ldots, W_{|\mathcal{T}|}^{(t)} \right)$), that is,
\begin{equation}\label{eq:SuperMartingale2}
\E\left[Z^{(t)}~|~ \mathfrak{F}^{(t-1)} \right] \le Z^{(t-1)},
\end{equation}
which would immediately imply
\[\E\left[Z^{(t_2)}\right] \leq Z^{(t_1)}.\]
It then follows
\[\Pro{\bigcap_{i\in \B}\left\{W_i^{(t_2)}\in D\right\}}=\E\left[Z^{(t_2)}\right] \le Z^{(t_1)} = \prod_{i\in \B} \M_{w_i^{(t_1)},D}^{[t_1+1,t_2]},\]
finishing the proof.

It remains to prove \cref{eq:SuperMartingale2}.
Without loss of generality we may assume that each matching consists only of one edge. This is because for any sequence of matchings $(\M^{(s)})_{s=t_1}^{t_2}$ in which there is a round $s \in [t_1+1,t_2]$ whose matching includes more than one edge, we can decompose such round $s$ into multiple ``sub-rounds'', each consisting of exactly one matching. It is clear that this yields the same process, only with a larger number of rounds.

Fix round $t$ with $t_1\le t\le t_2$ and assume $\{u,v\}$ is the single edge in $\M^{(t)}$.
First we compute $Z^{(t-1)}$.
We define $T_u \subseteq \B$ as the tokens from $\B$ which are on $u$ and $T_v \subseteq \B$ as the tokens from $\B$ which are on $v$ at the beginning of round $t$.
Note that it is possible that there are tokens $i \not\in \B$ on those nodes.
We get
\[ \M_{u,D}^{[t,t_2]}=\M_{v,D}^{[t,t_2]}=
\M_{u,v}^{(t)} \cdot \M_{v,D}^{[t+1,t_2]} + \M_{u,u}^{(t)} \cdot \M_{u,D}^{[t+1,t_2]} =
\frac{\M_{u,D}^{[t+1,t_2]}+\M_{v,D}^{[t+1,t_2]}}{2}. \]
For ease of presentation we define
\[
p:=\M_{u,D}^{[t+1,t_2]} \qquad \mbox{ and } \qquad q:=\M_{v,D}^{[t+1,t_2]}
.
\]
From the definition (\ref{eq:Z-Mar}) it follows that
\begin{align}
Z^{(t-1)} &=
\prod_{i \in \B \setminus ( T_u \cup T_v ) } Z_i^{(t-1)} \cdot \prod_{i \in T_u \cup T_v} Z_i^{(t-1)} 
=\prod_{i\in \B \setminus (T_u \cup T_v)}Z_i^{(t-1)} \cdot \left(\frac{p+q}{2}\right)^{|T_u \cup T_v|}. \label{eq:ztminus1}
\end{align}
We will continue to analyze $\E[Z^{(t)}~|~ \mathfrak{F}^{(t-1)}] $ and eventually upper bound it by the right hand side in \cref{eq:ztminus1}, which completes the argument.
For each token $i$ located at $\{u,v\}$ let $S(i)$ be its sibling token on the other node after the moving step of round $t$ and before the shuffling step; note that it is possible for a token to have no sibling.
We partition tokens in $T_u \cup T_v$ into three sets:
\begin{align*}
\B_{1} &= \left\{ i\in T_u \cup T_v ~|~  S(i) \in T_u\cup T_v \right\}, \\
\B_{2} &= \left\{i\in T_u \cup T_v ~|~ S(i) \in \mathcal{T} \setminus (T_u\cup T_v) \right\}, \\
\B_{3} &= \left\{ i\in T_u \cup T_v ~|~  \mbox{$S(i)$ does not exist.} \right\}.
\end{align*}
Note that the set $\B_{3}$ corresponds to the excess token (if there is one), and thus $|\B_{3}| \in \{0,1\}$. We will now apply this partitioning to $\Ex{ Z^{(t)} ~\Big|~ \mathfrak{F}^{t}}$,
\begin{align}
\lefteqn{\Exco{Z^{(t)}}{\mathfrak{F}^{(t-1)}} } \notag 
\\
&= \Exco{ \prod_{i \in B} Z_i^{(t)} }{\mathfrak{F}^{(t-1)}}
 \notag \\ 
&= \Ex{\prod_{i\in \B \setminus (T_u \cup T_v)}Z_i^{(t)}\cdot \prod_{i\in T_u \cup T_v} Z_i^{(t)}~\Bigg|~\mathfrak{F}^{(t-1)} }  \notag \\
& \overset{(a)}{=} \prod_{i\in \B \setminus (T_u \cup T_v)}Z_i^{(t-1)}\cdot \E\left[\prod_{i\in \B_{1}} Z_i^{(t)}\cdot \prod_{i\in \B_{2}} Z_i^{(t)}\cdot \prod_{i \in \B_{3}} Z_i^{(t)}~\Bigg|~\mathfrak{F}^{(t-1)}\right]  \notag \\&
\overset{(b)}{=} \prod_{i\in \B \setminus (T_u \cup T_v)}Z_i^{(t-1)}\cdot\E\left[\prod_{i\in \B_{1}} Z_i^{(t)}~\Bigg|~\mathfrak{F}^{(t-1)}\right] \cdot\E\left[\prod_{i\in \B_{2}} Z_i^{(t)}~\Bigg|~\mathfrak{F}^{(t-1)}\right]
\cdot \E\left[\prod_{i\in \B_{3}} Z_i^{(t)}~\Bigg|~\mathfrak{F}^{(t-1)}\right], \label{eq:corr}
\end{align}
where $(a)$ follows from the fact that for tokens $i\in \B \setminus (T_u\cup T_v)$, we have $Z_i^{(t)}=Z_i^{(t-1)}$ since $W_i^{(t)}$ is not matched with any node in this round; $(b)$ follows since tokens which are in different $\mathcal{B}_j$, $1 \leq j \leq 3$ must be at different heights and therefore move independently in the shuffling phase.

We will now analyze the three different expectations in \cref{eq:corr}
\paragraph{Case 1: The expectation over $\B_{1}$.}
Using that any token $i \in \mathcal{B}_1$ only depends on its sibling token $S(i)$ we can group all tokens in $\mathcal{B}_1$ in pairs, and we denote this relation by $\sim$. Further, a token $i$ with sibling $j$ will take the opposite action in the shuffling step, and therefore
\begin{align*}
    \E\left[\prod_{i\in \B_{1}} Z_i^{(t)}~\Bigg|~\mathfrak{F}^{(t-1)}\right]
    &=\prod_{i < j\in \B_{1} \colon i \sim j} \E\left[ Z_i^{(t)} \cdot Z_j^{(t)} ~\Bigg|~\mathfrak{F}^{(t-1)} \right] = \prod_{i < j\in \B_{1} \colon i \sim j} \left( p \cdot q \right).
\end{align*}

\paragraph{Case 2: The expectation over $\B_{2}$.} 
Since each token in $\B_{2}$ has no sibling in $\B$ (let alone $\B_2$), their movements in the shuffling step are independent. Further, each such token is on $u$ or $v$ with the same probability which yields,
\begin{align*}
    \E\left[\prod_{i\in \B_{2}} Z_i^{(t)}~\Bigg|~\mathfrak{F}^{(t-1)}\right]
    &=\prod_{i \in \B_2 } \E\left[ Z_i^{(t)}  ~\Bigg|~\mathfrak{F}^{(t-1)} \right] = \prod_{i \in B_2} \left( \frac{p+q}{2} \right).
\end{align*}

\paragraph{Case 3: The expectation over $\B_{3}$.} 
This is analogous to Case 2, as a token in $\B_{3}$ has no sibling in $\B$ (in fact there is not even any other token at this height). Hence,
\begin{align*}
 \E\left[\prod_{i\in \B_{3}} Z_i^{(t)}~\Bigg|~\mathfrak{F}^{(t-1)}\right]
    &= \prod_{i \in B_3} \left( \frac{p+q}{2} \right).
\end{align*}

Aggregating the contributions of all tokens in $\mathcal{B}$ in round $t$ using the three cases above, 
\begin{align*}
\lefteqn{\Ex{ Z^{(t)} ~\Big|~ \mathfrak{F}^{(t-1)} } } \\
&\leftstackrel{(\ref{eq:corr})}{=} \prod_{i\in \B \setminus (T_u \cup T_v)}Z_i^{(t-1)}\cdot\E\left[\prod_{i\in \B_{1}} Z_i^{(t)}~\Bigg|~\mathfrak{F}^{(t-1)}\right] \cdot\E\left[\prod_{i\in \B_{2}} Z_i^{(t)}~\Bigg|~\mathfrak{F}^{(t-1)}\right]
\cdot \E\left[\prod_{i\in \B_{3}} Z_i^{(t)}~\Bigg|~\mathfrak{F}^{(t-1)}\right] \qquad   \\
&\stackrel{(a)}{=} \prod_{i\in \B \setminus (T_u \cup T_v)}Z_i^{(t-1)}\cdot \prod_{i < j\in \B_{1} \colon i \sim j} \left( p \cdot q \right) \cdot \prod_{i\in \B_{2}} \frac{p+q}{2}\cdot \prod_{i\in \B_{3}} \frac{p+q}{2} 
\\&
\overset{(b)}{\le} \prod_{i\in \B \setminus (T_u \cup T_v)}Z_i^{(t-1)}\cdot \prod_{i < j\in \B_{1} \colon i \sim j} \left(\frac{p+q}{2}\right)^2 \cdot \prod_{i\in \B_{2}} \frac{p+q}{2} \cdot \prod_{i\in \B_{3}} \frac{p+q}{2}
\\& = \prod_{i\in \B \setminus (T_u \cup T_v)}Z_i^{(t-1)}\cdot\left(\frac{p+q}{2}\right)^{|T_u\cup T_v|} \\
&\stackrel{(c)}{=} Z^{(t-1)},
\end{align*}
where $(a)$ uses Cases 1, 2 and 3 above, $(b)$ follows from the fact that $p\cdot q \leq ((p+q)/2)^2$ and $(c)$ follows from \cref{eq:ztminus1}.
Hence $\Ex{ Z^{(t)} ~\Big|~ \mathfrak{F}^{(t-1)} }\leq Z^{(t-1)}$,
which is \cref{eq:SuperMartingale2} and, consequently, completes the proof.
\end{proof}

\subsection{Concentration Inequality for Convex Combinations of Loads}

The following ``Hoeffding-like'' concentration inequality is another key tool in our analysis. We use it frequently to derive (upper) bounds on the number of tokens on a set of nodes. 

\begin{lemma}[Hoeffding Bound for Convex Combinations of Loads]
\label{ChernoffBound}
Consider any load vector $x^{(0)}$ with $\discr(x^{(0)}) \leq K$ and any round $t\ge 1$ 
such that $(\M^{(s)})_{s=1}^{t}$ is a $(K,\kappa)$-smoothing sequence of matchings with $K > \kappa$. Then for any stochastic vector $(a_w)_{w \in V}$, it holds for any $\delta > 0$,
\[
  \Proco{ 
        \left| \sum_{w \in V} a_w \cdot X_w^{(t)} - \overline{x} \right| \geq \delta}{(\M^{(s)})_{s=1}^{t} } \leq 2 \cdot \exp\left( -\frac{ (\delta-\kappa)^2}
        { 4 \|a\|_2^2} \right).
\]
\end{lemma}

As a special case of this concentration inequality, we can take any node $v$, define the unit-vector $a_w=\mathbf{1}_{w=v}$ (and thus $\|a\|_2^2=1$), pick $\delta= \sqrt{12 \cdot \log (n)}$ and then apply a union bound over all nodes $v \in V$. This immediately yields the following bound on the discrepancy.
\begin{corollary}\label{cor:disc32}
Consider any load vector $x^{(0)}$ with $\discr(x^{(0)}) \leq K$.
Consider any round $t\ge 1$ and let $(\M^{(s)})_{s=1}^{t}$ be a fixed $(K,1)$-smoothing sequence of matchings. Then,
\[
 \Pro{ 
 \discr(X^{(t)}) \leq  \sqrt{48\cdot \log(n)} + 1} \geq 1 - 2\cdot n^{-1}.
\]
\end{corollary}
The exact same bound was shown in \cite[Theorem~3.6]{DBLP:conf/focs/SauerwaldS12}. 
\subsection{Proof of the Concentration Inequality (\cref{ChernoffBound})}

Before giving the proof of \cref{ChernoffBound}, we need to introduce some notation and auxiliary results.

We first recall a standard formula (e.g., \cite[Equation~2.5]{DBLP:conf/focs/SauerwaldS12}); for any node $w \in V$ ,
\begin{align}
X_w^{(t)} - \xi_{w}^{(t)}
&= \sum_{s=1}^{t} \sum_{[u:v] \in \M^{(s)}} E_{u,v}^{(s)} \cdot \left( \M_{u,w}^{[s+1,t]} - \M_{v,w}^{[s+1,t]} \right), \label{eq:disc_cont}
\end{align}
where $\xi_w^{(t)} = \sum_{u \in V} x_u^{(0)} \cdot \M_{u,w}^{[1,t]} $ is the load after $t$ rounds in the continuous model where load values can be fractional, and $E_{u,v}^{(t)} = \frac{1}{2}\cdot \Odd\left(X_u^{(t-1)}+X_v^{(t-1)}\right)\cdot \Phi_{u,v}^{(t)}$ is the rounding error; here $\Phi_{u,v}^{(t)} \in \{-1,+1\}$ is the random orientation which has Rademacher distribution (that is, it is uniform in $\{-1,+1\}$).
Therefore,
\begin{align*}
\sum_{w\in V} a_w \cdot \left( X_w^{(t)} - \xi_w^{(t)} \right) &= \sum_{w \in V} a_w \cdot \left( \sum_{s=1}^{t} \sum_{[u:v] \in \M^{(s)}} E_{u,v}^{(s)} \cdot \left( \M_{u,w}^{[s+1,t]} - \M_{v,w}^{[s+1,t]} \right) \right) \\
&= \sum_{s=1}^{t} \sum_{[u:v] \in \M^{(s)}} E_{u,v}^{(s)} \cdot \left[ \sum_{w \in V} a_w \cdot \left( \M_{u,w}^{[s+1,t]} -\M_{v,w}^{[s+1,t]} \right) \right].
\end{align*}

We will first state a Hoeffding-like concentration inequality, applied to last expression above:
\begin{lemma}[cf.~{\cite[Lemma~3.4]{DBLP:conf/focs/SauerwaldS12}}]\label{lem:ss12chernoff}
Consider an arbitrary but fixed sequence of matchings  $(\M^{(s)})_{s=1}^{t}$ for any rounds $t \geq 1$ and the 
load vector $x^{(0)}$. For any family of numbers $g_{u,v}^{(s)}, [u:v] \in \M^{(s)}, 1 \leq s \leq t$, define the random variable
\[
 Z := \sum_{s=1}^{t} \sum_{[u:v] \in \M^{(s)}} E_{u,v}^{(s)} \cdot g_{u,v}^{(s)}.
\]
Then, $\Ex{Z}=0$, and for any $\delta>0$, it holds that
\begin{align*}
 \Pro{ \left| Z - \Ex{Z} \right| \geq \delta } &\leq 2 \cdot \exp\left( - \frac{ \delta^2 }{ 2 \sum_{s=1}^{t} \sum_{[u:v] \in \M^{(s)}} \left( g_{u,v}^{(s)} \right)^2 }  \right).
\end{align*}
\end{lemma}

The next lemma can be seen as a generalization of the first statement of Lemma~3.2 from \cite{DBLP:conf/focs/SauerwaldS12},  which considers the fixed vector $a$ with $a_w:=1$ for node $w\in V$ and $a_u=0$ for $u\neq w$. We extend it to arbitrary stochastic vectors $a$, which is crucial to prove~\cref{ChernoffBound} in its full generality.
\begin{lemma}\label{lem:theorem32_new}
Consider an arbitrary sequence of matchings $\left(\M^{(s)}\right)_{s=1}^{\infty}$.
Let $(a_k)_{k \in V}$  be any stochastic vector. Then: 
\begin{enumerate}
\item For any two rounds $0 \leq t_1 \le t$ it holds, 
\[
\sum_{s=1}^{t_1} \sum_{[u:v] \in \M^{(s)}} \left( \sum_{k \in V} a_k \cdot \left( \M_{u,k}^{[s+1,t]} -  \M_{v,k}^{[s+1,t]}\right) \right)^2 
\le 2\cdot  \|a \|_2^2.
\]

\item For any two rounds $0 \leq t_1\le t$ it holds,
\[
\sum_{w\in V}\left(\sum_{k\in V} a_k\cdot \M_{w,k}^{[t_1+1,t]}\right)^2 \le \|a\|_2^2.
\]
\end{enumerate}
\end{lemma}
\begin{proof}

The proof follows closely the proof of Theorem~3.2 from \cite{DBLP:conf/focs/SauerwaldS12}.
We define a potential function as
\begin{align*}
\Psi^{(s)} := \sum_{w \in V} \left( \sum_{k \in V} a_k \M_{w,k}^{[s+1,t]} - \frac{1}{n} \right)^2,
\end{align*}
for any round $1 \leq s \leq t$. This is a generalization of the potential function in \cite{DBLP:conf/focs/SauerwaldS12}, where all $a_k:=1$.

\paragraph{Proof of the first statement.}
Consider now any round $1 \leq s \leq t$, and let $u,v $ be nodes with $[u:v] \in \M^{(s)}$.
We have
\[
\M_{u,k}^{[s,t]} = \M_{u,u}^{(s)}\cdot \M_{u,k}^{[s+1,t]} + \M_{u,v}^{(s)}\cdot \M_{v,k}^{[s+1,t]}  = \frac{ \M_{u,k}^{[s+1,t]} + \M_{v,k}^{[s+1,t]} }{2}.
\]
To simplify notation, define for any two nodes $u,v \in V$,
\[
y_{u,v} := \M_{u,v}^{[s+1,t]}.
\]
With that notation, the contribution of these nodes to $\Psi^{(s)} - \Psi^{(s-1)}$ is
\begin{align*}
& \left( \sum_{k \in V} a_k \M_{u,k}^{[s+1,t]} - \frac{1}{n} \right)^2 + \left( \sum_{k \in V} a_k \M_{v,k}^{[s+1,t]} - \frac{1}{n} \right)^2 - \left( \sum_{k \in V} a_k \M_{u,k}^{[s,t]} - \frac{1}{n} \right)^2 - \left( \sum_{k \in V} a_k \M_{v,k}^{[s,t]} - \frac{1}{n} \right)^2 \\
&= \left( \sum_{k \in V} a_k y_{u,k} - \frac{1}{n} \right)^2 + \left( \sum_{k \in V} a_k y_{v,k} - \frac{1}{n} \right)^2  \\
&\quad- \left( \sum_{k \in V} a_k \frac{y_{u,k}+y_{v,k}}{2} - \frac{1}{n} \right)^2 - \left( \sum_{k \in V} a_k \frac{y_{u,k}+y_{v,k}}{2} - \frac{1}{n} \right)^2\!\!,
\end{align*}
having used that $\M_{u,k}^{[s,t]} = \M_{v,k}^{[s,t]} = \frac{y_{u,k}+y_{v,k}}{2}$. We can further rewrite the last expression as
\begin{align*}
&= \left( \sum_{k \in V} a_k y_{u,k} \right)^2 - \frac{2}{n} \cdot \sum_{k \in V} a_k y_{u,k} + \frac{1}{n^2} + \left( \sum_{k \in V} a_k y_{v,k} \right)^2 - \frac{2}{n} \cdot \sum_{k \in V} a_k y_{v,k} + \frac{1}{n^2} \\
&~~ - 2 \cdot \left( \sum_{k \in V} a_k \cdot \frac{y_{u,k}+y_{v,k}}{2} \right)^2 + \frac{4}{n} \cdot \sum_{k \in V} a_k \cdot \frac{y_{u,k}+y_{v,k}}{2} - \frac{2}{n^2} \\
&= \left( \sum_{k \in V} a_k y_{u,k} \right)^2 + \left( \sum_{k \in V} a_k y_{v,k} \right)^2 - \frac{1}{2} \cdot \left( \sum_{k \in V} a_k y_{u,k} + \sum_{k \in V} a_k y_{v,k} \right)^2
\\& ~~ + \frac{2}{n}\cdot \underbrace{\left(-\sum_{k\in V} a_ky_{u,k}- \sum_{k\in V}a_ky_{v,k}+\sum_{k\in V} a_k\cdot(y_{u,k}+y_{v,k})\right)}_{=~0} \\
&\stackrel{(a)}{=} \frac{1}{2} \cdot \left( \sum_{k \in V} a_k y_{u,k} - \sum_{k \in V} a_k y_{v,k} \right) ^2,
\end{align*}
where in $(a)$ we used the fact that $a^2 + b^2 - \frac{1}{2} \cdot (a+b)^2 = \frac{1}{2} \cdot (a-b)^2$. If a node is not matched in round $s$, then its contribution to $\Psi^{(s)} - \Psi^{(s-1)}$ equals zero.
Accumulating the contribution of all nodes yields
\begin{equation}\label{eq:psi:non-increasing}
\Psi^{(s)} - \Psi^{(s-1)} = \frac{1}{2} \cdot \sum_{[u:v] \in \M^{(s)}} \left( \sum_{k \in V} a_k \cdot \left( \M_{u,k}^{[s+1,t]} - \M_{v,k}^{[s+1,t]} \right) \right)^2.
\end{equation}
Therefore,
\begin{align*}
\Psi^{(t_1)} \geq \Psi^{(t_1)}-\Psi^{(0)} &= \sum_{s=1}^{t_1} \left( \Psi^{(s)} - \Psi^{(s-1)}\right) = \frac{1}{2} \cdot \sum_{s=1}^{t_1} \sum_{[u:v] \in \M^{(s)}} \left( \sum_{k \in V} a_k \cdot \left( \M_{u,k}^{[s+1,t]} -  \M_{v,k}^{[s+1,t]} \right) \right)^2,
\end{align*}
and rearranging gives,
\begin{equation*}
    \sum_{s=1}^{t_1} \sum_{[u:v] \in \M^{(s)}} \left( \sum_{k \in V} a_k \cdot \left( \M_{u,k}^{[s+1,t]} -  \M_{v,k}^{[s+1,t]}\right) \right)^2 \leq 2 \cdot \Psi^{(t_1)} \leq 2 \cdot \Psi^{(t)}, 
\end{equation*}
where the last inequality holds since the potential $\Psi^{(s)}$ is non-decreasing over $s$ (\cref{eq:psi:non-increasing}). 
Moreover, since $\M^{[t+1,t]}$ is the identity matrix $\mathbf{I}$, we have $\M^{[t+1,t]}_{w,k}=1$ for $k=w$ and $\M^{[t+1,t]}_{w,k}=0$ for $k\neq w$.
 Hence
\begin{equation}
   \Psi^{(t)}= \sum_{w \in V} \left( \sum_{k \in V} a_k \M_{w,k}^{[t+1,t]} - \frac{1}{n} \right)^2 = \sum_{w \in V} \left( a_w - \frac{1}{n} \right)^2 = \|a \|_2^2-\frac{1}{n}.
   \label{eq:sec:removed}
\end{equation}
Combining the last two statements finishes the proof of the first statement.

\paragraph{Proof of the second statement.}
Note that since $t_1 \le t$, it follows from \cref{eq:psi:non-increasing} that
\[ \Psi^{(t_1)} \le \Psi^{(t)}.\]
Substituting the definition of $\Psi^{(t_1)}$ and $\Psi^{(t)}$ on both sides gives
\begin{equation}\label{eq:temp:hh}
\sum_{w\in V}\left(\sum_{k\in V} a_k\cdot \M_{w,k}^{[t_1+1,t]}-\frac{1}{n}\right)^2 \le \sum_{w\in V}\left(\sum_{k\in V} a_k\cdot \M_{w,k}^{[t+1,t]}-\frac{1}{n}\right)^2 \overset{\text{Eq.~(\ref{eq:sec:removed})}}{=} \|a\|_2^2-\frac{1}{n}.
\end{equation}
Finally, we have
\begin{align*}
\sum_{w\in V}\left(\sum_{k\in V} a_k\cdot \M_{w,k}^{[t_1+1,t]}\right)^2 \stackrel{(a)}{=} \sum_{w\in V} \left( \sum_{k\in V} a_k\cdot \M_{w,k}^{[t_1+1,t]} -\frac{1}{n}\right)^2 + \frac{1}{n} \overset{\text{Eq.~(\ref{eq:temp:hh})}}{=} \|a\|_2^2 ,
\end{align*}
where $(a)$ follows from \cref{obs:inter:secondNormBound} (since the matrix $\M^{[t_1+1,t]}$ is doubly stochastic and $\|a\|_1=1$).
\end{proof}

\begin{proof}[Proof of \cref{ChernoffBound}]
By \cref{eq:disc_cont},
\begin{align*}
\sum_{w\in V} a_w \cdot \left( X_w^{(t)} - \xi_w^{(t)} \right) &= \sum_{w \in V} a_w \cdot \left( \sum_{s=1}^{t} \sum_{[u:v] \in \M^{(s)}} E_{u,v}^{(s)} \cdot \left( \M_{u,w}^{[s+1,t]} - \M_{v,w}^{[s+1,t]} \right) \right) \\
&= \sum_{s=1}^{t} \sum_{[u:v] \in \M^{(s)}} E_{u,v}^{(s)} \cdot \left[ \sum_{w \in V} a_w \cdot \left( \M_{u,w}^{[s+1,t]} -\M_{v,w}^{[s+1,t]} \right) \right].
\end{align*}
Since the time interval $[1,t]$ is $(K,\kappa)$-smoothing, it holds that $-\kappa \le \xi_{w}^{(t)} - \overline{x} \le \kappa$ for all nodes $w \in V$, and therefore
\[
\sum_{w \in V} a_w \cdot \xi_w^{(t)} = \overline{x} + B,
\]
where $|B| \leq \kappa$. This implies 
\begin{equation}\label{eq:weightedSumLoadbound}
\sum_{w \in V} a_w \cdot X_w^{(t)} -\overline{x} = B + 
\sum_{s=1}^{t} \sum_{[u:v] \in \M^{(s)}} E_{u,v}^{(s)} \cdot \left[ \sum_{w \in V} a_w \cdot \left( \M_{u,w}^{[s+1,t]} -  \M_{v,w}^{[s+1,t]} \right) \right] 
\end{equation}
Following the notation of \cref{lem:ss12chernoff}, let us define
\[
Z:= \sum_{s=1}^{t} \sum_{[u:v] \in \M^{(s)}} E_{u,v}^{(s)} \cdot \left[ \sum_{w \in V} a_w \cdot \left( \M_{u,w}^{[s+1,t]} -  \M_{v,w}^{[s+1,t]} \right) \right],
\]
and
\[
 g_{u,v}^{(s)} := \sum_{w \in V} a_w \cdot \left( \M_{u,w}^{[s+1,t]} -  \M_{v,w}^{[s+1,t]} \right).
\]
Since $\Ex{E_{u,v}^{(s)}}=0$, then $\Ex{Z}=0$ and from \cref{lem:ss12chernoff} it follows that,
\[
\Pro{\lvert Z \rvert \ge \delta} \le 2 \cdot \exp\left( -\frac{\delta^2}{2 \sum_{s=1}^t \sum_{[u:v]\in \M^{(s)} } \left(g_{u,v}^{(s)}\right)^2 }\right) .
\]

From the first statement of \cref{lem:theorem32_new} 
it follows that
$\sum_{s=1}^{t} \sum_{[u:v] \in \M^{(s)}} \left( g_{u,v}^{(s)}\right)^2 \le 2\cdot \|a\|_2^2$.
Combining these two gives that
$\Pro{\lvert Z \rvert \ge \delta} \le 2 \cdot \exp( -{\delta^2}/(4 \cdot \|a\|_2^2))$. This together with \cref{eq:weightedSumLoadbound} implies that,
\[
\Pro{\left| \sum_{w\in V}a_w\cdot X_w^{(t)}-\overline{x} \right| \ge \delta+\kappa} \le  2 \cdot \exp\left( -\frac{\delta^2}{4 \cdot \|a\|_2^2}\right),
\]
and shifting $\delta$ by $\kappa$ completes the proof.
\end{proof}

\subsection{Reducing the Number of Tokens to Linear}

We will now present an application of \cref{ChernoffBound} to bound the number of tokens above the average load; this can be regarded as the number of tokens that need to be rearranged to achieve constant discrepancy. The proof relies on the flexibility of \cref{ChernoffBound} by applying it multiple times for different coefficients $(a_w)_{w \in V}$ in order to prove that there is no large subset in which all nodes have large load.
\begin{lemma}\label{lem:toNToken}
Assume our process applies a sequence of matchings $\left(\M^{(s)}\right)_{s=1}^{t_1}$ which is
$(K, 1/(2n))$-smoothing on an arbitrary initial load vector with $\discr(x^{(0)})\le K$.
Then it holds that 
\begin{align*}
\Pro{ \sum_{w \in V} \max \left\{ X_w^{(t_1)} - \overline{x}, 0 \right\} \leq 16 \cdot n }&\geq 1-2\cdot n^{-2}.
\end{align*}
\end{lemma}
\begin{proof}
For any integer $i=1,2,\ldots,2\cdot \sqrt{\log(n)} $, we define the event
\[
\mathcal{E}_i := \biggl\{ \left| \left\{ u \in V \colon X_u^{(t_1)} - \overline{x} \geq 4\cdot i \right\} \right| \le n \cdot 2^{-i} \biggr\} \qquad \text{and} \qquad \mathcal{E}_0 := \left\{  \max_{w\in V}X_{w}^{(t_1)} -\overline{x} \le 6 \sqrt{\log(n)} \right \}.
\]
From \cref{ChernoffBound}, second statement (with $\delta:=\sqrt{36\log(n)}$, $a_w=1$ and $a_u=0$ for $u\neq w$ and union bound over all nodes $w\in V$) we get
\[
\Pro{ \mathcal{E}_0 } \geq 1 - n^{-2}.
\] 
It remains to analyze the probability of the other $\mathcal{E}_i$. To this end, fix any $1 \leq i \leq 2\sqrt{\log(n)} $. By the union bound,
\begin{align}
\Pro{ \overline{\mathcal{E}_i }} &\leq \sum_{S \subseteq V \colon |S| = n \cdot 2^{-i}}
\Pro{ \bigcap_{u \in S} \left\{ X_u^{(t_1)} \geq \overline{x} + 4 \cdot i \right\} }. \label{eq:union_bound}
\end{align}
To calculate this probability we apply \cref{ChernoffBound} as follows.
We define for a fixed subset $S \subseteq V$ a vector $(a_w)_{w \in V}$ by $a_w := 1/|S|$ for $w \in S$ and $a_w := 0$ otherwise. Then $\sum_{w \in V} a_w =1$ and $\|a\|_2^2 = |S| \cdot (1/|S|)^2 = 1/|S|$, $\kappa=1/n$, and we obtain that
\begin{align*}
\Pro{ \bigcap_{u \in S} \left\{ X_u^{(t_1)} \geq \overline{x} \!+\! 4 \cdot i \right\} } \!&\leq \Pro{\sum_{u \in S} X_u^{(t_1)} \geq |S| \cdot \left( \overline{x} + 4 \cdot i \right) } \leq
\Pro{ \left| \sum_{w \in V} a_w \cdot X_w^{(t_1)} - \overline{x} \right| \geq 4 \cdot i} \\
&\stackrel{(a)}{\leq} 2 \cdot \exp\left( - \frac{ (3 i)^2}{ 4 \| a \|_2^2} \right) \leq 2 \cdot \exp\left( - (9/4) \cdot  i^2 \cdot |S| \right) 
\\&= 2 \cdot \exp\left( - (9/4) \cdot  i^2 \cdot n \cdot 2^{-i} \right),
\end{align*}
where $(a)$ follows from \cref{ChernoffBound}. Plugging this into \cref{eq:union_bound} yields
\begin{align*}
\Pro {\overline{\mathcal{E}_i }} &\leq \binom{n}{n \cdot 2^{-i}} \cdot 2 \cdot \exp\left( - (9/4)\cdot  i^2 \cdot n \cdot 2^{-i} \right) \\&\stackrel{(a)}{\leq} \left( e \cdot 2^{i} \right)^{ n \cdot 2^{-i}} \cdot 2 \cdot \exp\left( - (9/4)\cdot  i^2 \cdot n \cdot 2^{-i} \right)
= 2 \cdot \left( \frac{e \cdot 2^{i}}{ e^{(9/4)\cdot i^2} } \right)^{ n \cdot 2^{-i}} \stackrel{(b)}{\leq} n^{-3},
\end{align*}
where $(a)$ used the estimate $\binom{n}{k} \leq (en/k)^k$ and $(b)$ used that $i \in [1, 2\sqrt{\log(n)} ]$.
By another union bound over $i=0,1,\ldots,2\sqrt{\log(n)} $,
\[
\Pro{ \bigcap_{i=0}^{2\sqrt{\log(n)}} \mathcal{E}_i } \geq 1 - n^{-2} - 2\sqrt{\log(n)} \cdot n^{-3} \geq 1 - 2\cdot n^{-2}.
\]
We define a sequence $a_0 := -\infty$, $a_i:= 4 \cdot i$ for $1 \leq i \leq 2\sqrt{\log(n)}$ and assume the event $\bigcap_{i=0}^{2\sqrt{\log(n)}} \mathcal{E}_i $ occurs and get
\begin{align*}
\MoveEqLeft\sum_{w \in V} \max \left\{ X_w^{(t_1)} - \overline{x}, 0 \right\} \leq
\sum_{i=1}^{2\sqrt{\log(n)}} \left| \left\{ w \in V \colon X_w^{(t_1)} - \overline{x} \in (a_{i-1}, a_{i} ] \right\} \right| \cdot a_i \\
&\leq
\sum_{i=1}^{2\sqrt{\log(n)}} \left| \left\{ w \in V \colon X_w^{(t_1)} - \overline{x} \geq a_{i-1} \right\} \right| \cdot a_i \leq \sum_{i=1}^{\infty} \left( 4n \cdot 2^{-i+1} \cdot i \right)
\leq 8 n\cdot \sum_{i=1}^{\infty} \left( i \cdot 2^{-i} \right)
= 16 n.
\end{align*}
This completes the proof.
\end{proof}

\section{Balancing a Linear Number of Tokens}\label{sec:spreading}

In this section we establish the most challenging step in our analysis. 
We consider an initial load vector $x^{(0)}$ with at most $(\newL-\epsilon) n$ tokens for $\frac{4}{\log^4(n)} <\epsilon<\newL$ and integer $\newL \geq 1$ (this is more general than requiring a linear number of tokens, which is what we will need in the final proof).

The main result of this section is given in \cref{lem:NtokensTo25Height}, which states that after $O(\tauglobal + \taulocal\cdot \log(n)/\log\log(n)$) rounds the maximum load is $L+1$. An illustration of the proof method can be found in \cref{fig:pro_illustration}.

\begin{proposition}
\label{lem:NtokensTo25Height}
Let $t_2 :=2 \cdot \tauglobal+ {6\cdot \taulocal \cdot \log (n)}/{\log\log(n)}$, let $L$ be any integer with $1 \leq L \leq \log^7 (n)$ and let $\frac{4}{\log^4 (n)} \leq \epsilon< \newL$ (not necessarily constant).
Assume our process applies a $(\tauglobal,\taulocal)$-good sequence of matchings $\left(\M^{(s)}\right)_{s=1}^{\infty}$ to an arbitrary initial load vector $x^{(0)}$ with at most $(\newL-\epsilon) n$ tokens.
Then
\[
\Pro{ \max_{w \in V}X_w^{(t_2)} \leq \newL+1 } 
\geq 1-\exp\left(-\frac{\epsilon\cdot \log (n)}{2\cdot \newL\cdot \log\log (n)}
+ 8\cdot\log \log (n)\right) - 2\cdot n^{-2}.
\]
\end{proposition}

The proposition makes uses of \cref{lem:rmlemmaphaseone} and \cref{lem:rmlemmaphasetwo}, which concern Phase 1 and Phase 2, respectively. Both are stated and proved at the end of this section.

\begin{proof}[Proof of \cref{lem:NtokensTo25Height}]
\textbf{Phase 1.} Let $t_1:= \tauglobal + \taulocal\cdot\log(n)/\log \log (n) $ and define $Y^{(t_1)} := \sum_{u \in V} \max\{ X_u^{(t_1)} - \newL , 0 \}$ as the number of tokens at height at least $\newL+1$. Define the event
$
\mathcal{G}_1 := \{Y^{(t_1)} \le {n}/{\log(n)}\}.
$
From \cref{lem:rmlemmaphaseone} (Phase 1 Lemma) it follows that 
\[
\Pro{\mathcal{G}_1} \ge 1-\exp\left(-\frac{\epsilon\cdot \log (n)}{2\cdot \newL\cdot \log\log (n)}
+ 8\cdot\log \log (n)\right)- n^{-2}.
\]
\textbf{Phase 2.} For any $t \geq t_1$, we define an auxiliary load vector $\widehat{X}^{(t)}$ with
$\widehat{X}_u^{(t)}:= \max\{X_u^{(t)}-\newL,0\}$.
From \cref{obs:subAddAppl} it follows that for any $t \geq t_1$,
$X_u^{(t)} \leq \widehat{X}_u^{(t)} + \newL$.
Hence it is sufficient to show that $\max_{u \in V} \widehat{X}_u^{(t_2)} \leq 1$ for some suitable round $t_2$. Note that from  Phase 1 we get  $\sum_{u \in V} \widehat{X}_u^{(t_1)} \leq n/\log(n)$, i.e. the load vector
$\widehat{X}^{(t_1)}$ has at most $n/\log(n)$ tokens.
This time we define, for  any round $t \geq t_1$, 
$\widehat{Y}^{(t)} := \sum_{u\in V} \max\{\widehat{X}_u^{(t)}-1,0\}$,
which is equal to the number of tokens with height at least $2$ in $\widehat{X}^{(t)}$. Let $t_2:=t_1+\tauglobal+4\taulocal\cdot \frac{\log(n)}{-\log(\frac{1}{\log(n)} +\frac{ 2}{\log^4 (n)})}$.
We define a second event
$\mathcal{G}_2 : =\{ \widehat{Y}^{(t_2)} = 0\}$.
From \cref{lem:rmlemmaphasetwo} (Phase 2 Lemma) with $\epsilon:=1-\frac{1}{\log(n)}$ it follows that $\Pro{\mathcal{G}_2 ~|~ \mathcal{G}_1} \ge 1 - n^{-2},$
and by the definition of conditional probability we get
\begin{align*}
\Pro{\mathcal{G}_2} &\ge
\Pro{\mathcal{G}_2 ~|~ \mathcal{G}_1} \cdot \Pro{\mathcal{G}_1} \ge 1- n^{-2} - \Pro{\overline{\mathcal{G}_1}} 
\\& \geq 1 -\exp\left(-\frac{\epsilon\cdot \log (n)}{2\cdot \newL\cdot \log\log (n)}
+ 8\cdot\log \log (n)\right)- 2\cdot n^{-2}.
\end{align*}
From the definition of $t_2$ it follows that
\begin{align*}
    t_2 &= t_1 + \tauglobal +\frac{4\cdot \log (n)}{-\log\left(\frac{1}{\log(n)}+\frac{2}{\log^4 (n)}\right)}\cdot \taulocal \le t_1+ \tauglobal + \frac{5\cdot \log (n)}{\log\log(n)} \cdot \taulocal
    \\& = 2\cdot \tauglobal +  \frac{ 6\cdot\log(n)}{\log\log(n)}\cdot \taulocal,  
\end{align*}
finishing the proof.
\end{proof}
In the rest of this section we will prove \cref{lem:rmlemmaphaseone} and \cref{lem:rmlemmaphasetwo}. To do so, we have to provide several additional definitions and statements, which ultimately leads towards the key lemma (\cref{lem:generallemmaone:RM}), which proves a drop in the number of tokens with height at least $L+1$ within $\taulocal$ rounds. From this,  \cref{lem:rmlemmaphaseone} and \cref{lem:rmlemmaphasetwo} can be derived relatively easily.

First of all, we define two events corresponding to $\tauglobal$ (global mixing) and $\taulocal$ (local mixing).
\begin{definition}\label{taus}
For any round $t\ge \tauglobal$ we define
\[
\Gamma_g^{(t)}:= \left\{
\bigcap_{u \in V} \left\| \M_{u,.}^{[1,t]} - \vec{ \frac{1}{n}} \right\|_2^2 \leq \frac{1}{n^7} \right\}.
\]
\noindent For any round $t\ge \taulocal$ and arbitrary node $u\in V$ we define
\[
\Gamma_{\ell}^{(t)}(u):= \left\{ \left\| 
\M_{u,.}^{[t-\taulocal+1,t]}\right\|_{2}^{2} \le \frac{1}{\log^{10}(n)} \right\}.
\]
\end{definition}

Note that the event $\Gamma_g^{(t)} $ depends on the matching matrices in the time interval $[1,t]$. The event implies that the balancing matrices applied during rounds $[1,t]$ are $(n\cdot \log^{7}(n),1/n)$-smoothing
(see \cref{obs:tauGtauCRelation}).
We call that property \emph{globally mixing}, since it implies that a token starting from any node will reach each node in round $t$ with probability $(1\pm o(1))/n$.
In contrast, the event $\Gamma_{\ell}^{(t)}(u)$ adopts the viewpoint of a single node $u$.
It implies that the matchings chosen during the last $\taulocal$ rounds are locally smoothing from the viewpoint $u$. $\Gamma_{\ell}^{(t)}(u)$  depends on the matchings applied in the time interval $[t-\taulocal+1,t]$. 
We remark that for any round $t \geq 1$, the events $\Gamma_{g}^{(t)}$ and $\Gamma_{\ell}^{(t+\taulocal)}(u)$ are statements over disjoint time-intervals, hence for the random matching model they are independent events.
In the following we first calculate the probabilities that these events occur for a sequence of $(\tauglobal,\taulocal)$-$good$ matchings, the proof is a straightforward calculation.

\begin{lemma}
\label{lem:Gamma:GL}
For any process generating $(\tauglobal,\taulocal)$-$good$ matchings the following holds.
\begin{enumerate}\itemsep0pt
\item For any round $t\ge \tauglobal$,
$\displaystyle
\Pro{ \Gamma_g^{(t)} } \geq 1-1/n^3.
$
\item For any round $t \geq \taulocal$ and arbitrary node $u\in V$,
$\displaystyle
\Pro{ \Gamma_{\ell}^{(t)}(u) } \geq 1 - 1/\log^{11} n.
$ 
\end{enumerate}
\end{lemma}
\begin{proof}
    First note that
\begin{align}
\Pro{\overline{\Gamma_g^{(\tauglobal)}}} &= \Pro{\bigcup_{u \in V} \left\{\left\| \M_{u,.}^{[1,\tauglobal]} - \vec{ \frac{1}{n}} \right\|_2^2 > \frac{1}{n^7}\right\}} . \label{eq:notGammaGlobal}
\end{align}
Applying the definition of $\tauglobal$ (\cref{def:taus}), we get for 
$ t=\tauglobal$ that $\Pro{\overline{\Gamma_g^{(\tauglobal)}}}\le 1/n^3$.
The first statement follows from the basic fact that $ \left\| \M_{u,.}^{[1,t]} - \vec{ \frac{1}{n}} \right\|_2^2$ is non-increasing in $t$ (see 
\cref{monotone}).
For the second statement, consider an arbitrary round $t\ge \taulocal$ and a fixed node $u\in V$. Here, we get
\begin{align*}
\Pro{\overline{\Gamma_{\ell}^{(t)}(u)}} &= \Pro{ \left\{\left\| \M_{u,.}^{[t-\taulocal+1,t]} \right\|_2^2 > \frac{1}{\log^{10}(n)}\right\}} \le \frac{1}{\log^{11}(n)},
\end{align*}
where the last inequality follows from the definition of $\taulocal$ (\cref{def:taus}). As above, the proof follows from \cref{monotone}.
\end{proof}

Recall that  Phase 1 starts at round $\tauglobal$ and each of our two phases is subdivided into $\log (n)/\log\log(n)$ \emph{epochs} of $\taulocal$ many rounds. We refer to the last round of an epoch as \emph{milestone}. For any epoch $k$ we denote by $e: \N_{0} \rightarrow \N_{0}$ the function $e(k):=\tauglobal+k \cdot \taulocal$ which returns the last round of $k$-th epoch, i.e., $e(k)$ returns the $k$-th milestone. We define $E(k):=[e(k-1)+1,\ldots,e(k)]$ for the interval of rounds constituting the $k$-th epoch.

Let the random variable $Y^{(t)}$ denote the number of tokens of height at least $\newL+1$ at round~$t$, i.e.,
\begin{align}
Y^{(t)} := \sum_{u \in V} \max\left\{ X_u^{(t)}- \newL ,0 \right\} = \sum_{j \in \mathcal{T}} \1_{H_j^{(t)} \geq \newL+1}. 
\end{align} We wish to prove that in each epoch $Y^{(t)}$ drops by a constant factor (see \cref{{lem:generallemmaone:RM}}) such that at the end of Phase 1 the number of tokens of height at least $\newL+1$ is at most $n/\log(n)$ (\cref{lem:rmlemmaphaseone}). Since the
the height of the token is non-increasing this number will not increase for the rest of the process.

Let us first focus on the $k$-th milestone, for $k\in \N$. For now, let us fix an arbitrary location vector $w^{(e(k-1))}$. Consider an arbitrary token $i \in \mathcal{T}$ which has height at least $L+1$ at round $e(k-1)$.
Our goal is to prove that the expected number of tokens that collide with $i$ at round $e(k)$ is smaller than $L$.
To this end, we define an indicator random variable $Z_{i,j}^{(e(k)}
(v)$ for any token $j \in \mathcal{T}, j \neq i$ and any node $v \in V$ as
\[ Z_{i,j}^{(e(k))}(v) :=
\1_{ W_i^{(e(k))}=v \ \cap \ W_j^{(e(k))}=v } \quad \text{ and } \quad Z_i^{(e(k))}:=\sum_{v\in V} \sum_{j\in \mathcal{T} \colon j \neq i} Z_{ij}^{(e(k))}(v),
\]
i.e., $Z_{i,j}^{(e(k))}(v)=1$
if and only if tokens $i$ and $j$ are both on node $v$ at the $k$-th milestone.
Hence $Z_i^{(e(k))}$ counts the number of tokens colliding with token $i$ at that time.

The next lemma bounds the expected number of collisions.
Note that we assume in the lemma that the applied matrices are fixed. The randomness is due to the random decisions in the shuffling step.

\medskip
\begin{lemma}
\label{lem:collision}
Let $\left(\M^{(t)}\right)_{t=1}^{e(k)}$ be an arbitrary but fixed sequence of matchings.
For any milestone $k \in \N$ and any token $i \in \mathcal{T}$
it holds that
\begin{align*}
\MoveEqLeft
\Exco{ Z_i^{(e(k))}}{W^{(e(k-1))},\left(\M^{(t)}\right)_{t=1}^{e(k)}}
\le
\sum_{w\in V} \left(\sum_{v\in V} \M_{W_i^{(e(k-1))},v}^{[e(k-1)+1,e(k)]}\cdot \M_{w,v}^{[e(k-1)+1,e(k)]}\right) \cdot
X_w^{(e(k-1))}.
\end{align*}
\end{lemma}
\begin{proof}
First we focus on a fixed location vector $w^{(e(k-1))}$ and compute,
\begin{align}
\MoveEqLeft\Ex{Z_i^{(e(k))}~\Big|~ W^{(e(k-1))}= w^{(e(k-1))}, \left(\M^{(t)}\right)_{t=1}^{e(k)} }\notag \\& \stackrel{(a)}{=} \sum_{v \in V} \sum_{j\in \mathcal{T} \colon j \neq i}\E\left[Z_{ij}^{(e(k))}(v)~\Big|~W^{(e(k-1))}= w^{(e(k-1))},\left(\M^{(t)}\right)_{t=1}^{e(k)} \right] \notag \\ &\stackrel{(b)}{=}
\sum_{v \in V} \sum_{j\in \mathcal{T} \colon j \neq i}\Pro{Z_{ij}^{(e(k))}(v) = 1
~\Big|~ W^{(e(k-1))}= w^{(e(k-1))},\left(\M^{(t)}\right)_{t=1}^{e(k)} } \notag \\
&= \sum_{v \in V} \sum_{j\in \mathcal{T} \colon j \neq i} \Pro{W_i^{(e(k))} = v \cap W_j^{(e(k))} = v~\Big|~ W^{(e(k-1))}=w^{(e(k-1))},\left(\M^{(t)}\right)_{t=1}^{e(k)} },\label{eq:8:lem4.7:RM}
\end{align}
where $(a)$ uses linearity of conditional expectation and $(b)$ uses the fact that the $Z_{i,j}^{(e(k))}(v)$ are indicator random variables. Recall the definition $E(k)=[e(k-1)+1,\ldots,e(k)]$. Crucially, we can now apply our negative correlation result (\cref{lem:height}) to \cref{eq:8:lem4.7:RM} to obtain
\begin{align}
\MoveEqLeft \E\left[ Z_i^{(e(k))}~\Big|~ W^{(e(k-1))}= w^{(e(k-1))},\left(\M^{(t)}\right)_{t=1}^{e(k)} \right] \notag \\ &\leq \sum_{v \in V} \sum_{j\in \mathcal{T} \colon j \neq i} \Pro{ W_i^{(e(k))}=v~\Big|~ W^{(e(k-1))}=w^{(e(k-1))},\left(\M^{(t)}\right)_{t=1}^{e(k)} }\notag
\\& \qquad\qquad \qquad \cdot \Pro{W_j^{(e(k))}=v~\Big|~ W^{(e(k-1))}=w^{(e(k-1))},\left(\M^{(t)}\right)_{t=1}^{e(k)} } \notag \\
&= \sum_{v \in V} \M_{w_i^{(e(k-1))},v}^{E(k)} \cdot
\sum_{j\in \mathcal{T}} \M_{w_j^{(e(k-1))},v}^{E(k)} \notag \\
&=
\sum_{v \in V} \M_{w_i^{(e(k-1))},v}^{E(k)} \cdot \sum_{w \in V} \M_{w,v}^{E(k)} \cdot
x_w^{(e(k-1))}
\notag \\
&=
\sum_{w\in V} \left(\sum_{v\in V} \M_{w_i^{(e(k-1))},v}^{E(k)}\cdot \M_{w,v}^{E(k)}\right) \cdot x_w^{(e(k-1))}.\notag
\end{align}

Since the above estimate holds for all location vectors $w^{(e(k-1))}$, it follows 
\begin{align*}
\MoveEqLeft
\Ex{ Z_i^{(e(k))} ~\Big|~ W^{(e(k-1))},\left(\M^{(t)}\right)_{t=1}^{e(k)}}
\le
\sum_{w\in V} \left(\sum_{v\in V} \M_{W_i^{(e(k-1))},v}^{E(k)}\cdot \M_{w,v}^{E(k)}\right) \cdot X_w^{(e(k-1))}. 
\qed
\end{align*}
    
\end{proof}
We now define two more events which we use to track the decrease of $Y^{(t)}$ from epoch to epoch.

\begin{definition}\label{eq:LEvent}
Let $\frac{4}{\log^4 (n)} \leq \epsilon<\newL$. Fix a node $u$ and let $i$ be a token located on $u$ at milestone $e(k-1)$.

\begin{enumerate}
\item $ \displaystyle \cL^{(e(k))}(u):= \left\{ \Ex{ Z_i^{(e(k))} ~\Big|~ W^{(e(k-1))},\Gamma_{\ell}^{(e(k))}(u) } \leq \newL-\epsilon + \frac{1}{\log^4(n)}\right\} $.
\item
$\displaystyle
\cE^{(e(k))}(u):=\left\{
\Gamma_{g}^{(e(k-1))} \cap \left( \overline{\Gamma_{\ell}^{(e(k))}(u)}
\cup
\cL^{(e(k))}(u) \right)\right\}
\quad \text{and} \quad \cE^{(e(k))}:= \bigcap_{u\in V} \cE^{(e(k))}(u)
$.
\end{enumerate}
\end{definition}

The idea behind the events defined above is as follows (see also \cref{fig:events} for an illustration).
In the definition of $\cL^{(e(k))}(u)$ we condition on $\Gamma_{\ell}^{(e(k))}(u) $ which means that, from the viewpoint of node $u$ (or token $i$), the randomly chosen matchings in epoch $k$ ensure that token $i$ mixes ``locally''. If now the matchings chosen during the time interval $[1,e(k-1)]$ also suffice for a ``global mixing'', then the expected number of tokens colliding with token $i$ on node $v$ (the location of token $i$ at time $e(k)$) is less than $L$.
The event $\cE^{(e(k))}(u)$ occurs when the matchings chosen during the time interval $[1,e(k-1)]$ are globally mixing and the last epoch $k$ was ``locally mixing'' for token $i$.

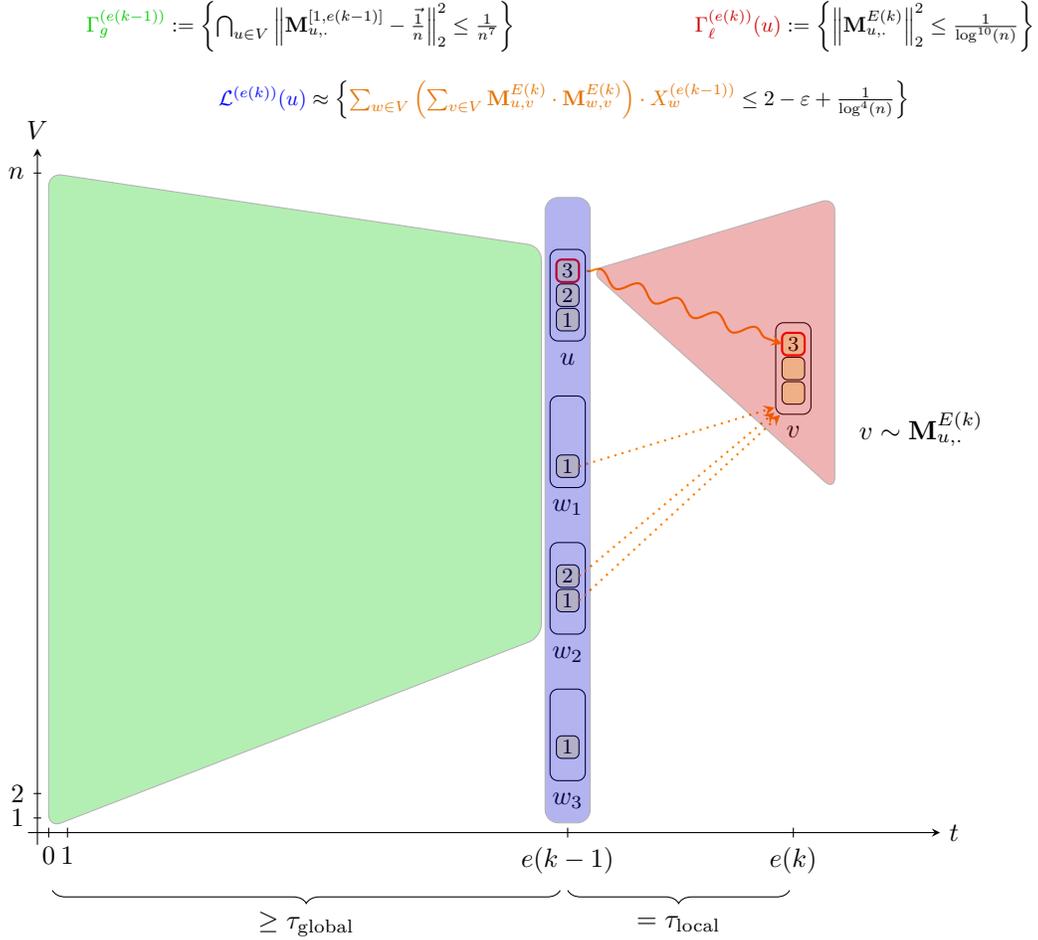
\begin{figure}[H]
\begin{center}
\begin{tikzpicture}[xscale=0.5,yscale=0.65,
win/.style={fill=green,opacity=0.4},
los/.style={fill=red,opacity=0.4},
r/.style={draw=black},
knoten/.style={rectangle,yscale=2.6,rounded corners=3pt,scale=2,draw=black, fill=white}]

\draw[-stealth] (0,-0.25) to node[pos=1.0,above] {$V$} (0,14);
\draw[-stealth] (-0.25,0) to node[pos=1.0,right] {$t$} (24,0);

\draw (0.3,-0.1) to node[pos=0.0,below] {$0$} (0.3,0.1);
\draw (0.8,-0.1) to node[pos=0.0,below] {$1$} (0.8,0.1);

\draw (-0.1,0.3) to node[pos=0.0,left] {$1$} (0.1,0.3);
\draw (-0.1,0.8) to node[pos=0.0,left] {$2$} (0.1,0.8);
\draw (-0.1,13.5) to node[pos=0.0,left] {$n$} (0.1,13.5);

\draw (14.1,-0.1) to node[pos=0.0,below] {$e(k-1)$} (14.1,0.1);

\draw (20.1,-0.1) to node[pos=0.0,below] {$e(k)$} (20.1,0.1);

\draw [decorate,decoration={brace,amplitude=5pt,mirror,raise=4ex}]
  (0.4,-0.25) -- (13.9,-0.25) node[midway,yshift=-3em]{$\geq \tauglobal$};
  \draw [decorate,decoration={brace,amplitude=5pt,mirror,raise=4ex}]
  (14.1,-0.25) -- (20,-0.25) node[midway,yshift=-3em]{$=\taulocal$};

\node[knoten] at (14.1,11) [label=below:$u$]{};
\token {14.1}{10.5}{1} 
\token {14.1}{11}{2} 
\ttoken {14.1}{11.5}{3} 

\node[knoten] at (14.1,8) [label=below:$w_1$]{};
\token {14.1}{7.5}{1} 

\node[knoten] at (14.1,5) [label=below:$w_2$]{};
\token {14.1}{4.75}{1} 
\token {14.1}{5.25}{2} 

\node[knoten] at (14.1,2) [label=below:$w_3$]{};
\token {14.1}{1.75}{1}

\node[knoten] at (20.1,9.5) [label=below:$v$]{};
\ttoken {20.1}{10}{3} 
\token {20.1}{9.5}{ } 
\token {20.1}{9}{ } 

\draw [-stealth,orange,thick,decorate,decoration=snake, segment length=5.5mm] (14.6,11.5) -- (19.8,10);

\draw [dotted,thick,orange,-stealth] (14.4,7.5) to (19.6,8.7);
\draw [dotted,thick,orange,-stealth] (14.4,5.25) to (19.6,8.6);
\draw [dotted,thick,orange,-stealth] (14.4,4.75) to 
(19.75,8.55);

\filldraw[fill=red!80!black, opacity=0.3, rounded corners=5pt] (14.7,11.5) to (21.2,13) to (21.2,7)  -- cycle;

\filldraw[fill=green!80!black, opacity=0.3, rounded corners=5pt] (13.4,12) to (0.3,13.5) to (0.3,0.1) to (13.4,4) -- cycle;

\filldraw[fill=blue!80!black, opacity=0.3, rounded corners=5pt] (13.5,0.2) rectangle (14.7,13);

\node[scale=0.8] at (22,16.5) {$ \textcolor{red!80!black}{\Gamma_{\ell}^{(e(k))}(u)}:= \left\{  \left\| \M_{u,.}^{E(k)} \right\|_2^2 \leq \frac{1}{\log^{10}(n)} \right\}$};

\node[scale=0.8] at (14,15) {$ \textcolor{blue}{\mathcal{L}^{(e(k))}(u)} \approx \left\{
\textcolor{orange!90!black}{\sum_{w\in V} \left(\sum_{v\in V} \M_{u,v}^{E(k)}\cdot \M_{w,v}^{E(k)}\right) \cdot X_w^{(e(k-1))}} \leq 2 - \epsilon + \frac{1}{\log^4(n)} \right\}$};

\node[scale=0.8] at (7,16.5) {$\textcolor{green!80!black}{\Gamma_{g}^{(e(k-1))}} := \left\{  \bigcap_{u \in V} \left\| \M_{u,.}^{[1,e(k-1)]} - \vec{\frac{1}{n}} \right\|_2^2 \leq \frac{1}{n^7} \right\}$};

\node () at (23.5,9) [label=below:$v \sim \M_{u,.}^{E(k)}$] {};

\end{tikzpicture}

\caption{Illustration of the events $\textcolor{red!80!black}{\Gamma_{\ell}^{(e(k)}(u)}$, $\textcolor{green!80!black}{\Gamma_{g}^{(e(k-1))}}$ and $ \textcolor{blue}{\mathcal{L}^{(e(k))}(u)}$, where $L=2$. The token $i$, marked in red, is located at node $u$ in round $e(k-1)$ and has height $L+1=3$. 
Its location $v$ in round $e(k)$ is random and chosen according to $\M_{u,.}^{E(k)}$. In order to keep its height $L+1=3$, there must be at least $L=2$ other tokens on $v$ in round $e(k)$. \newline
Roughly speaking, the event $ \textcolor{blue}{\mathcal{L}^{(e(k))}(u)}$ corresponds to a certain upper bound on a convex combination of loads at round $e(k-1)$ holding; this event depends both on the matchings and shuffling decisions in $[0,e(k-1)]$ and the matchings in $E(k)$.
The coefficients of the load vector are collision probabilities so that the \textcolor{orange}{left-hand side} equals the expected number of tokens that reach the (random) location of token $i$ in round $e(k)$.
\newline
We prove in \cref{lem:crucialfirststatement} that if $\textcolor{red!80!black}{\Gamma_{\ell}^{(e(k))}(u)}$ and $\textcolor{green!80!black}{\Gamma_{g}^{(e(k-1)}}$ occur, then also $ \textcolor{blue}{\mathcal{L}^{(e(k))}(u)}$ occurs (with high probability). Then, conditional on this implication, we prove in our key lemma (\cref{lem:generallemmaone:RM}) that the number of tokens with height $L+1=3$ drops within $E(k)$ by a suitable factor.
}\label{fig:events}
\end{center}
\end{figure}

\begin{lemma}
\label{lem:crucialfirststatement}
Assume that the load vector $x^{(0)}$ has at most $(\newL-\epsilon) n$ tokens for $1 \leq \newL \leq \log^7 (n)$ being an integer and any $0<\epsilon<\newL$. Then for any milestone $k \geq 1$ we have $\Pro{ \cE^{(e(k))}}\geq 1- \frac{2}{n^3}$.
\end{lemma}

\begin{proof}
Consider a node $u\in V$; recall token $i$ is located on $u$ at round $e(k-1)$.  First let us negate the two events from \cref{eq:LEvent} which gives 
\begin{align*}
 \overline{\cL^{(e(k))}(u)} &= \left\{ \Ex{ Z_i^{(e(k))} ~\Big|~ W^{(e(k-1))},\Gamma_{\ell}^{(e(k))}(u) }> \newL-\epsilon + \frac{1}{\log^4(n)}\right\}\quad\mbox{and}\\
\overline{\cE^{(e(k))}(u)} &= \left\{
\overline{\Gamma_{g}^{(e(k-1))}} \cup \left( \Gamma_{\ell}^{(e(k))}(u)
\cap
\overline{\cL^{(e(k))}(u)} \right)\right\}.
\end{align*}
For brevity, let us write $\alpha:=\Pro{\overline{\Gamma_{g}^{(e(k-1))}}}$ and note that $\alpha \leq 1/n^3$ by \cref{lem:Gamma:GL}. Then,
\begin{align}
& \Pro{\overline{\cE^{(e(k))}(u)} } \notag \\
&\leq \alpha+\Pro{\Gamma_{\ell}^{(e(k))}(u) \cap \overline{\cL^{(e(k))}(u)}} 
\notag \\
&\leq \alpha+\Pro{\overline{\cL^{(e(k))}(u)}}  \notag\\
&= \alpha+\Pro{
\Ex{ Z_i^{(e(k))}~\Big|~W^{(e(k-1))},\Gamma_{\ell}^{(e(k))}(u) } > \newL -\epsilon + \frac{1}{\log^4(n)}
}  \notag\\
&=\alpha+\mathbf{Pr} \left[
\Ex{ Z_i^{(e(k))}~\Big|~W^{(e(k-1))},\Gamma_{\ell}^{(e(k))}(u), \Gamma_{g}^{(e(k-1))}} \cdot \Pro{\Gamma_{g}^{(e(k-1))}} + \right. \notag\\
& \hspace*{11ex}\left. \Ex{ Z_i^{(e(k))}~\Big|~W^{(e(k-1))},\Gamma_{\ell}^{(e(k))}(u) , \overline{\Gamma_{g}^{(e(k-1))}}}\cdot  \Pro{\overline{\Gamma_{g}^{(e(k-1))}}}> \newL -\epsilon + \frac{1}{\log^4(n)} \right] \notag
\\& \leq \alpha+\Pro{
\Ex{ Z_i^{(e(k))}~\Big|~W^{(e(k-1))},\Gamma_{\ell}^{(e(k))}(u), \Gamma_{g}^{(e(k-1))}} > \newL -\epsilon + \frac{1}{\log^4(n)} - \frac{\log^7 (n)}{n^2}}, \notag
\end{align}
\noindent
where the last inequality holds since first $\Pro{\Gamma_{g}^{(e(k-1))}} \leq 1$,
and secondly, since (deterministically) $Z_i^{(e(k))}\le \newL \cdot n \leq \log^7(n) \cdot n$ we get,
\[
  \Ex{ Z_i^{(e(k))}~\Big|~W^{(e(k-1))},\Gamma_{\ell}^{(e(k))}(u) , \overline{\Gamma_{g}^{(e(k-1))}}}\cdot  \Pro{\overline{\Gamma_{g}^{(e(k-1))}}} \le  \frac{\log^7 (n)}{n^2}.
\]
We thus have
\begin{align}
\lefteqn{ 
\Pro{\overline{\cE^{(e(k))}(u)} } } \notag \\
&\leq
\alpha + \Pro{ 
\Exco{ Z_i^{(e(k))}}{W^{(e(k-1))},\Gamma_{\ell}^{(e(k))}(u), \Gamma_{g}^{(e(k-1))}} > \newL -\epsilon + \frac{1}{\log^4(n)} - \frac{\log^7 (n)}{n^2}}. \label{eq:E(e(k)):bound}
\end{align}

To upper bound \cref{eq:E(e(k)):bound} we apply  \cref{cl:good} below. 
Since \cref{cl:good} holds for any fixed sequence
$\left(\M^{(t)}\right)_{t=1}^{e(k)}$ satisfying $\Gamma_{\ell}^{(e(k))}(u) \cap \Gamma_{g}^{(e(k-1))}$, we have
\[
\Pro{\Ex{ Z_i^{(e(k))}~\Big|~W^{(e(k-1))},\Gamma_{\ell}^{(e(k))}(u),\Gamma_g^{(e(k-1))} } > \newL -\epsilon + \frac{1}{2\log^4(n)} } \le \frac{1}{n^5}.
\]
Applying the union bound twice
implies that 
\[
\Pro{\overline{\mathcal{E}^{(e(k))}}} = 
\Pro{ \overline{\Gamma_g^{(e(k-1))}} \cup \bigcup_{u \in V} \left(\Gamma_{\ell}^{(e(k))}(u) \cap \overline{\cL^{(e(k))}(u)} \right) } \leq \alpha + n \cdot \frac{1}{n^5} \leq
\frac{2}{n^3}.
\]
\end{proof}

\begin{claim}\label{cl:good}
Fix a node $u$ and let $\left(\M^{(t)}\right)_{t=1}^{e(k)} $ be a fixed sequence of matchings
satisfying
$$\Gamma_{\ell}^{(e(k))}(u) \cap \Gamma_{g}^{(e(k-1))}.$$  Then
\begin{align*}
\Pro{\left\{ \Ex{ Z_i^{(e(k))}~\Big|~W^{(e(k-1))},\left(\M^{(t)}\right)_{t=1}^{e(k)} } > \newL -\epsilon + \frac{1}{2\log^4(n)}\right\} } \le \frac{1}{n^5}.
\end{align*}
\end{claim}
\begin{proof}
As before, $E(k)=[e(k-1)+1,\ldots,e(k)]$ is the interval of rounds constituting the $k$-th epoch.
Let $W_i^{(e(k-1))}=u$ and note
\begin{align}
\MoveEqLeft \Pro{\left\{ \Ex{ Z_i^{(e(k))}~\Big|~W^{(e(k-1))},\left(\M^{(t)}\right)_{t=1}^{e(k)} } > \newL -\epsilon + \frac{1}{2\log^4(n)}\right\} } \notag\\
& \stackrel{(a)}{\leq} \Pro{\sum_{w\in V} \left( \sum_{v\in V} \M _{u,v}^{E(k)}\cdot \M_{w,v}^{E(k)}\right)\cdot X_w^{(e(k-1))} \geq \newL-\epsilon + \frac{1}{2\log^4(n)} }, \label{eq:NotLu:bound}
\end{align}
where $(a)$ follows from \cref{lem:collision}.
To bound the probability of \cref{eq:NotLu:bound} we will apply our concentration inequality (\cref{ChernoffBound}). Note that, since the matchings are fixed, the only randomness remaining in \cref{ChernoffBound} are the movements of tokens in the shuffling step.
We define an $n$-dimensional vector $\left(a_w\right)_{w\in V}$ as
$a_w:= \sum_{v\in V} \M_{u,v}^{E(k)}\cdot \M_{w,v}^{E(k)}$.
To apply \cref{ChernoffBound} we need to (1) show the time interval $[1,e(k-1)]$ is $(K,1/n)$-smoothing for some $K \le n\cdot \newL \le n\cdot \log^{7}(n)$, (2) show $\|a\|_1=1$ and (3) compute $\|a\|_2^2$.
Since the sequence of matchings $\left(\M^{(t)}\right)_{t=1}^{e(k-1)}$ satisfies the event $\Gamma_g^{(e(k-1))}$, it follows from \cref{obs:tauGtauCRelation} that this sequence is $(n\cdot \log^7 (n),1/n)$-smoothing.
Since the matrix $\M^{E(k)}:= \prod_{s=e(k-1)+1}^{e(k)} \M^{(s)}$ is doubly stochastic we get
\begin{align*}
\left\|a\right\|_1\!=\!\sum_{w\in V}a_w &= \sum_{w\in V}\sum_{v\in V} \M_{u,v}^{E(k)}\cdot \M_{w,v}^{E(k)} = \sum_{v\in V}\M_{u,v}^{E(k)}\cdot \sum_{w\in V}\M_{w,v}^{E(k)} = \sum_{v\in V}\M_{u,v}^{E(k)} = 1.
\end{align*}
Furthermore,
\begin{align*}
\left\|a \right\|_2^2 &= \sum_{w\in V} \left(\sum_{v\in V} \M_{u,v}^{E(k)}\cdot \M_{w,v}^{E(k)} \right)^2 \leftstackrel{(a)}{\le}
\left\|\M_{u,.}^{E(k)}\right\|_2^2 
\leftstackrel{(b)}{\leq} \frac{1}{\log^{10}(n)},
\end{align*}
where $(a)$ follows from the second statement of \cref{lem:theorem32_new}; $(b)$ follows from the definition of the event $\Gamma_{\ell}^{(e(k))}(u)$.
Since we have at most $(\newL -\epsilon) n$ many tokens the average load $\xbar$ satisfies $\xbar\le \newL -\epsilon$.
From \cref{ChernoffBound} with  $t:=e(k-1)$, $\kappa:=1/n$ 
, $\delta:=1/(2\log^{4}(n))$, and the bound on
$\|a\|_2^2$ from above, we have 
\begin{align}
\Pro{\sum_{w\in V} a_w\cdot X_w^{(e(k-1))} \geq  \newL -\epsilon +\frac{1}{2\log^4(n)}} \leq 2 \cdot \exp\left( -\frac{ \left(\frac{1}{2\log^4(n)} - \frac{1}{n}\right)^2}{ 4 \cdot \left( \frac{1}{\log^{10}(n)} \right) } \right) \le \frac{1}{n^5} \label{eq:lemma_e1-RM}.
\end{align}
Combining \cref{eq:NotLu:bound,eq:lemma_e1-RM} finishes the proof of the claim. 
\end{proof}

Next we define two random variables which will be used in the remainder of the proof.
Recall that $\mathcal{T}$ is the set of tokens, $Z_i^{(e(k))}$ counts the number of tokens colliding with token $i$ at the $k$-th milestone and for any $t \geq 0$, $Y^{(t)}$ is the number of tokens with height at least $\newL+1$ at round $t$. For any milestone $k\in \N$ we define, 
\begin{align}
\tilde{Y}^{(e(k))} := Y^{(e(k))} \cdot \1_{\cap_{i=0}^{k} \mathcal{E}^{(e(i))}}. \label{eq:ytilde}
\end{align}
Later we will show how to use $Z_i^{(e(k))}$ to bound $\tilde{Y}^{(e(k))}$, and then eventually $Y^{(e(k))}$.
In the next lemma we bound the expected value of $\tilde{Y}^{(e(k))}$.
We will denote by $\left(\mathfrak{F}^{(t)}\right)_{t\ge 0}$ the filtration of the random process; note that in particular, $\mathfrak{F}^{(t)}$ determines not only the current load vector and previous load vectors $X^{(t)},X^{(t-1)},\ldots, X^{(1)}$, but also all location vectors $W^{(t)}, W^{(t-1)},\ldots, W^{(1)}$. In case of randomly generated matchings $\left(\M^{(s)}\right)_{s=1}^{t}$ is determined by $\left(\mathfrak{F}^{(t)}\right)_{t\ge 0}$, too. In the following, to keep the notation brief, we use the convention that for any round $t \geq 0$, any random variable $X$ and event $\mathcal{E}$,
\[
 \Exf{t}{ X } := \Ex{ X ~\Big|~ \mathfrak{F}^{(t)}} \quad \mbox{ and } \quad   \Prof{t}{ \mathcal{E} } := \Pro{ \mathcal{E} ~\Big|~ \mathfrak{F}^{(t)}}.
\]

\begin{lemma}[Key Lemma -- Expected Drop in one Epoch]
\label{lem:generallemmaone:RM}
Assume that the initial load vector $x^{(0)}$ has at most $(\newL-\epsilon) n$ tokens for $1 \leq \newL \leq \log^7 n$ being an integer and $0<\epsilon < 1$. 
Then for any milestone $k\ge 1$,
\[
\Exf{e(k-1)}{ \tilde{Y}^{(e(k))}} \le
\left( 1 - \frac{\epsilon}{\newL } + \frac{2}{\newL \cdot \log^4 n} \right)
\cdot \tilde{Y}^{(e(k-1))}.
\]
\end{lemma}
We emphasize that in the above lemma, it is possible to have an $\epsilon$ that depends on $n$.

\begin{proof}
    Using the definition of $\tilde{Y}^{(e(k))}$ we get
\begin{align}
\lefteqn{ \Exf{e(k-1)}{\tilde{Y}^{(e(k))} } } \notag \\ &=
\Exf{e(k-1)}{ Y^{(e(k))} \cdot \1_{\cap_{i=0}^{k} \mathcal{E}^{(e(i))}} }\notag \\
&\stackrel{(a)}{=} \1_{\cap_{i=0}^{k-1} \mathcal{E}^{(e(i))}} \cdot \Exf{e(k-1)}{ Y^{(e(k))} \cdot
\1_{ \mathcal{E}^{(e(k))}} } \notag \\
&\stackrel{(b)}{=} \1_{\cap_{i=0}^{k-1} \mathcal{E}^{(e(i))}} \cdot \biggl( \Prof{e(k-1)}{ \Gamma_g^{(e(k-1))}} \cdot
\Exf{e(k-1)}{ Y^{(e(k))} \cdot \1_{ \mathcal{E}^{(e(k))}} ~\Big|~ \Gamma_g^{(e(k-1))} } \notag \\ &\quad \qquad \quad \qquad \mbox{} \quad +
\Prof{e(k-1)}{ \overline{\Gamma_g^{(e(k-1))}}} \cdot
0
\biggr) \notag\\
&\leq \1_{\cap_{i=0}^{k-1} \mathcal{E}^{(e(i))}} \cdot \Exf{e(k-1)}{ Y^{(e(k))} \cdot
\1_{ \mathcal{E}^{(e(k))} }~\Big|~ \Gamma_g^{(e(k-1))} }, \label{eq:last_with_y}
\end{align}
where $(a)$ is an application of the ``take-out-what-is-known'' rule in conditional expectations, since
$\mathfrak{F}^{(e(k-1))} $ determines the random variable $\1_{\cap_{i=0}^{k-1} \mathcal{E}^{(e(i))}}$ (\cref{lem:take-out-what-is-known}), and $(b)$ used that $\overline{ \Gamma_g^{(e(k-1))}}$ implies $\1_{\mathcal{E}^{(e(k))}}=0$.
Further, we have
\begin{align}
Y^{(e(k))} = \sum_{j\in \mathcal{T}} \1_{H_j^{(e(k))} \ge \newL+1} \leq
\sum_{j \in \mathcal{T}} \1_{H_j^{(e(k-1))} \geq \newL+1} \cdot \1_{Z_j^{(e(k))} \geq \newL },
\label{eq:y_z_relation-RM}
\end{align}
where the equality follows from the definition of $Y^{(e(k)}$.
To see the inequality, observe that by properties of the height-sensitive process the height of a token never increases and therefore, in order for a token $j$ to be at height at least $L+1$ in round $e(k)$, it must have had height at least $L+1$ in round $e(k-1)$ and there must be at least $L$ other tokens at its location in round $e(k)$.

Applying \cref{eq:y_z_relation-RM} to \cref{eq:last_with_y},
\begin{align}
\MoveEqLeft \Exf{e(k-1)}{ \tilde{Y}^{(e(k))} } \notag \\
&\leq \1_{\cap_{i=0}^{k-1} \mathcal{E}^{(e(i))}} \cdot \Exf{e(k-1)}{ \sum_{j \in \mathcal{T}} \1_{H_j^{(e(k-1))} \geq \newL+1} \cdot \1_{Z_j^{(e(k))}\ge \newL} \cdot
\1_{ \mathcal{E}^{(e(k))} }~\Big|~   \Gamma_g^{(e(k-1))} } \notag \\
&\stackrel{(a)}{=}
\1_{\cap_{i=0}^{k-1} \mathcal{E}^{(e(i))}} \cdot \sum_{j \in \mathcal{T}} \Exf{e(k-1)}{ \1_{H_j^{(e(k-1))} \geq \newL+1} \cdot \1_{Z_j^{(e(k))}\ge \newL} \cdot
\1_{ \mathcal{E}^{(e(k))} }~\Big|~   \Gamma_g^{(e(k-1))} } \notag
\\
&\stackrel{(b)}{=} \1_{\cap_{i=0}^{k-1} \mathcal{E}^{(e(i))}} \cdot \sum_{j \in \mathcal{T}} \1_{H_j^{(e(k-1))} \geq \newL+1} \cdot \Exf{e(k-1)}{ \1_{Z_j^{(e(k))}\ge \newL} \cdot
\1_{ \mathcal{E}^{(e(k))} }~\Big|~   \Gamma_g^{(e(k-1))} }, \label{eq:almost_finished}
\end{align}
where $(a)$ holds by linearity of conditional expectation and $(b)$ holds since $\mathfrak{F}^{(e(k-1))}$ reveals $\1_{H_j^{(e(k-1))}\ge \newL+1}$ (``take-out-what-is-known'', \cref{lem:take-out-what-is-known}).
To simplify the notation we will write $u$ instead of $w_j^{(e(k-1))}$ in what follows.
Conditioning on whether $\Gamma_{\ell}^{(e(k))}(u)$ holds, we can bound the expectation from \cref{eq:almost_finished} by
\begin{align*}
\MoveEqLeft \Exf{e(k-1)}{ \1_{Z_{j}^{(e(k))}\ge \newL} \cdot
\1_{\mathcal{E}^{(e(k))}} \cond  \Gamma_g^{(e(k-1))} } \\
&= \Prof{e(k-1)}{ \overline{\Gamma_{\ell}^{(e(k))}(u)} ~\Big|~ \Gamma_g^{(e(k-1))}} \cdot \Exf{e(k-1)}{ \1_{Z_{j}^{(e(k))}\ge \newL} \cdot
\1_{\mathcal{E}^{(e(k))}} ~\Big|~  \Gamma_g^{(e(k-1))}, \overline{ \Gamma_{\ell}^{(e(k))}(u)} } \\
&\quad + \Prof{e(k-1)}{ \Gamma_{\ell}^{(e(k))}(u) ~\Big|~ \Gamma_g^{(e(k-1))}} \cdot \Exf{e(k-1)}{ \1_{Z_{j}^{(e(k))}\ge \newL} \cdot
\1_{\mathcal{E}^{(e(k))}} ~\Big|~  \Gamma_g^{(e(k-1))}, \Gamma_{\ell}^{(e(k))}(u) } \\
&\leq \Prof{e(k-1)}{ \overline{\Gamma_{\ell}^{(e(k))}(u)} ~\Big|~ \Gamma_g^{(e(k-1))}} \cdot 1 \!+\! 
1 \cdot \Exf{e(k-1)}{ \1_{Z_{j}^{(e(k))}\ge \newL} \cdot
\1_{\mathcal{E}^{(e(k))}} ~\Big|~  \Gamma_g^{(e(k-1))}, \Gamma_{\ell}^{(e(k))}(u) }  .
\end{align*}
We upper bound the two remaining terms in the last line separately using the following two claims, whose proofs are given after the conclusion of the proof of this lemma.
\begin{claim}\label{clm:abcd_one}
$\displaystyle\Prof{e(k-1)}{ \overline{\Gamma_{\ell}^{(e(k))}(u)} ~\Big|~ \Gamma_g^{(e(k-1))}}
\leq
\frac{1}{\log^{11} (n) }
$.
\end{claim}

\begin{claim}
\label{clm:abcdtwo}
$\displaystyle\Exf{e(k-1)}{ \1_{Z_{j}^{(e(k))}\ge \newL} \cdot
\1_{\mathcal{E}^{(e(k))}} ~\Big|~  \Gamma_g^{(e(k-1))}, \Gamma_{\ell}^{(e(k))}(u) }
\leq
1 - \frac{\epsilon}{\newL} + \frac{1}{\newL \cdot \log^4 (n)}
$.
\end{claim}

Together \cref{clm:abcd_one} and \cref{clm:abcdtwo} yield, using the assumption that $L \leq \log^7(n)$,
\[
\Exf{e(k-1)}{ \1_{Z_{j}^{(e(k))}\ge \newL} \cdot
\1_{\mathcal{E}^{(e(k))}} ~\Big|~  \Gamma_g^{(e(k-1))} }
\leq
1 - \frac{\epsilon}{\newL} + \frac{1}{\newL \cdot \log^4 (n)} + \frac{1}{\log^{11} (n)}
\leq 1 - \frac{\epsilon}{\newL} + \frac{2}{\newL \cdot \log^4 (n)}.
\]
Applying this to~\cref{eq:almost_finished} gives us
\begin{align*}
\Exf{e(k-1)}{ \tilde{Y}^{(e(k))}   } &\leq \1_{\cap_{i=1}^{k-1} \mathcal{E}^{(e(i))}} \cdot \sum_{j\in \mathcal{T}} \1_{H_j^{(e(k-1))} \geq \newL+1} \cdot \left( 1 - \frac{\epsilon}{\newL} + \frac{2}{\newL \cdot \log^4 (n)} \right) \notag\\
&= \1_{\cap_{i=1}^{k-1} \mathcal{E}^{(e(i))}} \cdot Y^{(e(k-1))} \cdot \left( 1 - \frac{\epsilon}{\newL} + \frac{2}{\newL \cdot \log^4 (n)} \right)
\\& = \tilde{Y}^{(e(k-1))}\cdot \left( 1 - \frac{\epsilon}{\newL} + \frac{2}{\newL \cdot \log^4 (n)} \right). \qed
\end{align*}
\end{proof}

\begin{proof}[Proof of \cref{clm:abcd_one}]
As before we define $E(k)=[e(k-1)+1,\ldots,e(k)]$ for the interval of rounds constituting the $k$-th epoch.
Note that
$\Gamma_{\ell}^{(e(k))}(u)$ depends only on the matchings in the time interval $E(k)$. Moreover, as $\Gamma_g^{(e(k-1))}$ and $\Gamma_{\ell}^{(e(k))}(u)$ refer to disjoint time-intervals, we have 
\[
\Prof{e(k-1)}{\overline{\Gamma_{\ell}^{(e(k))}(u)} ~\Big|~ \Gamma_g^{(e(k-1))} } = \Prof{e(k-1)}{\overline{\Gamma_{\ell}^{(e(k))}(u)}} \stackrel{(a)}{\leq} \frac{1}{\log^{11}(n)},
\]
where $(a)$ follows \cref{lem:Gamma:GL}
(recall that $e(k)-e(k-1)=\taulocal$).
\end{proof}

\begin{proof}[Proof of \cref{clm:abcdtwo}]
Recall that $w_j^{(e(k-1))}=u$. 
Conditioning on whether $\mathcal{E}^{(e(k))}$ happens or not gives us

\begin{align*}
\MoveEqLeft \Exf{e(k-1)}{ \1_{Z_j^{(e(k))}\ge \newL} \cdot \1_{\mathcal{E}^{(e(k))}} ~\Big|~ \Gamma_{g}^{(e(k-1))}, \Gamma_{\ell}^{(e(k))}(u) } \notag
\\& = \Prof{e(k-1)}{ \mathcal{E}^{(e(k))}~\Big|~ \Gamma_{g}^{(e(k-1))}, \Gamma_{\ell}^{(e(k))}(u) } \\ &\qquad\cdot
\Exf{e(k-1)}{ \1_{Z_j^{(e(k))}\ge \newL} \cdot \1_{\cE^{(e(k))}} ~\Big|~  \Gamma_{g}^{(e(k-1))}, \Gamma_{\ell}^{(e(k))}(u),\mathcal{E}^{(e(k))} } \notag
\\&\quad + \Prof{e(k-1)}{ \overline{\mathcal{E}^{(e(k))}}~\Big|~ \Gamma_{g}^{(e(k-1))}, \Gamma_{\ell}^{(e(k))}(u) } \\
&\qquad \cdot
\Exf{e(k-1)}{ \1_{Z_j^{(e(k))}\ge \newL}\cdot \1_{\cE^{(e(k))}} ~\Big|~ \Gamma_{g}^{(e(k-1))}, \Gamma_{\ell}^{(e(k))}(u),\overline{\mathcal{E}^{(e(k))}} }, \notag
\intertext{and since $\overline{\cE^{(e(k))}}$ implies that $1_{\cE^{(e(k))}}=0$, the above is}
& =\Prof{e(k-1)}{ \mathcal{E}^{(e(k))}~\Big|~ \Gamma_{g}^{(e(k-1))}, \Gamma_{\ell}^{(e(k))}(u) } \\ &\qquad\cdot
\Exf{e(k-1)}{ \1_{Z_j^{(e(k))}\ge \newL}\cdot \1_{\cE^{(e(k))}} ~\Big|~  \Gamma_{g}^{(e(k-1))}, \Gamma_{\ell}^{(e(k))}(u),\mathcal{E}^{(e(k))} }
\notag
\\&
\leq 1 \cdot
\Exf{e(k-1)}{ \1_{Z_j^{(e(k))}\ge \newL} ~\Big|~  \Gamma_{g}^{(e(k-1))}, \Gamma_{\ell}^{(e(k))}(u),\mathcal{E}^{(e(k))} }
\notag
\\ &\stackrel{(a)}{=} \frac{\Exf{e(k-1)}{ \1_{Z_j^{(e(k))}\ge \newL}
\cdot \1_{\Gamma_{g}^{(e(k-1))}\cap \Gamma_{\ell}^{(e(k))}(u) \cap\mathcal{E}^{(e(k))}}
 }}{\Pro{  \Gamma_{g}^{(e(k-1))} \cap \Gamma_{\ell}^{(e(k))}(u) \cap \mathcal{E}^{(e(k))} }}  
 \\ &= \frac{\Exf{e(k-1)}{ \1_{Z_j^{(e(k))}\ge \newL \cap \Gamma_{g}^{(e(k-1))}\cap \Gamma_{\ell}^{(e(k))}(u) \cap\mathcal{E}^{(e(k))}}
 }}{\Pro{  \Gamma_{g}^{(e(k-1))} \cap \Gamma_{\ell}^{(e(k))}(u) \cap \mathcal{E}^{(e(k))} }} 
 \\&= \frac{ \Prof{e(k-1)}{ (Z_j^{(e(k))} \geq L) \cap \Gamma_{\ell}^{(e(k))}(u) \cap \mathcal{E}^{(e(k))}}}{\Pro{  \Gamma_{g}^{(e(k-1))} \cap \Gamma_{\ell}^{(e(k))}(u) \cap \mathcal{E}^{(e(k))} }} 
 \\&= \frac{ \Prof{e(k-1)}{ Z_j^{(e(k))} \cdot \1_{\Gamma_{g}^{(e(k-1))} \cap \Gamma_{\ell}^{(e(k))}(u) \cap \mathcal{E}^{(e(k))}} \geq \newL }}{\Pro{  \Gamma_{g}^{(e(k-1))} \cap \Gamma_{\ell}^{(e(k))}(u) \cap \mathcal{E}^{(e(k))} }} 
\\
&  \stackrel{(b)}{\leq} \frac{\frac{1}{\newL} \cdot \Exf{e(k-1)}{Z_{j}^{(e(k))} \cdot \1_{\Gamma_{g}^{(e(k-1))} \cap \Gamma_{\ell}^{(e(k))}(u) \cap \mathcal{E}^{(e(k))}} }} {\Pro{  \Gamma_{g}^{(e(k-1))} \cap \Gamma_{\ell}^{(e(k))}(u) \cap \mathcal{E}^{(e(k))} }}
\\
&\stackrel{(c)}{=}  \frac{1}{\newL} \cdot 
\Ex{Z_{j}^{(e(k))}  ~\Big|~  \Gamma_{g}^{(e(k-1))}, \Gamma_{\ell}^{(e(k))}(u),\cE^{(e(k))} },
\end{align*}
where $(a)$ and $(c)$
uses the definition of  expectations conditional on events $\mathcal{E}$, $\Ex{Z ~\Big|~ \mathcal{E}}=\frac{ \Ex{Z \cdot \1_{\mathcal{E}}}}{\Pro{ \mathcal{E}}}$; while 
$(b)$ follows from applying the so-called ``conditional Markov's inequality'' 
(Exercise 8.2.5 in \cite{DBLP:books/daglib/0019830}). 
Recall that $\Exf{e(k-1)}{Z_{j}^{(e(k))}  ~\Big|~  \Gamma_{g}^{(e(k-1))}, \Gamma_{\ell}^{(e(k))}(u),\cE^{(e(k))} }$ is not a number but a random function over $\mathfrak{F}^{(e(k-1))}$-measurable events. In the following, for any such random function $f$, let us write $\sup_{\omega} f$ for the largest value $f$ could attain over its arguments. 
With this notation, we obtain the bound
\begin{align*}
\lefteqn{ \Exf{e(k-1)}{ \1_{Z_j^{(e(k))}\ge \newL} \cdot \1_{\mathcal{E}^{(e(k))}} ~\Big|~  \Gamma_{g}^{(e(k-1))}, \Gamma_{\ell}^{(e(k))}(u) }} \notag\\
&\leq  \frac{1}{\newL} \cdot \sup_{\omega} \Ex{Z_{j}^{(e(k))}  ~\Big|~  \Gamma_{g}^{(e(k-1))}, \Gamma_{\ell}^{(e(k))}(u),\cE^{(e(k))} } \\
&\leftstackrel{(a)}{=} \frac{1}{\newL} \cdot \sup_{\omega}
 \Ex{Z_{j}^{(e(k))}  ~\Big|~  \Gamma_{g}^{(e(k-1))}, \Gamma_{\ell}^{(e(k))}(u),\cE^{(e(k))},\cL^{(e(k))}(u) } \\
&\leq \frac{1}{\newL} 
\cdot \sup_{\omega}
 \Ex{Z_{j}^{(e(k))}  ~\Big|~ \Gamma_{\ell}^{(e(k))}(u),\cL^{(e(k))}(u) } \\
&\leftstackrel{\text{Def}~(\ref{eq:LEvent})}{\leq} 
\frac{1}{L} \cdot 
\left( L - \epsilon + \frac{1}{\log^4(n)} \right)
=1 - \frac{\epsilon}{\newL} + \frac{1}{\newL \cdot \log^4(n)} \notag,
\end{align*}
where $(a)$ holds, since by definition, when the events $\Gamma_{g}^{(e(k-1))}$, $\Gamma_{\ell}^{(e(k))}(u)$, $\mathcal{E}^{(e(k))}$ hold, the event $\cL^{(e(k))}(u)$ must also hold.
This completes the proof of the second statement and the proof of the claim.
\end{proof}

By repeatedly applying \cref{{lem:generallemmaone:RM}} over subsequent epochs, we can complete the analysis of Phase 1 and Phase 2.
\begin{lemma}[Phase 1]
\label{lem:rmlemmaphaseone}
      Let $ \epsilon\ge 4/\log^4 (n)$ and $1 \leq \newL \leq \log^7(n)$ be an integer. We assume that $x^{(0)}$ has at most $(L-\epsilon)n$ tokens and choose $t_1:=\tauglobal+\frac{
      \log(n)}{\log \log (n)} \cdot \taulocal$. Then,
\[
\Pro{ Y^{(t_1)} \le \frac{n}{\log(n)}} \ge 1-\exp\left(-\frac{\epsilon}{2\newL} \cdot \frac{\log (n)}{\log \log (n)} 
+ 8\cdot \log \log (n)\right)- 2 n^{-2}.
\]
\end{lemma}
\begin{proof}
Recall that $e(0)=\tauglobal$. Hence $t_1=e(\log(n)/\log\log(n))$.
Note that $\widetilde{Y}^{(e(0))}\le Y^{(e(0))} \leq n\cdot \log^7(n) $.
Applying \cref{lem:generallemmaone:RM} for $\ell:=\frac{\log (n)}{\log \log(n)}$ epochs gives us
\begin{align*}
\Ex{ \tilde{Y}^{(e(\ell))}} &= \Ex{\Ex{ \tilde{Y}^{(e(\ell))} ~\Big|~ \mathfrak{F}^{(e(\ell-1))}} }
\leq \left(1 - \frac{\epsilon}{\newL} + \frac{2}{\newL \cdot \log^4(n)} \right) \cdot \Ex{\tilde{Y}^{(e(\ell-1))}}\nonumber \\& \stackrel{(a)}{\le} \left(
1-\frac{\epsilon}{2\newL}\right) \cdot \Ex{\tilde{Y}^{(e(\ell-1))}} \leq e^{-\epsilon/(2\newL)}\cdot \Ex{\tilde{Y}^{(e(\ell-1))}},
\end{align*}
where $(a)$ used that $\epsilon \geq 4/\log^4(n)$. By iterating this, it follows 
$$\Ex{ \tilde{Y}^{(e(\ell))}} \leq e^{-\epsilon \cdot \ell/(2\newL) } \cdot \widetilde{Y}^{(e(0))} \leq e^{-\epsilon \cdot \ell/(2 \newL)}\cdot n \cdot  \log^7 (n).$$
Let $\beta:=\epsilon \cdot \ell/(2\newL) - 8\cdot \log\log (n)$. By Markov's inequality,
\begin{align}
\Pro{ \tilde{Y}^{(e(\ell))} \geq \exp\left(-\frac{\epsilon \cdot \ell}{2L}   + \beta\right)\cdot  n \cdot \log^7(n)} \leq e^{-\beta}. \label{eq:tilteY:first}
\end{align}
 Note that
\begin{align}
\exp\left(-\frac{\epsilon \cdot \ell}{2\newL}  +\beta\right)\cdot n\cdot \log^7 (n) &= \exp\left(-\frac{\epsilon\cdot \ell}{2 \newL}+\frac{\epsilon\cdot \ell}{2 \newL} -  8\cdot \log\log(n)\right)\cdot n\cdot \log^7 (n)\notag
\\& = \exp\left(- 8\cdot \log \log(n)\right)\cdot n\cdot \log^7 (n) \notag
=  \frac{n}{\log (n)} \label{eq:nlognApprox}\notag.
\end{align}
By the definition of $\tilde{Y}^{(e(\ell))}$ and the union bound we get 
\begin{align}
\Pro{ Y^{(e(\ell))} \ge \frac{n}{\log (n)}}& \leq 
\Pro{ Y^{(e(\ell))} \ge \frac{n}{\log (n)} ~\Big|~ \bigcap_{s=0}^{\ell} \cE^{(e(s))} } + \Pro{\bigcup_{s=0}^{\ell} \overline{\mathcal{E}^{(e(s))}}}\notag
\\&\le \Pro{ \tilde{Y}^{(e(\ell))} \ge \frac{n}{\log (n)} } + \Pro{ \bigcup_{s=0}^{\ell} \overline{\mathcal{E}^{(e(s))}}}\notag \\
& \leftstackrel{\text{Eq.~(\ref{eq:tilteY:first})}}{\leq} e^{-\beta} + \sum_{s=0}^{\ell} \Pro{ \overline{\mathcal{E}^{(e(s))}}} \leftstackrel{(a)}{\leq} e^{-\beta} + 2 \cdot (\ell+1) \cdot n^{-3}, \notag
\end{align}
where $(a)$ follows from \cref{lem:crucialfirststatement}. Recalling
our choice of $\beta$ earlier in this proof, and the choice of $\ell=\frac{\log(n)}{\log \log(n)}$, finishes the proof.  
\end{proof}

\begin{lemma}[Phase 2]
\label{lem:rmlemmaphasetwo}
Assume that the load vector $x^{(0)}$ has at most $(1-\epsilon) \cdot n$ tokens, where $\frac{1}{2} < \epsilon < 1$. Then it holds for $t_2:=\tauglobal+\frac{4}{-\log \left( 1 - \epsilon + \frac{2}{\log^4 (n)} \right)} \cdot \log(n) \cdot \taulocal$,
\[
\Pro{ \max_{w \in V} X_{w}^{(t_2)} \leq 1 } \geq 1 - n^{-2}.
\]
\end{lemma}
\begin{proof}
    The proof of this lemma is similar to the one of the previous lemma, but here we have the special case $L=1$.
By assumption, $Y^{(e(0))}\leq Y^{(0)} \leq n$. Furthermore, $\tilde{Y}^{(e(0))} \leq Y^{(e(0))} \leq n$.
Applying \cref{lem:generallemmaone:RM} with $L=1$ yields for any epoch $k \geq 1$,
\begin{align*}
\Ex{ \tilde{Y}^{(e(k))}} &= \Ex{\Ex{ \tilde{Y}^{(e(k))}~\Big|~\mathfrak{F}^{(e(k-1))}} } \leq \left( 1 - \epsilon + \frac{2}{\log^4 (n)} \right) \cdot \Ex{\tilde{Y}^{(e(k-1))}}.
\end{align*}
We now consider $\ell:=\frac{4}{-\log\left( 1 - \epsilon + \frac{2}{\log^4 (n)} \right)} \cdot \log(n)$ many epochs. 

It follows that
\begin{align*}
\Ex{ \tilde{Y}^{(e(\ell))}} &\leq  \left( 1 - \epsilon + \frac{2}{\log^4 (n)} \right)^{\ell} \cdot \widetilde{Y}^{(e(0))}
\\ &\leq \exp\left( \log \left( 1 - \epsilon + \frac{2}{\log^4 (n)} \right) \cdot \frac{4}{-\log \left( 1 - \epsilon + \frac{2}{\log^4 (n)} \right)} \cdot \log(n) \right) \cdot n \\
&=\exp\left(-4\cdot \log(n)\right) \cdot n =  n^{-3}.
\end{align*}
By Markov's inequality,
$\Pro{ \tilde{Y}^{(e(\ell))} \geq 1} \leq n^{-3}$.
By the definition of $\tilde{Y}^{(e(\ell))}$ and the union bound,
\begin{align*}
\Pro{ Y^{(e(\ell))} \ge 1}
&\le \Pro{ Y^{(e(\ell))} \ge 1 ~\Big|~ \bigcap_{s=0}^{\ell} \cE^{(e(s))} } + \Pro{\bigcup_{s=0}^{\ell} \overline{\mathcal{E}^{(e(s))}}}
\\
&\leq \Pro{ \tilde{Y}^{(e(\ell))} \ge 1 } + \sum_{s=0}^{\ell} \Pro{ \overline{\mathcal{E}^{(e(s))}}}\\
&\leq n^{-3} + \sum_{s=0}^{\ell} \Pro{ \overline{\mathcal{E}^{(e(s))}}}
\stackrel{(a)}{\leq} n^{-3} + 2\cdot (\ell+1) \cdot n^{-3} \leq n^{-2},
\end{align*}
where the last line used that $\ell=O(\log(n))$ and $(a)$ follows from \cref{lem:crucialfirststatement},.
\end{proof}

\section{From Linear Number of Tokens to Constant Discrepancy}\label{sec:arbitrary}

\subsection{Reducing Discrepancy to $38$}\label{sec:discrepancy-38}
Here we show that once the load vector consists of only $O(n)$ tokens, then after additional
$O( \tauglobal + \frac{\log (n)}{\log \log (n)} \cdot \taulocal)$ rounds the discrepancy is reduced to $38$ (see \cref{lem:discrepancy-38}). The proof of the result relies on \cref{lem:toNToken} showing that the number of tokens with height at least $\overline{x}+1$ is at most $16n$. 

In this section we derive discrepancy bounds for load vectors where the number of tokens with height above the average load is $O(n)$. In the following lemma we show that the discrepancy can be reduced to $38$.

\begin{lemma}\label{lem:discrepancy-38}
Let $t^{\star}:= 2 \cdot \tauglobal+ \frac{6\log(n)}{\log\log(n)} \cdot \taulocal$.
Assume our process applies a $(\tauglobal,\taulocal)$-good sequence of matchings $\left(\M^{(s)}\right)_{s=1}^{\infty}$ to a load vector $x^{(0)}$, which satisfies 
\[
 \sum_{w \in V} \max\{ x_w^{(0)} - \overline{x}, 0 \} \leq 16 \cdot n.
\]
Then
\[
\Pro{ \discr(X^{(t^{\star})}) \leq 38 } \geq 1 - \exp(- (1/80)) \cdot \log(n) / \log \log(n)).
\]
\end{lemma}
\begin{proof}
    For any $t \geq 0$, define an auxiliary load vector $\widetilde{X}_u^{(t)} := \max\{ X_u^{(t)} - \lceil \overline{x} \rceil,0 \}$, $u \in V$, that is, we subtract $\lceil \overline{x} \rceil$ tokens from any node (as long as its load is large enough).
By assumption, we have $\sum_{w \in V} \widetilde{x}^{(0)}_w \le 16   \cdot n$. Here we can apply \cref{lem:NtokensTo25Height} with $\newL=17$, $\epsilon:=1/2$ and $t^{\star}:=2\cdot \tauglobal+ \frac{6\log(n)}{\log\log(n)} \cdot \taulocal$ and obtain for $c:=1/70$,
\[
\Pro{ \max_{w \in V} \widetilde{X}_w^{(t^{\star})} \leq 18} \geq 1 - \exp \left( - c \cdot \frac{\log(n)}{\log \log(n)} \right).
\]
Applying \cref{obs:subAddAppl} we get that the same holds for $X^{(t^{\star})}$, resulting in
\[
\Pro{ \max_{w \in V} X_w^{(t^{\star})} \leq \lceil \overline{x} \rceil + 18}
\geq 1 - \exp \left( - c \cdot \frac{\log(n)}{\log \log(n)} \right).
\]
Using a simple symmetry argument (see~\cref{obs:upLowDiscRelation}),
\begin{align*}
\Pro{ \min_{w \in V} X_w^{(t^{\star})} \geq \lfloor \overline{x} \rfloor - 19 } &\geq 1 - \exp \left( - c \cdot \frac{\log(n)}{\log \log(n)} \right).
\end{align*}
A final union bound gives
\[
\Pro{ \discr \left( X^{(t^{\star})} \right) \leq 38
} \geq 1 - 2 \exp \left( - c \cdot \frac{\log(n)}{\log \log(n)} \right) \geq 1 - \exp\left( -(1/80) \cdot \frac{\log(n)}{\log \log(n)} \right),
\]
which yields the statement.
\end{proof}

\subsection{Reducing Discrepancy from $38$ to $4$}

\begin{lemma}\label{lem:disc54to4}
Let $x^{(0)}$ be any load vector with $\discr(x^{(0)}) \leq 38$. Then for $t^{\star}:=2 \cdot \tauglobal+ {6\taulocal \cdot \log (n)}/{\log\log(n)}$ we have,
\[
 \Pro{ \discr(x^{(t^{\star})}) \leq 4 } \geq 1 -  \exp\left( - (1/160) \cdot \log (n) / \log \log (n) \right).
\]
\end{lemma}
\begin{proof}
Assume without loss of generality that the load values are $\{0,1,2,\ldots,38\}$.
Let us pick $\newL:=\left\lceil \overline{x} + \frac{1}{2} \right\rceil $ 
and $\epsilon:=\frac{1}{2}$; clearly $\newL\le 39$ since $\overline{x} \leq 38$. 
Using \cref{lem:NtokensTo25Height}, we obtain that at round $t^{\star}:=2 \cdot \tauglobal+ {6\taulocal \cdot \log (n)}/{\log\log(n)}$ the maximum load is at most $\newL+1 = \left\lceil \overline{x}+\frac{1}{2} \right\rceil + 1$ with probability at least 
\[
1-\exp\left(-\frac{\log(n)}{156\cdot \log\log(n)}+ 8\cdot \log\log (n)\right) - 2\cdot n^{-2}.
\]

Let us now consider the load vector $\tilde{x}^{(0)}:=38-x^{(0)}$. 
By \cref{obs:flip}, $\tilde{x}^{(t)}=38-x^{(t)}$ for all $t \geq 1$. Also, $\overline{(\tilde{x})} = 38 - \overline{x}$. Repeating the above argument, but now applied to $\tilde{x}$,
yields for any $u \in V$,
 \[
 \tilde{x}_u^{(t^{\star})} \leq L+1 = \left\lceil \overline{(\tilde{x})}+\frac{1}{2} \right\rceil + 1,
\]
which implies
\begin{align*}
 x_u^{(t^{\star})} &\geq 38- \left\lceil \overline{(\tilde{x})}+\frac{1}{2} \right\rceil - 1 = 37- \left\lceil 38 - \overline{x} +\frac{1}{2} \right\rceil  \stackrel{(a)}{=} 37 - \left(38 + \left\lceil -\overline{x} + \frac{1}{2} \right\rceil \right)   \stackrel{(b)}{=} \left\lfloor \overline{x} - \frac{1}{2} \right\rfloor - 1,
\end{align*}
where $(a)$ used the fact that for any integer $k$ and real $z$, $\lceil k + z \rceil = k + \lceil z \rceil$ and $(b)$ used the fact that $\lceil -z \rceil = -\lfloor z \rfloor$ for any real $z$. Hence we can also conclude that the minimum load in $x^{(t^{\star})}$ is at least $\left\lfloor \overline{x} - \frac{1}{2} \right\rfloor - 1$.
Since all load values at round $t^{\star}$ are integers in the interval $\left[ \left\lfloor \overline{x}-\frac{1}{2} \right\rfloor - 1, 
\left\lceil \overline{x}+\frac{1}{2} \right\rceil + 1 \right]$, which are at most $5$ values, the discrepancy is at most $4$. Hence a final union bound gives, 
\[
\Pro{\discr(X^{(t^{\star})})\le 4}
 \ge 1-2\exp\left(-\frac{\log(n)}{156\cdot \log\log(n)}\right) - 4\cdot n^{-2} \ge 1-\exp\left(-\frac{\log(n)}{160\cdot \log\log(n)}\right).\]
\end{proof}

\subsection{Reducing Discrepancy from $4$ to $3$}\label{Sec:ReduceTo3}

To obtain discrepancy $3$, we may not be able to apply \cref{lem:NtokensTo25Height} directly. For example, if $\overline{x}$ is an integer (or close to an integer), the highest load is 
$\overline{x}+2$, and the minimum load is $\overline{x}-2$, then there is no $L$ and $\epsilon$ such that the application of \cref{lem:NtokensTo25Height} would result in a reduced discrepancy. To overcome this problem, we first remove a small number of tokens at the two highest levels (they will be called ``secondary"), and focus on the movement of the remaining tokens (called ``primary"). We can then show, by \cref{lem:NtokensTo25Height}, that most of the primary tokens reduce their height to at most $\overline{x}$. Then we can show that the amount of these primary tokens together with the secondary tokens is smaller than $(L-\epsilon) \cdot n$ (for $L=2$ and proper $\epsilon$), and by applying \cref{lem:rmlemmaphasetwo} we obtain that all tokens at height $\overline{x}+2$ will reduce their height.

\begin{lemma}\label{reductiontothree}
Consider any initial load vector with $\discr(x^{(0)}) \leq 4$.
Then for any $0 < \delta < 1/2$ and $t^*:=2 \tauglobal + \frac{5}{\delta \log \log(n)} \cdot \log(n) \cdot \taulocal$,
\[
\Pro{\discr (X^{(t^*)}) \leq 3 } \geq 1 - 2 \cdot \exp\left( -\log^{1-2  \delta} (n) \right).
\]
\end{lemma}

\begin{proof}
For simplicity, we assume that the load values at time $0$ are $\{0,1,2,3,4\}$ (this can be achieved by reducing the load by the same value at each node).

To show the lemma, we consider two different cases ($1$ and $2$). In Case 1, we assume that $\overline{x} \leq 2$. 
Ideally, we would like to apply \cref{lem:NtokensTo25Height} for $L=2$, but this requires that the total number of tokens is at most $(L-\epsilon) \cdot n$ for $\epsilon > \frac{4}{\log^4 n}$, but this may not be satisfied if $\overline{x} \in (2 - \frac{4}{\log^4 n},2]$. To overcome this problem, we will first apply the first phase in the proof of \cref{lem:NtokensTo25Height} to a load vector with slightly fewer tokens. To this end, we first mark $\frac{n}{\log^{\delta} n}$ tokens at height $3$ and $4$, and call these tokens ``secondary'' tokens (all other tokens are called ``primary'').
We now apply a $(\tauglobal,\taulocal)$-good sequence of matchings $\left(\M^{(s)}\right)_{s=1}^{\infty}$,
and consider the load balancing process with a fixed sequence of orientations on the graph w.r.t. the primary  tokens only (in which the secondary tokens are removed) versus the process with the same sequence of orientations w.r.t. both types of tokens. The load vector at some time $t$, which results from the load balancing process w.r.t. the primary tokens only, is denoted by $p^{(t)}$.

Clearly, the number of tokens in $p^{(0)}$  is at most $ 2n - \frac{n}{\log^{\delta} n} = (L-\epsilon) \cdot n$, where $L:=2$ and $\epsilon:=\frac{1}{\log^{\delta} n}$.
Then we can apply~\cref{lem:rmlemmaphaseone} for $t_1:=\tauglobal+\log(n)/\log\log(n) \cdot \taulocal$, (note that $\epsilon \geq 4/\log^4(n)$ is satisfied) and conclude that the number of nodes with load at least $3$ in $p^{(t_1)}$ satisfies
\begin{align*}
\Pro{  \left|\left\{ u \in V \colon p_u^{(t_1)} \geq 3 \right\} \right| \leq \frac{n}{\log(n)} } &\geq 1 - \exp\left( - \frac{ (1/\log^{
 \delta} n)}{4} \cdot \frac{\log(n)}{\log \log(n)} + 8 \cdot \log \log (n) \right) - 2 n^{-2} \\ &\geq 1 - \exp\left( -\log^{1-2  \delta} (n) \right).
\end{align*}

According to the second statement of \cref{obs:subAdd}, we have for all rounds $t \geq 0$ and $u \in V$,
\[
 x_u^{(t)} \geq 
 p_u^{(t)}.
\]
As there are $n/\log^{\delta} (n)$ secondary tokens, the number of tokens at height at least $3$ in $x^{(t_1)}$ 
is at most $$2n/\log (n) + n/\log^{\delta} (n) \leq 2 n/ \log^{\delta}(n)=(1-\epsilon) \cdot n,$$ for $\epsilon:=1-\frac{2}{\log^{\delta}(n)} $. Then we consider an additional phase, applied to $x^{(t)}$, $t \geq t_1$, where we consider only these $(1-\epsilon) \cdot n$ tokens at height $3$ and $4$. By \cref{lem:rmlemmaphasetwo}, by round $t^*:=t_1+ \tauglobal+\frac{4}{-\log \left(1 - \epsilon + \frac{2}{\log^4 (n)} \right)} \cdot \log(n) \cdot \taulocal$, all tokens at height $4$ are eliminated with probability $1-n^{-2}$. Hence the total number of rounds is
\[
 t^{*} = 2 \tauglobal + \frac{\log(n)}{\log\log(n)} \cdot \taulocal + \frac{4}{\delta \log \log(n) -2} \cdot \log(n) \cdot \taulocal 
 \leq 2 \tauglobal + \frac{5}{\delta \log \log(n)} \cdot \log(n) \cdot \taulocal,
\]
and by the union bound the success probability for both phases is
\[
  1 -  \exp\left( -\log^{1-2  \delta} (n) \right) + n^{-2} \geq 1 - 2 \cdot \exp\left(-\log^{1-2  \delta} (n) \right).
\]
In Case 2, where $\overline{x} \geq 2$, we consider the flipped load vector with the entries $y_i^{(0)} = 4-x^{(0)}_{i}$, and apply the analysis of the first case to this vector. 
\end{proof}

\section{Proofs of our Main Results}\label{sec:proofsmainresults}

Here we provide a proof of our main \Cref{thm:main-result}.

Before proving the theorem, we present a simple lemma which relates $(\tauglobal,\taulocal)$-good sequences of matchings to the property of $(K,\epsilon)$-smoothing.
\begin{lemma}
\label{lem:TauglobalSmoothingRelation}
Consider a sequence of $(\tauglobal,\taulocal)$-good matchings $\left(\M^{(s)}\right)_{s=1}^{\infty}$.
Let $K\ge 2n^2$ and $0<\epsilon\le 1$. Then for $t^* := (3\cdot \log (\frac{K}{\epsilon})/\log(n))\cdot \tauglobal$,
\[
 \Pro{ \text{ $\left(\M^{(s)}\right)_{s=1}^{t^{*}}$ is $(K,\epsilon)$-smoothing }} \geq 1-n^{-2}. 
\]
\end{lemma}

\begin{proof}\label{Pr:TaglobalSmoothingRelation}
We first define $x:=3\left\lceil\log_{2n^2} (\frac{K}{\epsilon})\right \rceil$. We will consider $x$ subsequent and disjoint subsequences of matchings, 
$\left(\M^{(s)}\right)_{s=(i-1)\cdot \tauglobal+1}^{i\cdot \tauglobal}$, where $i \in [1,x]$. For any $i\in \left[1, x \right]$ we define a random variables $Z_i$ to be zero if  the sequence of matchings $\left(\M^{(s)}\right)_{s=(i-1)\cdot \tauglobal+1}^{i\cdot \tauglobal}$ is $(1,1/(2n^2))$-smoothing and one otherwise.
From \cref{def:taus} together with \cref{obs:tauGtauCRelation} we get
\[
\Pro{Z_i=1} \le n^{-3}.
\]
We define $Z:=\sum_{i=1}^{x}Z_i$ and an event
$
\Psi :=\{Z\le x/2\}.
$
By linearity of expectation we get that
$\Ex{Z} \le x/n^3$. From Markov's inequality it follows that
$\Pro{Z\ge \frac{x}{2}} \le \Pro{Z\ge \frac{x}{n}}\le n^{-2}$ implying that
$\Pro{\Psi}\ge 1-n^{-2}$.

In the remainder of the proof we assume that the event $\Psi$ occurs. This implies that at least $x/2$ matching subsequences of length $\tauglobal$ are $(1,1/(2n^2))$-smoothing.
Note that the discrepancy is non-increasing over the time. Hence conditioning on the event $\Psi$, we get that after $x\cdot \tauglobal$ rounds the discrepancy is at most
\[K\cdot \left(\frac{1}{2n^2}\right)^{x/2} \le K\cdot \left(\frac{1}{2n^2}\right)^{\log_{2n^2} (K/\epsilon)} = \epsilon.\]
Finally, we have that
\[
3\cdot \left\lceil\log_{2n^2}\left(\frac{K}{\epsilon}\right) \right\rceil = 3\cdot \left\lceil\frac{\log (\frac{K}{\epsilon})}{\log (2n^2)} \right \rceil \stackrel{(\star)}{\le} 3 \cdot \frac{\log (\frac{K}{\epsilon})}{\log(n)} = t^{\star} \cdot \frac{1}{\tauglobal},
\]
where $(\star)$ used that $K \geq 2n^2$ and $\epsilon \leq 1$, which implies that the argument of $\lceil . \rceil$ is at least $1$. This completes the proof.
\end{proof}

\begin{proof}[Proof of \cref{thm:main-result}]
To ease the notation, we define $\ell:=\frac{\log(n)}{\log\log(n)}$.
Let $t_0:=3\tauglobal \cdot \log(2Kn)/\log(n)$. First, we define two events 
\[
\mathcal{G}^{\star}:=\left\{ \left(\M^{(s)}\right)_{s=1}^{t_0} ~\mbox{is $(K,1/(2n))$-smoothing}\right\},
\]
and
\[ \mathcal{G}_0 := \left\{ \sum_{w \in V} \max\{ X_w^{(t_0)} - \overline{x}, 0\} \leq 16 \cdot n \right\}.
\]
From \cref{lem:TauglobalSmoothingRelation} it follows that $\Pro{\mathcal{G}^{\star}}\ge 1-n^{-2}$ and from \cref{lem:toNToken} it follows that 
\[\Proco{\mathcal{G}_0 }{\mathcal{G}^{\star}} \geq 1-2\cdot n^{-2}.\] 

Let $t_1:=t_0+2 \cdot\tauglobal + {6 \ell} \cdot \taulocal$.
Here we define an event
$
\mathcal{G}_1:= \left\{\discr(X^{(t_1)}) \leq 38\right\}$
and from \cref{lem:discrepancy-38} it follows that \[\Pro{\mathcal{G}_1 ~\Big|~ \mathcal{G}_0 }\ge 1-\exp(- (1/80) \cdot \ell).\]
Let $t_2:=t_1+ 2\cdot \tauglobal+{6 \ell} \cdot \taulocal$. We define another event 
$
\mathcal{G}_2:= \left\{\discr(X^{(t_1)}) \leq 4\right\}$, and
from \cref{lem:disc54to4} it follows that
\[
\Proco{\mathcal{G}_2}{\mathcal{G}_1} \ge 1-\exp(-(1/160)\cdot \ell).\]
Finally, by the definition of conditional probability we obtain
\begin{align*}
\Pro{\mathcal{G}_2}& \geq \Proco{\mathcal{G}_2}{\mathcal{G}_1}\cdot \Pro{\mathcal{G}_1}\ge \Proco{\mathcal{G}_2}{\mathcal{G}_1}\cdot \Proco{\mathcal{G}_1}{\mathcal{G}_0}\cdot \Proco{\mathcal{G}_0}{\mathcal{G}^{\star}}\cdot \Pro{\mathcal{G}^{\star}} \\& \geq 1-\exp(-(1/160)\cdot \ell) - \exp(-(1/80)\cdot \ell) - 2\cdot n^{-2} -n^{-2} \ge 1-\exp(-(1/200)\cdot \ell).
\end{align*}
Note that 
\begin{align*}
    t_2: &= t_1+ 2\cdot \tauglobal+{6 \ell} \cdot \taulocal = t_0+4 \cdot\tauglobal + {12 \ell} \cdot \taulocal
    \\&= \left(\frac{3\log(2Kn)}{\log(n)}+4\right)\cdot \tauglobal + {12 \ell} \cdot \taulocal,
\end{align*}
finishing the proof of the first statement. 

Let $c>0$ be any constant and $\tau:= t_2+ 2\cdot \tauglobal + \lceil 10/c \rceil \cdot \ell \cdot \taulocal$. Here we define an event $\mathcal{G}_3:=\left\{\discr(X^{(\tau)})\le 3\right\}$.
From \cref{reductiontothree} with $\delta:=c/2$ it follows that \[\Proco{\mathcal{G}_3}{\mathcal{G}_2}\ge 1-2\exp(-\log^{1-c}(n)).\] By the definition of conditional probability we get,
\begin{align*}
    \Pro{\mathcal{G}_3} & \ge \Proco{\mathcal{G}_3}{\mathcal{G}_2}\cdot \Pro{\mathcal{G}_2} \ge 1 -2\exp(-\log^{1-c}(n)) - \Pro{\overline{\mathcal{G}_2}}
    \\&\ge 1- 2\exp(-\log^{1-c}(n)) - \exp(-(1/200)\cdot \ell) \ge 1-\exp(-\log^{1-2c}(n)) .
\end{align*}
Finally, 
\begin{align*}
    \tau &= t_2 + 2\cdot \tauglobal + \frac{\lceil 10/c\rceil \log (n)}{\log\log(n)}\cdot \taulocal \\
    &= \left(\frac{3\log(2Kn)}{\log(n)}+6\right)\cdot\tauglobal + \frac{(\lceil 10/c\rceil+12)\log(n)}{\log\log(n)} \cdot \taulocal.
\end{align*}
Since $c$ is a constant, this finishes the proof.
\end{proof}

We now provide formal statements for the remarks in \Cref{sec:mainresintro} regarding application of the main theorem in various concrete settings: balancing circuit and random matchings.

\begin{corollary}
\label{cor:BC}
Assuming the premise of \cref{thm:main-result}, for the balancing circuit model we have
$\tau = O\left(\Delta\cdot \log(Kn)/{\left(1-\lambda\left(\vphantom{X^1}\smash{\M^{[1,\Delta]}}\right)\right)}\right)$.
\end{corollary}
\begin{proof}
To simplify notation, we define $\lambda:=\lambda(\M^{[1,\Delta]})$.
    Using standard spectral arguments (e.g.,~\cite[Theorem~1]{DBLP:conf/focs/RabaniSW98}) yields that the sequence of matchings $\left(\M^{(s)}\right)_{s=1}^t$ for
$t:=\frac{4\Delta}{1-\lambda}\cdot \log\left(\frac{Kn}{\epsilon}\right)$
is $(K,\epsilon)$-smoothing. Note that this corresponds to the multiplication of $t/\Delta$ round matrices. 

In the following analysis we will show that the balancing circuit model is also $(\tauglobal,\taulocal)$-good with $\tauglobal:= \frac{7\Delta\cdot \log (n)}{1-\lambda}$ and $\taulocal:=\frac{11\Delta\cdot \log\log (n)}{1-\lambda}$.
 Fix an arbitrary $t\ge 1$.
From \cite[Lemma~2.4]{DBLP:conf/focs/SauerwaldS12} (restated as \cref{lem:prob:BalancingCircuit}) it follows for any $u\in V$ and for the choice of $\tauglobal$ as above that we have
\[
\left\| \M_{u,.}^{[t,t+ \tauglobal]} - \vec{ \frac{1}{n}} \right\|_2^2 \le \left(1-\lambda \right)^{\log (n^7)/(1-\lambda)} \le e^{-\log (n^7)} =\frac{1}{n^7}.
\]

Now we consider $\taulocal$. We fix a node $u\in V$ and an arbitrary $t\ge 1$.
Then we have
\begin{align*}
\left\| \M_{u,.}^{[t,t+\taulocal]}\right\|_2^2 -\frac{1}{n} \stackrel{\text{Obs.~\ref{obs:secNormExpansion}}}{=} \left\| \M_{u,.}^{[1,\taulocal]} - \vec{ \frac{1}{n}} \right\|_2^2 \stackrel{\text{Lem.~\ref{lem:prob:BalancingCircuit}}}{\leq} \frac{1}{\log^{11}n}.
\end{align*}
From this we get
\[
\left\| \M_{u,.}^{[t,t+\taulocal]}\right\|_2^2 \le \frac{1}{\log^{11}n} + \frac{1}{n}\le \frac{1}{\log^{10}n}.
\]
\end{proof}

\begin{corollary}
\label{cor:RM}
Assuming the premise of \cref{thm:main-result}, for the random matching model we have
$\tau =O\left(\log(Kn)/\left(p_{\min}\cdot \Delta \cdot \left(1-\lambda\left(\vphantom{X^1}\P\right)\right)\right)\right)$.
\end{corollary}
\begin{proof}
    It suffices to show that the random matching model is
$(\tauglobal,\taulocal)$-\emph{good} with $\tauglobal := \frac{14\cdot \log(n)}{p_{\min}\cdot \Delta\cdot (1-\lambda(\P))}$ and $\taulocal := \frac{22\cdot \log\log(n)}{p_{\min}\cdot \Delta\cdot (1-\lambda(\P))}$. 

Fix an arbitrary $t\ge 1$. From \cite[Corollary~2.7]{DBLP:conf/focs/SauerwaldS12} (see \cref{cor:prob:randomMatcing}) it follows for any $u\in V$ and the choice of $\tauglobal$ as above (and sufficiently large $n$) that we have
\begin{align*}
\Pro{ \left\| \M_{u,.}^{[t,t+\tauglobal]} - \vec{ \frac{1}{n}} \right\|_2^2\le \frac{1}{n^7}}
& \ge 1-\frac{1}{n^{7}} .
\end{align*}
Applying the union bound over all nodes $u\in V$ implies that
\[
\Pro{\bigcap_{u \in V} \left\| \M_{u,.}^{[t,t+\tauglobal]} - \vec{ \frac{1}{n}} \right\|_2^2 \le \frac{1}{n^7}} \ge 1-\frac{1}{n^6}.
\]

Now we consider $\taulocal$. We fix a node $u\in V$ and an arbitrary $t\ge 1$.
Then we have
\begin{align*}
\Pro{\left\| \M_{u,.}^{[t,t+\taulocal]} \right\|_2^2 \le \frac{1}{\log ^{10} n}} &\ge \Pro{\left\| \M_{u,.}^{[t,t+\taulocal]} \right\|_2^2\le \frac{1}{\log ^{11} n} +\frac{1}{n}}
\\& \leftstackrel{\text{Obs.~\ref{obs:secNormExpansion}}}{=}
\Pro{\left\| \M_{u,.}^{[t,t+\taulocal]} - \vec{ \frac{1}{n}}
\right\|_2^2\le \frac{1}{\log ^{11} n} } \stackrel{\text{Lem.~\ref{cor:prob:randomMatcing}}}{\ge} 1- \frac{1}{\log ^{11} n}.\qed
\end{align*}
\end{proof}

\bibliographystyle{siamplain}
\bibliography{Camera-Ready-arXiv-long}

\section*{Appendix}
\section{Assorted Tools}

\subsection{Basic Properties of the Height-Sensitive Process}

\begin{lemma}
\label{lem:OneTokenRW}
Consider any pair of rounds $t_1 < t_2$ and let $(\M^{(s)} )_{s=t_1+1}^{t_2}$ be an arbitrary but fixed sequence of matchings. For any token $i\in \mathcal{T}$ it holds that (1) $H_i^{(t_1)} \geq H_i^{(t_2)}$, and (2)
$
\Proco{W_i^{(t_2)}=v}{w_i^{(t_1)}=u} = 
\M_{u,v}^{[t_1+1,t_2]}.
$
\end{lemma}
\begin{proof}
    The height of a token can only change in the moving step of a round $t$. This only happens if the token is on one endpoint of an edge $[u:v] \in \M^{(t)}$ with $x_{u}^{(t-1)} \ge x_{v}^{(t-1)}$. Then tokens on node $u$ at height $x_{v}^{(t-1)}+\lceil \frac{x_{u}^{(t-1)} - x_{v}^{(t-1)}}{2}\rceil +i$ for some $i>0$ will decrease their height by $\lceil \frac{x_{u}^{(t-1)} - x_{v}^{(t-1)}}{2}\rceil$ and move to node $v$. No other token on $u$ or $v$ will change its height during the moving step, and hence the first statement follows.

    Similar to the proof of \cite[Lemma~4.1]{DBLP:conf/focs/SauerwaldS12}, we prove the second statement by induction over $t \in [t_1+1,t_2]$, that is, for all $u, v \in V$, we have
    $\Proco{W_i^{(t)}=v}{w_i^{(t_1)}=u} = \M_{u,v}^{[t_1+1,t]}$. For $t=t_1+1$, $\M^{[t_1+1,t]}$ is the identity matrix, which means that the induction base holds. For the induction hypothesis, consider $\Proco{W_i^{(t)}=v}{w_i^{(t_1)}=u}$. If $v$ is not part of a matching in round $t$, then 
    \[
    \Proco{W_i^{(t)}=v}{w_i^{(t_1)}=u} = \Proco{W_i^{(t-1)}=v}{w_i^{(t_1)}=u} \stackrel{(\star)}{=} \M_{u,v}^{[t_1+1,t-1]}= \M_{u,v}^{[t_1,t]},
    \]
    where $(\star)$ used the induction hypothesis.
    If $v$ is matched with a node $w$ in round $t$, then by definition of the height-sensitive process the token $i$ has a probability of $\frac{1}{2}$ of reaching $v$ from either $v$ or $w$, and hence
    \begin{align*}
         \Proco{W_i^{(t)}=v}{w_i^{(t_1)}=u} &= 
        \Proco{W_i^{(t-1)}=v}{w_i^{(t_1)}=u} \cdot \frac{1}{2}
         +  \Proco{W_i^{(t-1)}=w}{w_i^{(t_1)}=u} \cdot \frac{1}{2} \\
         &\stackrel{(\star)}{=} 
         \M_{u,v}^{[t_1,t]} \cdot \frac{1}{2}
         + \M_{u,w}^{[t_1,t]} \cdot \frac{1}{2} \\
         &= \M_{u,v}^{[t_1,t]} \cdot \M_{v,v}^{(t+1)}
         + \M_{u,w}^{[t_1,t]} \cdot \M_{w,v}^{(t+1)} \\
         &= \M_{u,v}^{[t_1,t+1]},
    \end{align*}
 where $(\star)$ used the induction hypothesis.
\end{proof}

\subsection{Tail Bounds}

For completeness we state the following tail bound that we use in our proofs.

\begin{theorem}[cf.~{\cite[page 92, Theorem 4.16, Azuma's Inequality]{DBLP:books/cu/MotwaniR95}}]
\label{thm:MethodOfBoundedDifferences}
Let $X_0, X_1, \ldots$ be a martingale sequence such that for each $k$, $|X_k - X_{k-1}| \le c_k$, where $c_k$ may depend on $k$. Then, for all $t\ge 0$ and any $\lambda>0$,
\[\Pro{\lvert X_t-X_0 \rvert \ge \lambda} \le 2\cdot \exp\left(- \frac{\lambda^2}{2\sum_{k=1}^t c_{k}^2}\right).\]
\end{theorem}

\begin{lemma}[cf.~{\cite[Proposition~13.2.6]{Rosenthal2006}}]\label{lem:take-out-what-is-known}
Let $X$ and $Y$ be random variables, and let $\mathfrak{F}$ be a sub-$\sigma$-algebra. Suppose that $\Ex{X}$ and $\Ex{XY}$ are finite, and furthermore that $X$ is $\mathfrak{F}$-measurable. Then with probability~$1$,
\[ \Ex{X \cdot Y~\Big|~\mathfrak{F}} = X \cdot \Ex{Y~\Big|~\mathfrak{F}}.
\]
\end{lemma}

\subsection{Basic Results from \cite{DBLP:conf/focs/SauerwaldS12}}

In this appendix we list basic results from \cite{DBLP:conf/focs/SauerwaldS12} that we use in our analysis.

\begin{lemma}[cf.~{\cite[Lemma 2.4]{DBLP:conf/focs/SauerwaldS12}}]\label{lem:prob:BalancingCircuit}
Consider the balancing circuit model with sequence of matchings $\left(\M^{(s)}\right)_{s=1}^{\infty}$. Let $\lambda:=\lambda(\M^{[1,\Delta])}$ Then for any node $u\in V$ it holds
\[
\left\|\M_{u,.}^{[1,t\cdot \Delta]} - \vec{ \frac{1}{n}} \right\|_2^2 \le \lambda^{t}.
\]
\end{lemma}

\begin{lemma}[{cf.~\cite[Corollary 2.7]{DBLP:conf/focs/SauerwaldS12}}]\label{cor:prob:randomMatcing}
Consider the random matching model with sequence of matchings $\left(\M^{(s)}\right)_{s=1}^{\infty}$. Then for any node $u\in V$ and any round $t \geq 1$ we have
\[
\Pro{ \left\| \M_{u,.}^{[1,t]} -\vec{ \frac{1}{n}} \right\|_2^2 \le e^{-\frac{p_{\min}\cdot \Delta}{2}\cdot (1-\lambda(\P))\cdot t}} \ge 1-e^{-\frac{p_{\min}\cdot \Delta}{2}\cdot (1-\lambda(\P))\cdot t},
\]
\end{lemma}
We remark that in contrast to \cite[Corollary 2.7]{DBLP:conf/focs/SauerwaldS12}, we have dropped the constraint that $p_{\min}=\Omega(\frac{1}{\Delta})$ as it is easy to see that the proof works without this constraint.
\begin{theorem}[cf.~{\cite[Theorem 2.9]{DBLP:conf/focs/SauerwaldS12}}]\label{tauCount:RM}
Let $G$ be any graph with maximum degree $\Delta$ and consider the random matching model. Then with probability at least $1-n^{-1}$ the sequence of matchings $\left(\M^{(s)}\right)_{s=1}^t$ is $(K,1/(2n))$-smoothing for
\[
t:= \frac{8}{\Delta\cdot p_{\min}} \cdot \frac{1}{1-\lambda(\P)}\cdot \log\left(4Kn^2\right).
\]

\end{theorem}

Next observation bounds the tail of the lower gap by that of upper gap.
\begin{observation}[cf.~{\cite[Observation 2.12]{DBLP:conf/focs/SauerwaldS12}}]
\label{obs:upLowDiscRelation} 
Assume
$\discr(x^{(0)})=K$.
Fix a sequence of matchings $\left(\M^{(s)}\right)_{s=1}^{\infty}$.
Then for arbitrary positive integers $\alpha$ and $t$ it holds that
\[\max_{\substack{y\in \Z^n:\\ \discr(y)\le K}}\left\{\Pro{X_{\min}^{(t)} \le \lfloor \overline{x} \rfloor - \alpha ~\Big|~ x^{(0)}=y}\right\} \le \max_{\substack{y\in \Z^n:\\ \discr(y)\le K}}\left\{\Pro{X_{\max}^{(t)} \ge \lfloor \overline{x} \rfloor + \alpha ~\Big|~ x^{(0)}=y}\right\}. \]
\end{observation}

Occasionally, we may add/subtract the same number of tokens to/from each node.
\begin{observation}[cf.~{\cite[Observation 2.11]{DBLP:conf/focs/SauerwaldS12}}]\label{obs:subAdd}
Fix a sequence of matchings $\left(\M^{(s)}\right)_{s=1}^{\infty}$.
Consider two executions of the discrete load balancing protocol with the same matchings and the same random choices for the excess tokens but with different initial load vectors $x^{(0)}$ and $\widetilde{x}^{(0)}$. Then the following two statements hold:
\begin{enumerate}
\item If $\widetilde{x}^{(0)}=x^{(0)}+\alpha\cdot \1$ for some $\alpha\in \Z$, then $\widetilde{X}^{(t)}=X^{(t)} + \alpha\cdot \1$ for all $t\ge 1$.
\item If $x_u^{(0)} \le \widetilde{x}_u^{(0)}$ for all $u\in V$, then $X_u^{(t)} \le \widetilde{X}_u^{(t)}$ for all $u\in V$ and $t\ge 1$.
\end{enumerate}
\end{observation}

The following observation shows an important relationship between \cref{def:taus,def:smoothing} 

\begin{observation}[cf.~{\cite[Lemma~2.2]{DBLP:conf/focs/SauerwaldS12}}]
\label{obs:tauGtauCRelation}
Assume the sequence of matchings $\left(\M^{(s)}\right)_{s=1}^t$ satisfies for all
$u\in V$ that $\left\|\M_{u,.}^{[1,t]}-\vec{ \frac{1}{n} }\right\|_2^2 \le {\left(\frac{\epsilon}{2 K\cdot n}\right)^2}$.
Then this sequence is $(K,\epsilon)$-smoothing.
\end{observation}
We will omit the proof, since the result follows immediately 
by noting that $ \max_{v \in V} \left| \M_{u,v}^{[1,t]} - \frac{1}{n} \right| \leq \left\| \M_{u,.}^{[1,t]} - \vec{ \frac{1}{n} } \right\|_2 $ and then applying the third statement of \cite[Lemma 2.2]{DBLP:conf/focs/SauerwaldS12}.

\subsection{Other Tools}

In this section we collect a few other tools and results that are frequently used in our analysis.

\begin{observation}\label{obs:subAddAppl}
Fix a sequence of matchings $\left(\M^{(s)}\right)_{s=1}^{\infty}$.
Consider two executions of the discrete load balancing protocol with the same matchings and the same random choices for the excess tokens' movements, but with different initial load vectors $x^{(0)}$ and $\widetilde{x}^{(0)}$ such that for some $\alpha\in \N$ and for all $u\in V$ it holds
\[
\widetilde{x}_u^{(0)} := \max\left\{ x_u^{(0)}-\alpha,0 \right\}.
\]
Then for any round $t\ge 0$ it holds
\[
X_u^{(t)} \le \widetilde{X}_u^{(t)} + \alpha.
\]
It then directly follows that the height of any token in load vector $X^{(t)}$ is at most the maximum height of any token in load vector $\widetilde{X}^{(t)}$ plus $\alpha$.
\end{observation}
\begin{proof}
We consider two auxiliary executions of discrete load balancing with initial load vectors $\widehat{x}^{(0)}$ and $\widetilde{x}^{(0)}$ in which for each node $u\in V$ we have
\[\widehat{x}^{(0)}_u:= \max\{x_u^{(0)},\alpha\} \quad\quad \text{and} \quad\quad \widetilde{x}^{(0)}_u:= \widehat{x}^{(0)}_u - \alpha.\]
We run the discrete load balancing for these executions with the same matching $\M^{(t)}$ and the same random choices for the excess tokens in each round $t\in \N$ but with different initial load vectors constructed as above.
From the second statement of \cref{obs:subAdd}, it follows that for any round $t\ge 0$ and each node $u\in V$ we have $X_u^{(t)} \le \widehat{X}_u^{(t)}$. Similarly, from the first statement of \cref{obs:subAdd} it follows that
$\widetilde{X}_u^{(t)} = \widehat{X}_u^{(t)}-\alpha$. A combination of these two implies that
\begin{equation*}\label{eq:diffExecutions}
X_u^{(t)} \le \widetilde{X}_u^{(t)} + \alpha,
\end{equation*}
for any node $u\in V$ and any round $t\ge 0$.
\end{proof}

The next observation uses the concept of orientation, which is defined in \cref{eq:disc_cont}.
\begin{observation}\label{obs:flip}
Fix a sequence of matchings $\left(\M^{(s)}\right)_{s=1}^{\infty}$. Consider two executions of the discrete load balancing protocol with the same matchings but with different initial load vectors $x^{(0)}$ and $\widetilde{x}^{(0)}$, where $\widetilde{x}^{(0)}:=K\cdot \vec{ 1} -x^{(0)}$ for some $K \in \N_0$, and flipped orientations, i.e., $\widetilde{\Phi}_{u,v}^{(t)}=-\Phi_{u,v}^{(t)}$. Then for any round $t \geq 1$,
\[
\widetilde{X}^{(t)} = K\cdot \vec{1} -X^{(t)}.
\]
\end{observation}
\begin{proof}
The proof is by induction over $t \in \N_0$. The base case $t=0$ holds by assumption. Now assuming $\widetilde{x}^{(t-1)}= K\cdot\1-X^{(t-1)}$ for some $t-1 \in \N_0$, we will prove that the same equation holds for round $t$. If a node $u$ is not matched in round $t$ then \[\widetilde{X}_u^{(t)}= \widetilde{X}_u^{(t-1)} = K-X_u^{(t-1)} = K-X_u^{(t)}.\]
Now assume the edge $[u:v]$ is in the matching $\M^{(t)}$. 
Then we have
\[
X_u^{(t)} = \frac{X_u^{(t-1)}+X_v^{(t-1)}}{2} + \frac{1}{2}\cdot \Odd\left(X_u^{(t-1)}+X_v^{(t-1)}\right)\cdot \Phi_{u,v}^{(t)}
\]
and
\[
X_v^{(t)} = \frac{X_u^{(t-1)}+X_v^{(t-1)}}{2} + \frac{1}{2}\cdot \Odd\left(X_u^{(t-1)}+X_v^{(t-1)}\right)\cdot \Phi_{v,u}^{(t)}.
\]
Furthermore, we have
\begin{align}
\widetilde{X}_u^{(t)} &= \frac{\widetilde{X}_u^{(t-1)} + \widetilde{X}_v^{(t-1)}}{2} + \frac{1}{2}\cdot \Odd\left(\widetilde{X}_u^{(t-1)} + \widetilde{X}_v^{(t-1)}\right)\cdot \widetilde{\Phi}_{u,v}^{(t)}\notag \\
& =
\frac{2K-\left(X_u^{(t-1)}+X_v^{(t-1)}\right)}{2} + \frac{1}{2} \cdot \Odd\left( 2K- \left(X_u^{(t-1)}+X_v^{(t-1)}\right)\right) \cdot \widetilde{\Phi}_{u,v}^{(t)}\notag\\
& = \frac{2K-\left(X_u^{(t-1)}+X_v^{(t-1)}\right)}{2} + \frac{1}{2} \cdot \Odd\left( X_u^{(t-1)}+X_v^{(t-1)}\right) \cdot \widetilde{\Phi}_{u,v}^{(t)}\notag\\
& = K - \frac{\left(X_u^{(t-1)}+X_v^{(t-1)}\right)}{2} - \frac{1}{2} \cdot \Odd\left(X_u^{(t-1)}+X_v^{(t-1)}\right) \cdot \Phi_{u,v}^{(t)}\notag\\
&= K-X_u^{(t)}.\notag
\end{align}
The case for $v$ is analogous to the case for $u$. Hence for each node $u\in V$, we get $\widetilde{X}_u^{(t)}=K-X_u^{(t)}$ and consequently, $\widetilde{X}^{(t)}=K\cdot \vec{ 1} -X^{(t)}$. 
This finishes the induction step and completes the proof.
\end{proof}

\begin{observation}\label{obs:secNormExpansion}
Let $\M$ be an $n\times n$ doubly stochastic matrix. Then for any node $u\in V$,
\[
\left\|\M_{u,.} -\vec{ \frac{1}{n} } \right\|_2^2 = \left\| \M_{u,.}\right\|_2^2 - \frac{1}{n}.
\]
\end{observation}
\begin{proof} We calculate
\begin{align*}
\left\|\M_{u,.} -\1\cdot \frac{1}{n}\right\|_2^2 &= \sum_{v\in V} \left(\M_{u,v}-\frac{1}{n}\right)^2 = \sum_{v\in V} \left((\M_{u,v})^2 -\frac{2}{n}\cdot \M_{u,v} - \frac{1}{n^2}\right)
\\&= \sum_{v\in V} (\M_{u,v})^2 - \frac{2}{n}\cdot \sum_{v\in V} \M_{u,v} + \sum_{v\in V} \frac{1}{n^2} = \left\| \M_{u,.}\right\|_2^2 - \frac{2}{n} + \frac{1}{n}. \qed
\end{align*}
\end{proof}

\begin{observation}\label{obs:inter:secondNormBound}
Let $\M$ be any doubly stochastic matrix and $\left(a_k\right)_{k\in V}$ be any stochastic vector. Then
\[\sum_{w\in V} \left(\sum_{k \in V} a_k\cdot \M_{w,k} \right)^2 = \sum_{w\in V} \left(\sum_{k \in V} a_k\cdot \M_{w,k} -\frac{1}{n} \right)^2 + \frac{1}{n}.\]
\end{observation}
\begin{proof}
\begin{align*}
\sum_{w\in V}\left(\sum_{k\in V} a_k\cdot \M_{w,k}\right)^2 &= \sum_{w\in V}\left(\sum_{k\in V} \left(a_k-\frac{1}{n}\right)\cdot \M_{w,k}+\frac{1}{n}\right)^2
\\
\\&= \sum_{w\in V} \left(\sum_{k\in V} \left(a_k-\frac{1}{n}\right)\cdot \M_{w,k} \right)^2 +
\frac{2}{n}\cdot \sum_{w\in V}\sum_{k\in V} \left(a_k-\frac{1}{n}\right)\cdot \M_{w,k} + \sum_{w\in V} \frac{1}{n^2}
\\& =\sum_{w\in V} \left(\sum_{k\in V} \left(a_k-\frac{1}{n}\right)\cdot \M_{w,k} \right)^2 +
\frac{2}{n}\cdot  \sum_{k\in V} \left(a_k-\frac{1}{n}\right)\cdot \sum_{w\in V} \M_{w,k} + \frac{1}{n}
\\
& \stackrel{(a)}{=} \sum_{w\in V} \left( \sum_{k\in V} a_k\cdot \M_{w,k} -\frac{1}{n}\right)^2 + \frac{1}{n},
\end{align*}
where $(a)$ holds since $\M_{w,.}$ is stochastic, from $\sum_{k\in V} a_k=1$ and $|V|=n$ we get $\sum_{k\in V} \left(a_k-\frac{1}{n}\right)=0$.
\end{proof}
The next observation is a well-known statement and in fact, it is implied by Equation \cref{eq:psi:non-increasing}.
\begin{observation} \label{monotone}
For each node $w\in V$ the expression $\left\|\M^{[1,t]}_{w,.}-\vec{ \frac{1}{n}}\right\|_2^2$ is non-increasing over~$t$. 
\end{observation}
\end{document}